\documentclass[12pt]{article}
\usepackage[left=1in,top=1in,right=1in,bottom=1in,head=.1in,nofoot]{geometry}

\usepackage{setspace,url,bm,amsmath} %

\usepackage[small]{titlesec}
\titlelabel{\thetitle.\quad}
\titleformat*{\section}{\bf\large\center}
\titlespacing*{\section}
{0pt}{0ex}{0ex}
\titlespacing*{\subsection}
{0pt}{0ex}{0ex}
\titlespacing*{\subsubsection}
{0pt}{0ex}{0ex}

\usepackage{graphicx} %
\usepackage{bbm}
\usepackage{latexsym}
\usepackage{caption}
\usepackage[margin=20pt]{subcaption}
\usepackage{titletoc}
\usepackage{hyperref}
\usepackage{booktabs}
\usepackage{multirow}

\usepackage{enumitem}

\usepackage[table]{xcolor}

\newcommand{\GG}[1]{}

\usepackage{amsthm}
\usepackage{amsmath}

\usepackage{newtxtext,newtxmath}

\usepackage{color}

\usepackage{comment}
\theoremstyle{definition}
\newtheorem{assumption}{Assumption}
\newtheorem*{theorem*}{Theorem}
\newtheorem{theorem}{Theorem}
\newtheorem*{rmk*}{Remark}

\newtheorem{proposition}{Proposition}
\newtheorem{lemma}{Lemma}

\newtheorem{remark}{Remark}

\newtheorem*{corollary*}{Corollary}

\usepackage{natbib} %
\bibpunct{(}{)}{;}{a}{}{,} %

\usepackage{etoolbox} %
\apptocmd{\sloppy}{\hbadness 10000\relax}{}{} %

\usepackage{color}
\usepackage{listings}

\def\ind{\begin{picture}(9,8)
         \put(0,0){\line(1,0){9}}
         \put(3,0){\line(0,1){8}}
         \put(6,0){\line(0,1){8}}
         \end{picture}
        }

\def\Pr{\mathbb{P}}

\def\converged{\stackrel{d}{\longrightarrow}}

\def\I{\mathbbm{1}}

\allowdisplaybreaks

\usepackage{textcomp}
\usepackage{xr}
\usepackage[textsize=tiny, textwidth = 2cm, shadow]{todonotes}

\usepackage{soul}

\def\liminf{\underline{\lim}}

\def\Ker{\text{Ker}}

\def\P{\mathbb{P}}

\def\COR{\textup{COR}}
\def\CRR{\textup{CRR}}

\def\OR{\textup{OR}}

\def\I{\mathbbm{1}}
\def\ind{\begin{picture}(9,8)
         \put(0,0){\line(1,0){9}}
         \put(3,0){\line(0,1){8}}
         \put(6,0){\line(0,1){8}}
         \end{picture}
        }

\def\rev{\color{black}}

\newcommand\redsout{\bgroup\markoverwith{\textcolor{red}{\rule[0.5ex]{2pt}{0.4pt}}}\ULon}

\allowdisplaybreaks

\begin{document}

\def\spacingset#1{\renewcommand{\baselinestretch}%
{#1}\small\normalsize} \spacingset{1}

\title{\bf 
\LARGE
Sensitivity Analysis for the Test-Negative Design
}
\author{
	Soumyabrata Kundu, Peng Ding, Jingshu Wang and Xinran Li
\footnote{
    Soumyabrata Kundu is a Ph.D. graduate, Department of Statistics, University of Chicago, Chicago, IL 60637 (E-mail: \href{mailto:soumyabratakundu@uchicago.edu}{soumyabratakundu@uchicago.edu}). 
    Peng Ding is Associate Professor, Department of Statistics, University of California, Berkeley, CA 94720 (E-mail: \href{mailto:pengdingpku@berkeley.edu}{pengdingpku@berkeley.edu}). 
    Jingshu Wang and Xinran Li are Assistant Professors, Department of Statistics, University of Chicago, Chicago, IL 60637 (E-mail: 
    \href{mailto:jingshuw@uchicago.edu}{jingshuw@uchicago.edu}
    and \href{mailto:xinranli@uchicago.edu}{xinranli@uchicago.edu}).\newline
    P.D. is partly supported by the National Science Foundation (DMS-2514234).
    J.W. is partly supported by the National Science Foundation (DMS-2113646 and DMS-2238656)
    and 
    National Institute of General Medical Sciences (R35GM162500).
    X.L. is partly supported by the National Science Foundation
(DMS-2400961).
}
}
\date{}
\maketitle

\begin{abstract}
\singlespacing
The test-negative design has become popular for evaluating the effectiveness of post-licensure vaccines using observational data. 
In addition to its logistical convenience on data collection, 
the design is also believed to control for the differential health-care-seeking behavior between vaccinated and unvaccinated individuals, 
an important while often unmeasured confounder between the vaccination and infection.
Hence, the design has been employed routinely to   monitor seasonal flu vaccines and more recently to measure the COVID-19 vaccine effectiveness.  
Despite its popularity, the design has been questioned, in particular about its ability to fully control for the unmeasured confounding. 
In this paper, we explore deviations from a perfect test-negative design, and propose various sensitivity analysis methods for estimating the effect of vaccination measured by the causal odds ratio on the subpopulation of individuals with good health-care-seeking behavior.  
We start with point identification of the causal odds ratio under a test-negative design, 
comparing different forms of identification assumptions and their corresponding estimands.
We then propose
two approaches for conducting sensitivity analysis, addressing the influence of the unmeasured confounding in two different ways. 
Specifically, one approach investigates partial control for unmeasured confounding in the test-negative design, while the other examines the impact of unmeasured confounding on both vaccination and infection. 
Furthermore, we combine these approaches to provide narrower bounds on the true causal odds ratio, 
and further 
sharpen the bounds by restricting
the treatment effect heterogeneity. 
Finally, we apply the proposed methods to evaluate the effectiveness of COVID-19 vaccines using observational data from test-negative designs. 
\end{abstract}

\noindent%

\doublespacing

\section{Introduction}
Large-scale randomized, placebo-controlled phase 3 trials have confirmed the efficacy of currently authorized COVID-19 vaccines \citep{polack2020safety, Baden2021}. However, with the continuous virus evolution and the development of new vaccines, randomized experiments may not be available to evaluate the vaccine effectiveness in a timely manner.
By contrast, observational studies can provide invaluable sources for understanding the real-world vaccine effectiveness. 
In particular, 
the test-negative design has emerged as an
efficient 
and practical 
approach
to
assess vaccine effectiveness compared with the traditional cohort or case-control designs, 
and 
has been used extensively during the COVID-19 pandemic \citep{Dean2021, thompson2021effectiveness, Abu21} and even before the pandemic to estimate vaccine effectiveness against seasonal influenza \citep{Jackson2013}. 
In general, 
the test-negative design recruits units that have some illness symptoms, attend some healthcare facility and test for a particular disease. %
The cases are then those who test positive, and the non-cases or controls are those who test negative. 
Finally, the vaccine efficacy is estimated by $1$ minus the odds ratio of testing positive between vaccinated and unvaccinated groups.

Compared with traditional designs, the test-negative design is often logistically simpler, applicable to large electronic health records, and thus more cost-effective, making it increasingly popular in measuring and monitoring vaccine efficacy \citep{sullivan2014potential, Dean2021}. 
While recognizing its advantages, it is crucial to understand whether the test-negative design yields unbiased estimation for vaccine efficacy and, if it does, in what sense.
\citet{Jackson2013} studied the rationale of the test-negative design and showed its validity  under certain assumptions. 
Importantly, they pointed out that the test-negative design
could help 
to remove confounding due to the unmeasured differential health-care-seeking behavior between vaccinated and unvaccinated individuals, which is a key advantage over the traditional designs. 
However, the assumptions justifying the validity of the test-negative design are likely to fail; see, e.g.,  \citet{sullivan2016theoretical}, \citet{shi2023current} and \citet{ortiz-brizuela_potential_2025} for illustrations using causal diagrams. 
Recently, 
the test-negative design has received growing interest in statistics. 
For example, 
\cite{wang2023randomization} studied cluster-randomized test-negative designs, and 
\citet{Tchetgen23} and \citet{yu2023test} studied test-negative designs with additional information such as negative controls and various reasons for testing.

In this paper, we propose methods for sensitivity analysis under the test-negative design. 
Specifically, we start with an ideal scenario in which the test-negative design is able to completely control for the unmeasured confounding such as health-care-seeking behavior, 
and show  
that the observed odds ratio regarding the vaccination's effect on infection from a test-negative design can be consistent for the 
corresponding true
causal odds ratio among individuals who will seek care when ill. 
Such an identification of causal odds ratio has been demonstrated in 
\citet{ciocuanea2021adjustment} and is closely related to the identification of causal odds ratio under outcome-dependent sampling \citep{Didelez2010}, which includes both the test-negative designs and the classical case-control studies as special cases. 
We also discuss and compare alternative identification strategies that rely on different identification assumptions and target distinct causal estimands.

We will then conduct two forms of sensitivity analysis,
and investigate how sensitive the observed odds ratio will be for estimating the causal odds ratio.  
One allows the test-negative design to only partially remove the unmeasured confounding, 
whereas the other investigates how the strength of unmeasured confounder affects the inference of the causal odds ratio.
More interestingly, 
we can combine
these two sensitivity analysis approaches to provide narrower bounds on the true causal odds ratio. 
Additional 
restrictions on
the treatment effect heterogeneity across different levels of unmeasured confounding can also be incorporated into our sensitivity analysis to sharpen the inference. 

Our sensitivity analysis is 
closely related to \cite{Lee2015}, \cite{smith2019bounding}, and \cite{gabriel2022causal}, which focused on selection bias and outcome-dependent sampling.
This is because in a test-negative design we can only observe individuals who attend some healthcare facility for testing.
Different from the existing literature, our sensitivity analysis methods are tailored for the test-negative design, in terms of both the causal estimand and the constraints on the unmeasured confounding.

The paper proceeds as follows. Section \ref{sec:framework} introduces the test-negative design under the potential outcome framework. 
Section \ref{sec:sen_partial} studies identification and sensitivity analysis for the test-negative design with fully or partially controlled unmeasured confounding. 
Section \ref{sec:sen_conf_strength} investigates the design's sensitivity to the strength of unmeasured confounding. 
Section \ref{sec:integrate} combines the two sensitivity analysis methods in Sections \ref{sec:sen_partial} and \ref{sec:sen_conf_strength} to obtain narrower bounds on the causal parameter. 
Section \ref{sec:categorical} extends our methods to scenarios with categorical exposure, outcome and a general unmeasured confounder.  
Section \ref{sec:sen_effect_heter} considers further constraints on the treatment effect heterogeneity. 
Section \ref{sec:prac} discusses some practical implementation issues, such as sensitivity parameter specification and confidence bounds construction. 
Section \ref{sec:app} applies our methods to evaluate the effectiveness of COVID-19 vaccines.
Section \ref{sec:conclusion} concludes with discussion. 
The supplementary material contains the proofs of all theorems, additional numerical results, and additional technical details.

\section{Framework, Notation, and Assumption}\label{sec:framework}

\subsection{Exposure, potential outcomes, and causal effects}\label{sec:potential}

We start with the setting 
of
a binary exposure indicating, e.g., whether an individual takes the vaccine or not, and a binary outcome indicating, e.g., whether an individual is infected or not, 
and focus on evaluating the causal effect of the exposure on the outcome. 
Let $Z\in \{0,1\}$ denote the binary exposure, 
$Y\in \{0,1\}$ denote the binary outcome, and $C$ denote the observed 
covariates.
{\rev 
For example, $Y$ may denote the infection status when an individual is tested for a particular disease, $Z$ may denote the vaccination status prior to the potential infection, and $C$ may denote some pretreatment covariates measured or determined before the vaccination decision, such as age and sex.} 
We further introduce a binary $U\in \{0,1\}$ to denote the unmeasured confounder that can affect both the exposure and outcome. 
For example,
$U$ is the unmeasured health-care-seeking behavior that can affect both vaccination and infection. 
In Section \ref{sec:categorical},
we will extend our analysis to include categorical exposures that can correspond to various types and levels of vaccination, 
categorical outcomes that can correspond to different levels of symptom severity in addition to infection, 
and general unmeasured confounders that can correspond to different levels of health-care-seeking behavior.

Following the potential outcome framework, 
under the stable unit treatment value assumption \citep[SUTVA;][]{Rubin:1980},
let $Y(1)$ and $Y(0)$ be the potential outcomes of an individual with and without the exposure, respectively. 
Note that in infectious disease studies, SUTVA may be violated due to interference \citep{hudgens2008toward, schnitzer_estimands_2022}. For simplicity, we maintain this assumption here; extending the analysis to allow interference is challenging both theoretically and in terms of data availability, and we leave this for future work.
We are interested in the causal effectiveness of the exposure on the outcome, measured by $1$ minus the causal odds ratio \citep{world2020design}
within stratum defined by the levels of 
observed covariates and unmeasured confounder: 
\begin{align}\label{eq:COR}
    \COR_{uc} \equiv \frac{\P(Y(1)=1\mid U = u, C = c)/\P(Y(1) = 0 \mid U = u, C=c)}{\P(Y(0)=1\mid U = u, C = c)/\P(Y(0) = 0 \mid U = u, C=c)}. 
\end{align}
We will focus on $\COR_{uc}$ at a particular level of $U$ that is controlled, or at least partially controlled, by the test-negative design 
and any given level of the observed covariates $C$.
Scientifically, it is also
interesting to infer the causal odds ratio for units with the other unmeasured confounding level or both unmeasured confounding levels. However, it may be challenging due to either the nature of the test-negative design or the noncollapsibility of the odds ratio \citep{Robins99, Geng1995}; see related discussions at the end of Section \ref{sec:sen_strength} and the beginning of Section \ref{sec:sen_effect_heter}.

\begin{remark}
Sometimes, 
the vaccine efficacy is defined as 1 minus the causal risk ratio. 
We mainly focus on the causal odds ratio under the test-negative design 
because, under outcome-dependent sampling, the causal odds ratio is the only identifiable parameter without further assumptions \citep{Didelez2010}. 
The causal odds ratio approximates the causal risk ratio for rare outcomes, but they can be different in general.
The existing literature has also considered inferring the causal risk ratio under the test-negative design 
\citep{Jackson2013, yu2023test}. 
We will discuss the difference between the identification assumptions for these two causal estimands in detail in Section 
\ref{sec:sen_partial}, {\rev as well as a connection between them in Section \ref{sec:cat_connect}}. 
\end{remark}

\begin{remark}
{\rev 
Under additional assumptions that may be plausible when $Y$ is a properly defined categorical outcome with more than two levels, the causal estimand in \eqref{eq:COR} reduces to a causal risk ratio; see Section \ref{sec:cat_connect} for details.}
\end{remark}

\subsection{Test-negative design}\label{sec:tnd}

Figure \ref{fig:dag} illustrates the test-negative design. 
In Figure~\ref{fig:dag}, the binary $T$ indicates whether 
or not an individual is 
included in 
the test-negative design. 
As shown in Figure \ref{fig:dag}, 
$T$ 
can
depend on the outcome, the observed covariates and the unobserved confounder.
For example, in COVID-19 studies, the test-negative design may include only individuals who have the disease symptoms and seek for a RT-PCR test, under which the selection of study units $T$ could depend on both the infection status $Y$ and the unmeasured health-care-seeking behavior $U$.
From the observed samples with $T=1$, we can only identify 
\begin{align}\label{eq:o_zy}
    \pi_{zy\mid c} = \P( Z = z, Y= y \mid C = c, T = 1), \quad (\forall z, y, c), 
\end{align}
which represents a \(2\times2\) contingency table of vaccination and infection statuses for each observed covariate value $c$.

\begin{figure}[htb]
    \centering
    \includegraphics[width=0.33\linewidth]{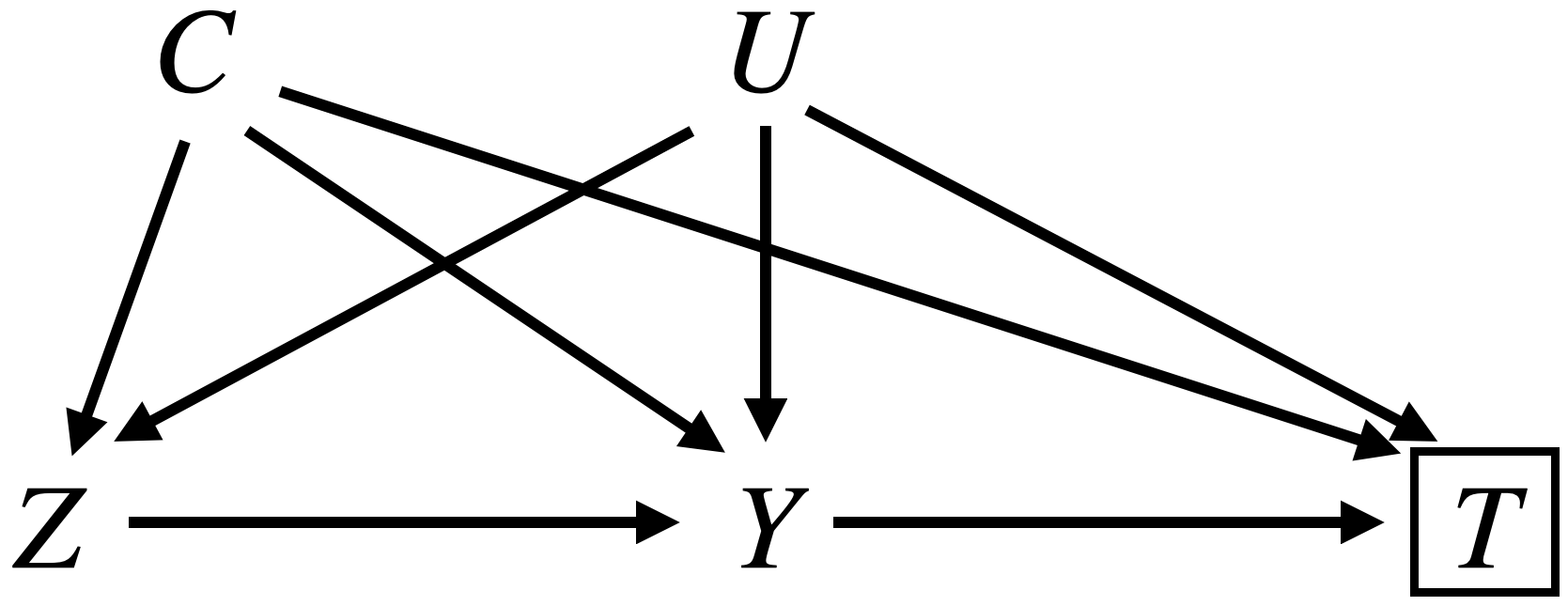}
    \caption{
    A causal diagram for the test-negative design. 
    $Z$ denotes the exposure, $Y$ denotes the outcome, $C$ denotes the observed covariates, $U$ denotes the unmeasured confounder, 
    and $T$ denotes the inclusion indicator. 
    Similar to \citet{pearl12},
    we put a box around $T$ since the design involves selection on $T=1$. 
    In particular, the design may include only units who test for a certain disease.
    }
    \label{fig:dag}
\end{figure}

A key assumption depicted in Figure~\ref{fig:dag} is that
the 
inclusion indicator
$T$ and the exposure $Z$ are conditionally independent
given the outcome $Y$, covariates $C$, and unmeasured confounder $U$. 
Additionally, we assume that, given the covariates $C$ and unmeasured confounder $U$, the exposure $Z$ is conditionally independent of the potential outcomes $Y(z)$, $z=0,1$. 
We emphasize that the causal diagram in Figure \ref{fig:dag} is mainly for illustration purpose, and the assumptions are formally stated below in terms of potential outcomes and conditional independence. 

\begin{assumption}\label{asmp:U}
(i) $Z \ind Y(z) \mid U, C$ for all exposure level $z$; \ 
(ii) $Z \ind T \mid Y, U, C$.

\end{assumption}
Assumption \ref{asmp:U}(i) is frequently invoked in a sensitivity analysis \citep{Rubin1983, ding2016sensitivity}. 
Assumption \ref{asmp:U}(ii) indicates that 
the exposure has no direct effect on 
the inclusion indicator, 
except for the indirect effect mediated through the outcome $Y$. 
In practice, we may want to include more information in $Y$ to block the direct causal path from $Z$ to $T$; see Sections \ref{sec:cat_connect} and \ref{sec:categorical}. 
{\rev 
Assumption \ref{asmp:U}(ii) can also be written in terms of potential outcomes: for any possible values $y, u$ and $c$ of the outcome, unmeasured confounder and observed covariates, 
\begin{align*}
        \P(T=1\mid Y(1) = y, Z=1, U=u, C=c) = \P(T=1\mid Y(0) = y, Z=0, U=u, C=c).
\end{align*}
This states that treated units with potential outcome $Y(1)=y$ and control units with potential outcome $Y(0)=y$ have the same probability of being included in the test-negative design, conditional on having the same values of the unmeasured confounder and observed covariates.}
Assumption \ref{asmp:U} is similar to the assumptions invoked for identifying causal effects under outcome-dependent sampling \citep[][Theorem 7]{Didelez2010}, except that it allows for an unmeasured confounder $U$.

Throughout the paper, our analysis for the test-negative design will always condition on the observed covariates $C=c$. 
We will make this conditioning explicit for clarity.

\section{Identification and sensitivity analysis under fully or partially controlled unmeasured confounding}\label{sec:sen_partial}

\subsection{Identification 
under a perfect test-negative design}\label{sec:perfect_test}

A key advantage of the test-negative design, compared with traditional designs, is that it can help 
reduce unmeasured confounding
between the exposure and outcome.
We illustrate this advantage using the 
example in \citet{Jackson2013}. 
Suppose we are interested in assessing the influenza vaccine effectiveness. 
Health-care-seeking behavior is an important but often unmeasured confounder that influences both an individual's inclination to receive a vaccine and their susceptibility to influenza infection.
The test-negative design collects only units that go to hospital for testing after they have 
influenza-like symptoms. 
These collected units generally 
have positive or good health-care-seeking behavior.  
If 
units in the test-negative design could all have the same level of the unmeasured confounder, 
then 
the unmeasured confounder is implicitly controlled for. 
Below we formally discuss the identification of causal effects under such a \emph{perfect} test-negative design. 

We consider a specific level $c$ of the baseline covariates.
Without loss of generality, we use $0$ to denote the level of the unmeasured confounder that the test-negative design is believed to control for.  
We then invoke the following assumption that the test-negative design fully controls for the unmeasured confounder. 

\begin{assumption}[A perfect test-negative design]\label{asmp:U_T}
$\P( U = 0 \mid T = 1, C = c ) = 1$. 
\end{assumption}

Under Assumptions \ref{asmp:U} and \ref{asmp:U_T}, 
we can identify the causal odds ratio for units with unmeasured confounder $U=0$ and baseline covariates $C=c$, using the observed odds ratio between the exposure and outcome among tested units with baseline covariates $C=c$: 
\begin{align}\label{eq:OR}
    \OR_c
    & = \frac{\P(Y=1\mid Z=1, T = 1, C = c)/\P(Y = 0 \mid Z = 1, T=1, C=c)}{\P(Y=1\mid Z=0, T=1,  C = c)/\P(Y = 0 \mid Z=0, T= 1, C=c)}
    \nonumber
    \\
    & = 
    \frac{\pi_{11\mid c} \pi_{00\mid c}}{
    \pi_{10\mid c} \pi_{01\mid c}}. 
\end{align}

\begin{proposition}\label{prop:cor_perfect_test}
Under Assumptions \ref{asmp:U} and \ref{asmp:U_T}, 
we have 
$\COR_{0c} = \OR_c$. 
\end{proposition}

The identification of causal odds ratio in Proposition \ref{prop:cor_perfect_test} has been shown by \citet{ciocuanea2021adjustment}, 
which follows from 
\citet[][Theorem 7]{Didelez2010} once we notice that the test-negative design 
contains only units with
the unmeasured confounder $U=0$ under Assumption \ref{asmp:U_T}. 
Below we include a proof of Proposition \ref{prop:cor_perfect_test}, since it is elementary and
helpful for understanding our later discussion. 
In addition, it also shows the crucial role of  Assumptions \ref{asmp:U} and \ref{asmp:U_T}, as well as the {\it invariance} property of the odds ratio. 
The invariance means that, for any two binary random variables $A$ and $B$, the ratio between the odds of $A=1$ given $B=1$ and that of $A=1$ given $B=0$ will be invariant when we switch the roles of $A$ and $B$.

\begin{proof}[Proof of Proposition \ref{prop:cor_perfect_test}]
By Assumptions \ref{asmp:U} and \ref{asmp:U_T}, 
we have, 
\begin{align*}%
    \COR_{0c} & = \frac{\P(Y=1\mid Z=1, U = 0, C = c)/\P(Y = 0 \mid Z = 1, U = 0, C=c)}{\P(Y=1\mid Z=0, U = 0, C = c)/\P(Y = 0 \mid Z=0, U = 0, C=c)} 
    \nonumber
    \\
    & \quad \ \text{(\textit{by Assumption \ref{asmp:U}(i)})}
    \nonumber
    \\
    & = 
    \frac{\P(Z=1\mid Y=1, U = 0, C = c)/\P(Z = 0 \mid Y = 1, U = 0, C=c)}{\P(Z=1\mid Y=0, U = 0, C = c)/\P(Z = 0 \mid Y=0, U = 0, C=c)} 
    \nonumber
    \\
    & \quad \ \text{(\textit{by invariance of the odds ratio})}
    \nonumber
    \\
    & = 
    \frac{\P(Z=1\mid Y=1, T = 1, U = 0, C = c)/\P(Z = 0 \mid Y = 1, T=1, U = 0, C=c)}{\P(Z=1\mid Y=0, T=1, U = 0, C = c)/\P(Z = 0 \mid Y=0, T= 1, U = 0, C=c)} 
    \nonumber
    \\
    & \quad \ \text{(\textit{by Assumption \ref{asmp:U}(ii)})}
    \nonumber
    \\
    & = 
    \frac{\P(Z=1\mid Y=1, T = 1, C = c)/\P(Z = 0 \mid Y = 1, T=1, C=c)}{\P(Z=1\mid Y=0, T=1,  C = c)/\P(Z = 0 \mid Y=0, T= 1, C=c)}
    \nonumber
    \\
    & \quad \ \text{(\textit{by Assumption \ref{asmp:U_T}})}
    \nonumber
    \\
    & = \OR_c.
    \\
    & \quad \ \text{(\textit{by invariance of the odds ratio})}
    \nonumber,
\end{align*}
where the reason for each equality is explained in the parentheses. 
\end{proof}

Proposition \ref{prop:cor_perfect_test} implies that, when the test-negative design can fully control for the unmeasured confounding between the exposure and outcome, 
we can perform usual analyses for the exposure-outcome odds ratio using the observed data, taking into account the measured baseline covariates.

\subsection{Comparison with alternative identification strategies for the causal risk ratio}\label{sec:comp_ident}

In contrast to Proposition \ref{prop:cor_perfect_test} that focuses on the identification of causal odds ratio, 
the existing literature on test-negative designs has focused primarily on the causal risk ratio:
\begin{align}\label{eq:CRR}
	 \CRR_{0,c} &  \equiv \frac{\P(Y(1)=1 \mid U=0, C=c)}{\P(Y(0)=1 \mid U=0, C=c)} . 
\end{align}
Below we discuss alternative identification strategies for causal risk ratio
 in 
\citet{Jackson2013} and \citet{yu2023test}. 
We will formalize their identification using the potential outcome framework and discuss the difference between the identification assumptions for causal odds ratios and risk ratios.
Because both strategies invoke similar assumptions, we focus on the identification in \citet{Jackson2013}, and relegate the detailed discussion about \citet{yu2023test} to the supplementary material. 
The identification of the causal risk ratio in \citet{Jackson2013} 
also requires Assumptions \ref{asmp:U}(i) and \ref{asmp:U_T},
that is, 
$U$ is the unmeasured confounder between the exposure and the outcome, and the test-negative design contains only units with positive health-care-seeking behavior (denoted by $U=0$). 
\citet{Jackson2013} made the following two additional assumptions, either implicitly or explicitly.

\begin{assumption}\label{asmp:jackson_all}
\begin{enumerate}[label={(\roman*)}, topsep=1ex,itemsep=-0.3ex,partopsep=1ex,parsep=1ex]
    \item $\Pr( T = 1\mid Z= z, Y = 1, U=0, C = c ) = 1$ for $z=0,1$.
    \item $\Pr( Y=0, T=1\mid Z=1, U=0, C = c ) = \Pr( Y=0, T=1\mid Z=0, U=0, C = c )$. 
\end{enumerate}
\end{assumption}

We first give some intuition for Assumption \ref{asmp:jackson_all}(i). 
Assumption \ref{asmp:jackson_all}(i) can be relaxed to 
\begin{align}\label{eq:weaker_Jackson_1}
	Z\ind T \mid Y=1, U=0, C=c
\end{align}
which is part of Assumption \ref{asmp:U}(ii) except that it requires only conditional independence between the exposure and testing among infected units. 
Moreover,
\citet{Jackson2013} focused on influenza vaccine and considered the setting where units with positive health-care-seeking behavior ($U=0$) must seek for testing when they are infected with either influenza or other pathogens, and they will not seek for testing when they are not infected, regardless of their vaccination status. 
Under their setting, 
Assumption \ref{asmp:jackson_all}(i) must hold.

We then give some intuition for Assumption \ref{asmp:jackson_all}(ii). 
Under the setting considered by \citet{Jackson2013}, 
units with positive health-care-seeking behavior ($U=0$)
will seek testing without having an influenza infection ($Y=0$ and $T=1$) if and only if they are infected with other pathogens.
Consequently, 
$\I\{Y=0, T=1\} = \I\{\text{infected with other pathogens}\}$ for units with $U=0$, and 
Assumption \ref{asmp:jackson_all}(ii) reduces to 
\begin{align}\label{eq:jackson_other_pathogens}
   \I\{\text{infected with other pathogens}\} \ind Z \mid U=0, C=c;
\end{align}
that is, given positive health-care-seeking behavior $U=0$ and observed covariates $C=c$, whether a unit is infected with other pathogens is independent of the unit's vaccination status.

The proposition below shows that the causal risk ratio can be identified by the observed odds ratio under the above assumptions.

\begin{proposition}\label{prop:iden_jackson}
    Under Assumptions \ref{asmp:U}(i), \ref{asmp:U_T}, and \ref{asmp:jackson_all}, 
    we have $\CRR_{0,c} = \OR_c$. 
\end{proposition}

{\rev 
\citet{schnitzer_estimands_2022} studied a similar identification strategy in which the outcome is jointly defined by infection, symptoms, and hospitalization, and further proposed an inverse probability weighting approach to estimate the marginal causal risk ratio, averaged over the covariate distribution.
}

We then give a proof of Proposition \ref{prop:iden_jackson} below, since it is elementary and also helpful for understanding the importance of the identification assumptions. 

\begin{proof}[Proof of Proposition \ref{prop:iden_jackson}]
Under Assumptions \ref{asmp:U}(i), \ref{asmp:U_T}, and \ref{asmp:jackson_all}, we have
\begin{align*}
    \OR_c
    & = \frac{\P(Y=1\mid Z=1, T = 1, C=c)/\P(Y = 0 \mid Z = 1, T=1, C=c)}{\P(Y=1\mid Z=0, T=1, C=c)/\P(Y = 0 \mid Z=0, T=1, C=c)}
    \nonumber\\
    & \quad \ \text{(\textit{by definition})}
    \nonumber
    \\
    & = \frac{\P(Y=1\mid Z=1, T = 1, U=0, C=c)/\P(Y = 0 \mid Z = 1, T=1, U=0, C=c)}{\P(Y=1\mid Z=0, T=1, U=0, C=c)/\P(Y = 0 \mid Z=0, T=1, U=0, C=c)}
    \nonumber
    \\
    & \quad \ \text{(\textit{by Assumption \ref{asmp:U_T}})}
    \nonumber
    \\
    & = \frac{\P(Y=1, T=1 \mid Z=1, U=0, C=c)}{\P(Y=1, T=1 \mid Z=0, U=0, C=c)}
    \cdot 
    \frac{\P(Y = 0, T=1 \mid Z=0, U=0, C=c)}{\P(Y = 0, T=1 \mid Z = 1, U=0, C=c)}
    \\
    & \quad \ \text{(\textit{by some algebra})}
    \\
    & =  \frac{\P(Y=1 \mid Z=1, U=0, C=c)}{\P(Y=1 \mid Z=0, U=0, C=c)}
    \cdot 1
    \nonumber
    \\
    & \quad \ \text{(\textit{the first term follows from Assumption \ref{asmp:jackson_all}(i) or the weaker assumption in \eqref{eq:weaker_Jackson_1}, }} \\
    & \quad \ \text{\textit{
    the second term follows from Assumption \ref{asmp:jackson_all}(ii)})}
    \nonumber
    \\
    & = \frac{\P(Y(1)=1 \mid U=0, C=c)}{\P(Y(0)=1 \mid U=0, C=c)} \\
    & \quad \ \text{(\textit{by Assumption \ref{asmp:U}(i)})}
    \nonumber
    \\
    & = \CRR_{0,c}, \\
    & \quad \ \text{(\textit{by definition})}
\end{align*}
where the reason for each equality is explained in the parentheses. 
\end{proof}

Assumption \ref{asmp:jackson_all}(ii) or the intuitive version \eqref{eq:jackson_other_pathogens} is the key for the identification of the causal risk ratio. 
However, it may not hold 
when the vaccination
increases the infection from  viruses other than the target one, due to, say, the virus interference. 
The virus interference refers to the phenomenon 
where natural influenza infection may decrease the risk of noninfluenza respiratory viruses (\citet{Alharbi2022}). 
In this case, vaccinated individuals could face a higher risk of noninfluenza respiratory virus infections than the unvaccinated. 
See \citet{Alharbi2022} for more related references and real-data examples of such phenomena. 
\citet{HABER2015} has also conducted simulation under violations of Assumption \ref{asmp:jackson_all}(ii) due to, say, the virus interference, and showed that the bias of vaccine effectiveness estimate from a test-negative design can be quite substantial. 
In addition, they found that the test-negative design can be more sensitive to violations of Assumption \ref{asmp:jackson_all}(ii) than the traditional case-control design.

By contrast, 
the identification for causal odds ratio in Proposition \ref{prop:cor_perfect_test} does not require Assumption \ref{asmp:jackson_all}(ii) that may fail due to virus interference. Instead, it requires Assumption \ref{asmp:U}(ii) that the choice of testing is independent of exposure for both infected and uninfected units given their observed covariates and unmeasured confounder. 
This assumption can be made more plausible when we include richer information in the outcome, such as both infection and symptoms; see Sections \ref{sec:cat_connect} and \ref{sec:categorical}.

{\rev 
\subsection{Categorical outcome and connections between the two identification strategies}\label{sec:cat_connect}

Because the test-negative design may include only tested individuals with disease-like illness symptoms, 
Assumption \ref{asmp:U}(ii) may be violated when the outcome $Y$ is defined solely as a binary infection status and the exposure has direct effects on symptoms that are not mediated through infection. 
To address this issue, we can redefine the outcome $Y$ as a categorical variable 
that jointly captures the infection status and the illness symptoms. 
To be more specific, let $I$ denote the binary indicator of infection status and $S$ denote the binary indicator of illness symptoms; that is, $I=1$ if and only if the individual is infected, and $S=1$ if and only if the individual has illness symptoms. 
We then define the outcome $Y$ as a four-level categorical variable corresponding to the joint levels of $(I,S)$.
In particular, let $Y=1$ if $I=S=1$ (i.e., the individual is infected and symptomatic), and $Y=0$ if $I=0$ and $S=1$ (i.e., the individual is symptomatic but not infected with the disease of interest); the remaining two levels of $(I,S)$ can be encoded as $Y=2$ and $Y=3$.
Here we consider a test-negative design in which the inclusion indicator $T=1$ if and only if an individual has illness symptoms and seeks testing. Accordingly, units in the design must have outcome $Y \in \{0,1\}$. 

Importantly, with the outcome $Y$ incorporates both infection status and symptoms, the conditional independence between the exposure $Z$ and inclusion indicator $T$ given $(Y, U, C)$ in Assumption \ref{asmp:U}(ii) is more likely to hold. 
Suppose that Assumption \ref{asmp:U} holds with this categorical outcome. We can verify that the derivation in the proof of Proposition \ref{prop:cor_perfect_test} continues to hold. 
Consequently, the observed odds ratio $\OR_c$ defined as in \eqref{eq:OR} still identifies the causal estimand $\COR_{0c}$ defined as in \eqref{eq:COR}; see also related discussion in Section \ref{sec:categorical}.
It is worth pointing out that $\COR_{0c}$ is no longer a conventional odds ratio, because the potential outcomes $Y(1)$ and $Y(0)$ take more than two levels. In contrast, $\OR_c$ remains an odds ratio, as the design includes only units with outcome $Y\in \{0,1\}$.

If we further assume that 
\begin{align}\label{eq:asmp_no_virus_int}
    \P(Y(1) = 0 \mid U = 0, C=c) = \P(Y(0) = 0 \mid U = 0, C=c),
\end{align}
then the causal estimand $\COR_{0c}$ defined as in \eqref{eq:COR} reduces to the causal risk ratio $\CRR_{0c}$ defined as in \eqref{eq:CRR}. 
Recall that the outcome equals $0$ if and only if the individual has illness symptoms but is not infected with the disease of interest, i.e., the symptoms are due to other diseases.
Thus, the assumption in \eqref{eq:asmp_no_virus_int} is analogous to the assumption considered by \citet{Jackson2013} and discussed in Section \ref{sec:comp_ident}. 
In other words, if we are willing to assume that, conditional on the observed covariates and unmeasured confounder, the exposure has no effect on the risk of illness symptoms not caused by the disease of interest, then the causal estimand $\COR_{0c}$ reduces to the causal risk ratio in \eqref{eq:CRR}, which targets the effectiveness of the exposure against symptomatic infection and may be easier to interpret.
Note that the assumption in \eqref{eq:asmp_no_virus_int} is generally not reasonable when the outcome is binary; 
it would imply $\COR_{0c}=1$, which may be inconsistent with the observed data from a test-negative design.

Throughout the paper, we focus first on sensitivity analysis with binary exposure, outcome and unmeasured confounder, and then extend the framework to more general settings. As discussed above, when the outcome has more than two levels and an assumption analogous to \eqref{eq:asmp_no_virus_int} hold, the causal estimand $\COR_{0c}$ in \eqref{eq:COR} that we focus on reduces to the corresponding causal risk ratio $\CRR_{0c}$ in \eqref{eq:CRR}. 

}

\subsection{Sensitivity analysis for an imperfect test-negative design}\label{sec:sen_U_not_0}

As discussed in Section \ref{sec:perfect_test}, 
a perfect test-negative design can fully control for the unmeasured confounding and identify the causal odds ratio at a given level of measured and unmeasured confounders simply by the observed odds ratio. 
The identification
needs 
Assumption \ref{asmp:U_T}, which requires that the unmeasured confounder is constant  among  units in the test-negative design (referred to as tested units for simplicity). 
However, this assumption can be strong and questionable in practice, 
challenging the interpretation of results from test-negative designs.
{\rev 
In addition, as studied in \citet{ortiz-brizuela_potential_2025}, some recent applications of the test-negative design include all tested individuals rather than restricting to those with symptoms; in such settings, individuals may seek testing for reasons other than symptoms---such as exposure to the virus or requirements for travel or work---and health-care-seeking behavior is therefore less likely to be fully controlled.}
In the following, we will consider sensitivity analysis for the violation of Assumption \ref{asmp:U_T} to assess the robustness of results from a test-negative design. 
We further 
assume that all $\pi_{zy \mid c}$s are positive to avoid the complication due to undefined odds ratios.

While Assumption \ref{asmp:U_T} may not be entirely realistic, it is still plausible that the test-negative design can 
help partially control for 
the unmeasured confounder by 
bounding its probability mass at zero. 
This motivates us to consider sensitivity analysis with an upper bound on the probability 
of a nonzero unmeasured confounder in the design.

\begin{assumption}\label{asmp:bound_u_not_0}
    For some $0\le \delta_c \le 1$, we have
    $\P(U\ne 0\mid C=c, T=1) \le \delta_c$. 
\end{assumption}

In Assumption \ref{asmp:bound_u_not_0}, $\delta_c$ is a sensitivity parameter that needs to be specified in data analysis and can depend on the value of observed covariates; Assumption \ref{asmp:U_T} is a special case with $\delta_c = 0$.
We will investigate how the inferred causal conclusions change as $\delta_c$ varies. 
If the conclusions are robust to a wide range of $\delta_c$, then the results from the test-negative design will be more credible. 
We further introduce 
\begin{align}\label{eq:pzy_u}
    p_{zy\mid u c} = \P(Z=z, Y=y \mid U = u, C=c, T=1) \quad (z=0,1; y=0,1; u=0,1)
\end{align}
to 
denote the joint probability of exposure and outcome at given levels of observed covariates and unmeasured confounder among tested units. 
Different from the identifiable probabilities $\pi_{zy\mid c}$s in \eqref{eq:o_zy}, 
the $p_{zy\mid uc}$s in \eqref{eq:pzy_u} condition further on the unmeasured confounder and are thus generally unidentifiable from the observed data without  additional assumptions. 
The theorem below gives the sharp bounds on the true causal odds ratio 
using the observed data distribution from the test-negative design. 
Let $(c)_{+} = \max\{c,0\}$ for $c\in \mathbb{R}$.

\begin{theorem}\label{thm:simple_bound_cor}
Under Assumptions \ref{asmp:U} and \ref{asmp:bound_u_not_0}, 
the causal odds ratio for units with unmeasured confounder zero is bounded by  
\begin{align*}
    \min\left\{ \frac{(\pi_{11\mid c}-\delta_c)_{+}\pi_{00\mid c}}{\pi_{10\mid c} \pi_{01\mid c}}, \  \frac{\pi_{11\mid c}(\pi_{00\mid c}-\delta_c)_{+}}{\pi_{10\mid c} \pi_{01\mid c}} \right\} 
    \le \COR_{0c} \le 
    \max\left\{ \frac{\pi_{11\mid c}\pi_{00\mid c}}{(\pi_{10\mid c}-\delta_c)_{+}\pi_{01\mid c}}, \frac{\pi_{11\mid c}\pi_{00\mid c}}{\pi_{10\mid c}(\pi_{01\mid c}-\delta_c)_{+}} \right\}, 
\end{align*}
where $\pi_{zy\mid c}$s are defined as in \eqref{eq:o_zy}. Equivalently, 
\begin{align*}
    \min\left\{ (1-\delta_c/\pi_{11\mid c})_{+}, \ (1-\delta_c/\pi_{00\mid c})_{+} \right\} 
    \le 
    {\COR_{0c}}/{\OR_c}
    \le 
    \max\left\{ 
    {1}/{(1-\delta_c/\pi_{10\mid c})_{+}}, \ {1}/{(1-\delta_c/\pi_{01\mid c})_{+}}
    \right\}.
\end{align*}
These bounds are sharp and attainable when one of the $p_{zy\mid 1c}$ approaches 1 while the other three approach 0, and $\P(U\ne 0 \mid C=c, T=1)$ is as close as possible to its upper bound $\delta_c$.
\end{theorem}

From Theorem \ref{thm:simple_bound_cor}, the ratio between the causal and observed odds ratios depends on $\delta_c$ specified in Assumption \ref{asmp:bound_u_not_0} and the probabilities $\pi_{zy\mid c}$s in \eqref{eq:o_zy} for tested units. 
When $\delta_c = 0$, the upper and lower bounds on $\COR_{0c}$ reduce to $\OR_c$, and $\COR_{0c}$ is identifiable.
Therefore, Proposition \ref{prop:cor_perfect_test} is a special case of Theorem \ref{thm:simple_bound_cor}. 
As $\delta_c$ increases, $\COR_{0c}$ can differ from $\OR_c$.
Moreover, the sharp bounds on the causal odds ratio under our sensitivity analysis 
depend on 
the whole contingency table $\pi_{zy\mid c}$s. Consequently, studies with the same observed odds ratio can have different sensitivity to imperfections of test-negative designs; see the supplementary material for numerical illustration.

From Theorem \ref{thm:simple_bound_cor}, 
the sharp bounds of $\COR_{0c}$ are attained when the joint distribution $p_{zy\mid 1c}$s of $(Z, Y)$ for tested units at unmeasured confounding level $U=1$ place all the mass at one cell of the $2\times 2$ contingency table, with $\Pr(U=1\mid C=c, T=1)$ being as close to $\delta_c$ as possible. 
The corresponding values of $p_{zy\mid 0c}$s for tested units at unmeasured confounding level $U=0$ can be derived accordingly based on the $\pi_{zy\mid c}$s. 
In Theorem \ref{thm:simple_bound_cor}, 
we use $p_{zy\mid 1c}$s to denote the scenarios in which the sharp bounds are attained, since it is more convenient and intuitive.

\begin{remark}\label{rmk:sharp}
    In Theorem \ref{thm:simple_bound_cor} and the following theorems, the sharpness of the bounds means that there exists a joint distribution of $(Z, Y(1), Y(0), U, C, T)$
    such that 
    (i) Assumption \ref{asmp:U} holds, 
    (ii) the sensitivity analysis constraints, such as Assumption \ref{asmp:bound_u_not_0} here, hold,  
    (iii) the implied distributions of $(Z, Y)$ conditional on $C$ and $T=1$ is coherent with the observed data distribution, 
    and 
    (iv) the corresponding causal odds ratio $\COR_{0c}$ attains the lower or upper bound. 
    We verify the sharpness 
    in the supplementary material.
\end{remark}

\section{Sensitivity analysis for the strength of unmeasured confounding}\label{sec:sen_conf_strength}

In this section, we consider an alternative sensitivity analysis based on the strength of the unmeasured confounding, measured by the confounder's association with the exposure and outcome as quantified shortly. 
Unlike Section \ref{sec:sen_partial} that is tailored for test-negative design, 
the discussion throughout this section can also be applied to more general outcome-dependent sampling such as case-control studies. 
More importantly, this alternative sensitivity analysis can be combined with that in Section \ref{sec:sen_partial}, based on which we can derive narrower bounds for the causal effect of interest, as detailed in Section  \ref{sec:integrate}.

\subsection{Sensitivity analysis based on 
unmeasured confounding
strength}\label{sec:sen_strength}

We quantify the unmeasured confounding strength in the test-negative design by the  probability ratios between the joint distributions of exposure and outcome across tested units with different confounding levels.
Specifically, recalling the definition of $p_{zy\mid uc}$s in \eqref{eq:pzy_u}, we invoke the following assumption.
\begin{assumption}\label{asmp:bound_density}
    For some $\Gamma_c \ge 1$,
    we have 
    $1/\Gamma_c \le p_{zy\mid 1c}/p_{zy\mid 0c} \le \Gamma_c$ for all $z,y\in \{0,1\}$. 
\end{assumption}

In Assumption \ref{asmp:bound_density}, $\Gamma_c$ is a sensitivity parameter that needs to be specified in data analysis and can depend on the value of the observed covariates.
The form of Assumption \ref{asmp:bound_density} on unmeasured confounding is 
reminiscent of
Rosenbaum's sensitivity analysis model for matched observational studies \citep{Rosenbaum02a}, although their meanings are quite different.
When $\Gamma_c=1$, Assumption \ref{asmp:bound_density} 
implies that $(Z, Y) \ind U \mid C=c, T=1$, 
under which the causal odds ratio 
can be identified by the observed odds ratio, 
as shown in the following proposition. 
\begin{proposition}\label{prop:iden_no_conf_strength}
    Under Assumptions \ref{asmp:U} and \ref{asmp:bound_density} with $\Gamma_c = 1$, $\COR_{0c}=\OR_c$.
\end{proposition}

When $\Gamma_c$ becomes larger, 
the strength of unmeasured confounder increases, and there will be more uncertainty, or equivalently wider bounds,  on the underlying true causal odds ratio. 
Below we bound the true causal odds ratio $\COR_{0c}$ based on the observed data distribution $\pi_{zy\mid c}$s in \eqref{eq:o_zy} from the test-negative design and the constraint on the unmeasured confounding strength from Assumption  \ref{asmp:bound_density}. 

\begin{theorem}\label{thm:cor_bound_density}
Under Assumptions \ref{asmp:U} and \ref{asmp:bound_density}, 
the $p_{zy\mid 0c}$ in \eqref{eq:pzy_u} is bounded
by $l_{zy\mid c} \le p_{zy\mid 0c} \le u_{zy\mid c}$, where
\begin{align}\label{eq:lu}
    l_{zy\mid c} = \frac{\pi_{zy\mid c}}{\Gamma_c}
    \quad \text{and} \quad
    u_{zy\mid c} = \min \{\pi_{zy\mid c} \Gamma_c, \  1 \},  
    \quad \ \ 
    (z, y \in \{0,1\}).
\end{align} 
\begin{enumerate}[label={(\roman*)}, topsep=1ex,itemsep=-0.3ex,partopsep=1ex,parsep=1ex]
\item If $l_{11\mid c}+l_{00\mid c}+u_{01\mid c} + u_{10\mid c}\geq 1$, 
then the 
sharp lower bound
on $\COR_{0c}$
is $\underline{\COR}_{0c} = q_{11} q_{00} / (q_{10} q_{01})$, with
$q_{11} = l_{11\mid c}$, $q_{00} = l_{00\mid c}$,
\[
q_{10} = \min\!\left\{
  \max\!\left\{
    l_{10\mid c},\,
    1 - l_{11\mid c} - l_{00\mid c} - u_{01\mid c},\,
    \tfrac{1 - l_{11\mid c} - l_{00\mid c}}{2}
  \right\},
  \, u_{10\mid c},\,
  1 - l_{11\mid c} - l_{00\mid c} - l_{01\mid c}
\right\}, 
\]
and $q_{01} = 1 - q_{11} - q_{00} - q_{10}$; 
otherwise,
the 
sharp lower bound
on $\COR_{0c}$ is 
$\underline{\COR}_{0c} = \min\{{\COR}_{0c}^{(1)}, {\COR}_{0c}^{(2)}\}$, where 
${\COR}_{0c}^{(j)} = q_{11}^{(j)} q_{00}^{(j)} / (q_{10}^{(j)} q_{01}^{(j)})$ with 
$q_{10}^{(j)} = u_{10\mid c}$, $q_{01}^{(j)} = u_{01\mid c}$, 
\begin{align*}
    q_{11}^{(j)} =  
    \begin{cases}
       \max\left\{ l_{11\mid c}, 1-u_{10\mid c}-u_{01\mid c}-u_{00\mid c} \right\}, & \text{ if } j=1, \\
       \min\left\{ u_{11\mid c}, 1-u_{10\mid c}-u_{01\mid c}-l_{00\mid c} \right\},  & \text{ if } j=2, 
    \end{cases}
\end{align*}
and $q_{00}^{(j)} = 1 -q_{10}^{(j)} - q_{01}^{(j)} - q_{11}^{(j)}$. 

\item 
 The 
sharp upper bound
on $\COR_{0c}$, denoted by $\overline{\COR}_{0c}$,
is equal to the reciprocal of 
$\underline{\COR}_{0c}$ calculated using relabeled exposure.
\end{enumerate}

\end{theorem}

Below we explain the rationale for Theorem \ref{thm:cor_bound_density}. 
As verified in the supplementary material,
Assumption \ref{asmp:bound_density} 
imposes constraints on the joint probabilities of exposure and outcome $p_{zy\mid 0c}$s for tested units at unmeasured confounding level $U=0$, as shown in \eqref{eq:lu}. 
We can verify that 
the causal odds ratio $\COR_{0c}$ equals the odds ratio from the joint probabilities $p_{zy\mid 0c}$s, 
and the bounds on $\COR_{0c}$ depend on the observed data distribution $\pi_{zy\mid c}$s and the sensitivity parameter $\Gamma_c$ only through the bounds $l_{zy}$s and $u_{zy}$s in \eqref{eq:lu}. 
We can further show that the minimum of the causal odds ratio $\COR_{0c} = p_{11\mid 0c} p_{00 \mid 0}/(p_{10\mid 0c} p_{01 \mid 0})$ under the constraint in \eqref{eq:lu} must be obtained when either $(p_{11\mid 0c}, p_{00 \mid 0})$ in the numerator take their smallest possible values $(l_{11\mid c}, l_{00\mid c})$ or  
$(p_{10\mid 0c}, p_{01 \mid 0})$ in the denominator take their largest possible values $(u_{10\mid c}, u_{01\mid c})$, which correspond to the two cases in Theorem \ref{thm:cor_bound_density}. 

As a side remark, the bounds from Theorem \ref{thm:cor_bound_density} also work for the true causal odds ratio $\COR_{1c}$ at unmeasured confounding level $U=1$, due to the symmetry of the sensitivity analysis constraint in Assumption \ref{asmp:bound_density}. 
In general, the test-negative design is believed to, at least partially, control for the unmeasured confounding; for example, as discussed before, units in the design may mostly have positive health-care-seeking behavior. Therefore, 
under a test-negative design, 
different levels of unmeasured confounding is generally not symmetric, and we will mainly focus on the causal odds ratio at 
the unmeasured confounding level that 
can be partially controlled by the test-negative design.

\section{A combined sensitivity analysis 
}\label{sec:integrate}

In Sections \ref{sec:sen_partial} and \ref{sec:sen_conf_strength},  we consider two distinct 
sensitivity analyses with closed-form sharp bounds on the true causal odds ratio. 
Assumptions \ref{asmp:bound_u_not_0} and \ref{asmp:bound_density} for the two sensitivity analyses in  Sections \ref{sec:sen_partial} and \ref{sec:sen_conf_strength} have quite different meanings.  
Assumption \ref{asmp:bound_u_not_0} comes from the violation of Assumption \ref{asmp:U_T} regarding how the choice of testing is associated the unmeasured confounding, while Assumption \ref{asmp:bound_density} focuses on the 
strength of the unmeasured confounding in terms of its impact on the exposure and outcome.
Below we consider a combined sensitivity analysis incorporating both Assumptions \ref{asmp:bound_u_not_0} and \ref{asmp:bound_density}.

\begin{theorem}\label{thm:cor_bound_additional}
Under Assumptions \ref{asmp:U}, \ref{asmp:bound_u_not_0} and \ref{asmp:bound_density},  
the $p_{zy\mid 0c}$ in \eqref{eq:pzy_u} is bounded by $\tilde{l}_{zy\mid c} \le p_{zy\mid 0c} \le \tilde{u}_{zy\mid c}$ for $z,y=0,1$, where 
\begin{align}\label{eq:lu_delta_Gamma}
    \tilde{l}_{zy\mid c} = \max \left\{ \frac{\pi_{zy\mid c}}{\delta_c\Gamma_c+(1-\delta_c)},  \  \frac{\pi_{zy\mid c}-\delta_c}{1-\delta_c} \right\}, 
    \quad \ \ 
    \tilde{u}_{zy\mid c} = \min \left\{\frac{\pi_{zy\mid c} \Gamma_c }{\delta_c+(1-\delta_c)\Gamma_c}, \  1 \right\}, 
\end{align}
and $(\pi_{zy\mid c}-\delta_c)/(1-\delta_c)$ is defined as $0$ when $\delta_c = 1$. 
The
sharp bounds 
on the causal odds ratio $\COR_{0c}$ are achieved 
in the same way as that in Theorem \ref{thm:cor_bound_density}, except that $l_{zy\mid c}$s and $u_{zy\mid c}$s there are replaced by $\tilde{l}_{zy\mid c}$s and $\tilde{u}_{zy\mid c}$s in \eqref{eq:lu_delta_Gamma}. 
\end{theorem}

The proof of Theorem \ref{thm:cor_bound_additional} follows by the same logic as that for Theorem \ref{thm:cor_bound_density}, 
because
Assumptions \ref{asmp:bound_u_not_0} and \ref{asmp:bound_density} impose the bounds in \eqref{eq:lu_delta_Gamma} on the $p_{zy\mid 0c}$s. 
From Propositions \ref{prop:cor_perfect_test} and \ref{prop:iden_no_conf_strength}, when either $\delta_c = 0$ or $\Gamma_c = 1$, the bounds on the causal odds ratio $\COR_{0c}$ from our sensitivity analysis in Theorem \ref{thm:cor_bound_additional} will shrink to a point, which is the observed odds ratio $\OR_c$. When $\delta_c$ and $\Gamma_c$ increase, $\COR_{0c}$ differs from $\OR_c$.
We defer the discussion on specification of these sensitivity parameters to Section \ref{sec:specify}.

Both Theorems \ref{thm:simple_bound_cor} and \ref{thm:cor_bound_density} are special cases of the combined sensitivity analysis in 
Theorem \ref{thm:cor_bound_additional}, with either $\Gamma_c=\infty$ or $\delta_c = 1$. 
Below we explain how the combined sensitivity analysis improves the previous individual ones. 
First, the bounds in Theorem \ref{thm:simple_bound_cor} are achieved at 
very extreme and probably unrealistic scenarios in which the joint distribution of the exposure and outcome for tested units with confounding level $U=1$  places all the mass at one point. Such extreme scenarios will be avoided once we have Assumption \ref{asmp:bound_density}. 
Second, with the partial control of unmeasured confounding under Assumption \ref{asmp:bound_u_not_0}, 
we have narrower bounds on $p_{zy\mid 0c}$s 
compared to that 
in Theorem \ref{thm:cor_bound_density}. 
Specifically, 
for every $z,y\in \{0,1\}$, 
the bound $[\tilde{l}_{zy\mid c}, \tilde{u}_{zy\mid c}]$ for $p_{zy\mid 0c}$ in \eqref{eq:lu_delta_Gamma} becomes tighter as $\delta_c$ decreases.

\section{Extensions to categorical exposure and outcome with a general unmeasured confounder}\label{sec:categorical}

Sections \ref{sec:framework}--\ref{sec:integrate} focused on sensitivity analysis with binary unmeasured confounder, outcome and exposure. 
We now 
generalize our methods
to allow categorical values for the exposure and outcome, and general values for the unmeasured confounder, 
as detailed below.
In summary, the theory below demonstrates that the sensitivity analysis methods for binary exposure, outcome and unmeasured confounder in Sections \ref{sec:framework}--\ref{sec:integrate} are still applicable for 
this more general setting, 
but require adjustments in constraints and interpretation to accommodate both the categorical nature of the exposure and outcome and the general values of the unmeasured confounder. 
Therefore, the theory below greatly enhances the applicability of the proposed sensitivity analysis.

In real-world scenarios, exposures and outcomes often have multiple levels. For instance, in the context of Covid-19, the exposure could be different vaccines or various dosages of the same vaccine, and the outcome could be different levels of severity, such as mild, moderate, severe, and critical illness
\citep{COVID19Guidelines}. 
Besides, 
allowing multiple levels of the outcome is also helpful for ensuring Assumption \ref{asmp:U}(ii), 
which requires the outcome to block all direct paths from the exposure to testing. 
For example, 
similar to the discussion in Section \ref{sec:cat_connect}, 
Assumption \ref{asmp:U}(ii) is more plausible when the outcome also contains the severity of the symptoms in addition to the infection status. 
Furthermore, the unmeasured confounding may also have multiple levels, representing different levels of health-care-seeking behavior.

Suppose that the exposure \(Z\) and outcome \(Y\) take values, respectively, in \(\{0, 1, \ldots, I\}\) and \(\{0, 1, \ldots, J\}\) for some \(I, J \geq 1\), and the unmeasured confounder $U$ takes values in a general set $\mathcal{U}$. 
We consider again the causal estimand as in \eqref{eq:COR}, comparing two treatment levels in terms of two outcome values:
for $0\le l \ne l'\le I$, $0\le j \ne j'\le J$, and $u \in \mathcal{U}$, 
\begin{align}\label{eq:COR_cat}
    \COR_{uc}(l, l'; j,j') \equiv \frac{\P(Y(l)=j\mid U = u, C = c)/\P(Y(l) = j' \mid U = u, C=c)}{\P(Y(l')=j\mid U = u, C = c)/\P(Y(l') = j' \mid U = u, C=c)}. 
\end{align}
When the outcome has two levels, \eqref{eq:COR_cat} reduces to the usual causal odds ratio. 
We assume that the two treatment levels of interest are $Z=0,1$ and  the two outcome values of interest are $Y=0,1$, 
and consider the causal estimand $\COR_{uc}$ as in \eqref{eq:COR} and \eqref{eq:COR_cat}. 
This does not lose any generality since we can always relabel the exposure and outcome.

Define 
\begin{align}\label{eq:T_tilde}
    \tilde{T} = T \cdot \I(Z\in \{0,1\}) \cdot \I(Y \in \{0,1\}).
\end{align}
By \eqref{eq:T_tilde}, $\tilde{T}=1$ if and only if the unit is included in the test-negative design and the exposure and outcome values fall into the categories of interest. 
We denote the joint distribution of exposure and outcome among units with $\tilde{T}=1$ and covariate value $c$ by: 
\begin{equation}\label{eq:o_tilde}
    \tilde{\pi}_{zy\mid c} = \P(Z=z, Y=y\mid C=c, \tilde{T} = 1),
\end{equation}
which will be used to infer the casual estimand defined as in \eqref{eq:COR}.  
We define further 
\begin{align}\label{eq:p_tilde}
    \tilde{p}_{zy\mid uc} = \P(Z=z, Y=y\mid U=u, C=c, \tilde{T}=1)
\end{align}
for $z,y\in \{0,1\}$ and $u\in \mathcal{U}$, analogous to \eqref{eq:pzy_u} but conditioning on $\tilde{T}$ instead of $T$.

We now consider the sensitivity analysis. 
Without loss of generality, we still assume that $U=0\in \mathcal{U}$ is the level that the test-negative design aims to control for,
and 
focus on the causal effect $\COR_{0c}$ at level $U=0$ defined as in \eqref{eq:COR} and \eqref{eq:COR_cat}.
The following assumptions generalize Assumptions \ref{asmp:bound_u_not_0} and \ref{asmp:bound_density}, respectively, to allow for categorical exposure and outcome. 

\begin{assumption}\label{asmp:categorical_ZY_not_0}
    For some $0\le \delta_c \le 1$, we have 
   $\Pr(U\ne 0 \mid C=c, \tilde{T}=1) \le \delta_c$. 
\end{assumption}

\begin{assumption}\label{asmp:categorical_ZY_density_ratio}
     For some $\Gamma_c\ge 1$, we have $1/\Gamma_c \leq \tilde{p}_{zy\mid uc}/\tilde{p}_{zy\mid 0c} \leq \Gamma_c$ for all $z, y\in \{0,1\}$ and $u \in \mathcal{U}$.
\end{assumption}

Theorem \ref{thm:cat_ZYU} below shows that the resulting sensitivity analysis for $\COR_{0c}$ based on $\tilde{\pi}_{zy\mid c}$s in \eqref{eq:o_tilde} and Assumptions \ref{asmp:U}, \ref{asmp:categorical_ZY_not_0} and \ref{asmp:categorical_ZY_density_ratio} can be conducted in the same way as that in Theorem \ref{thm:cor_bound_additional}. 

\begin{theorem}\label{thm:cat_ZYU}
    The sharp bounds on the causal effect $\COR_{0c}$ under Assumptions \ref{asmp:U}, \ref{asmp:categorical_ZY_not_0} and \ref{asmp:categorical_ZY_density_ratio} are the same as those in Theorem \ref{thm:cor_bound_additional} under Assumptions \ref{asmp:U},
    \ref{asmp:bound_u_not_0} and \ref{asmp:bound_density}, 
    with $\pi_{zy\mid c}$s there replaced by $\tilde{\pi}_{zy\mid c}$s and the same specification for $(\delta_c, \Gamma_c)$. 
\end{theorem}

Theorem \ref{thm:cat_ZYU} has two implications. First, %
when the exposure and outcome have more than two levels, we can simply focus on the two exposure levels and two outcome values of interest, and 
analyze
the resulting $2\times 2$ contingency tables. 
The subsequent sensitivity analysis can be done in the same way as that for the binary exposure and outcome. 
Such a strategy has also been used in the real-data application in Section \ref{sec:app}.

Second, 
when Assumption \ref{asmp:categorical_ZY_not_0} holds with $\delta_c = 0$, i.e., the test-negative design perfectly controls the unmeasured confounder, 
$\COR_{0c}$ is identifiable and equals $\widetilde{\OR}_c \equiv \tilde{\pi}_{11\mid c}\tilde{\pi}_{00\mid c}/(\tilde{\pi}_{10\mid c}\tilde{\pi}_{01\mid c})$, which is the odds ratio based on  \eqref{eq:o_tilde}. 
Moreover, $\widetilde{\OR}_c$ 
is equivalent to the expression of $\OR_c$ in \eqref{eq:OR} with $Z$ and $Y$ potentially having more than two levels. 
This implies that Proposition \ref{prop:cor_perfect_test} holds more generally 
with categorical exposure and outcome. 
This further generalizes the discussion in Section \ref{sec:cat_connect}.

{\rev

Below we provide several remarks on the potential benefits of a carefully chosen categorical outcome.
First, as discussed above, Assumption \ref{asmp:U}(ii) is more likely to hold when $Y$ contains richer information.
Moreover, the definition of $Y$ and the outcome values of interest may depend on the inclusion criterion of the test-negative design. 
For example, when the design includes only units who seek testing for COVID-19 and have COVID-19-like illness symptoms leading to hospitalization lasting more than 24 hours, we may focus on two outcome levels: $Y=1$ corresponding to COVID-19 infection with illness symptoms leading to hospitalization of more than 24 hours, and $Y=0$ corresponding to no COVID-19 infection with such illness symptoms.

Second, with an appropriate choice of the two outcome values of interest, the unmeasured confounder may be better controlled, so that Assumption \ref{asmp:categorical_ZY_not_0} holds with a smaller $\delta_c$.
For example, some test-negative designs include all units who seek testing regardless of symptoms. In this setting, if we focus on the two outcome values defined by infection status together with illness symptoms (infection with illness symptoms versus no infection with illness symptoms), then conditioning on $\tilde{T}$ in \eqref{eq:T_tilde} effectively restricts the sample to units with illness symptoms. This can be viewed as modifying the selection criterion, under which unmeasured health-care-seeking behavior may be better controlled; see also \citet{ortiz-brizuela_potential_2025} for related discussion.

Third, among the two outcome values we focus on, in addition to the one whose risk is of primary interest, we may choose the other to be less affected by the exposure. If an assumption analogous to \eqref{eq:asmp_no_virus_int} holds, then the causal estimand $\COR_{0c}$ reduces to the causal risk ratio $\CRR_{0c}$ as defined in \eqref{eq:CRR}, which may be easier to interpret. 
}

\section{Sensitivity analysis with further constraints on treatment effect heterogeneity}\label{sec:sen_effect_heter}

\subsection{Binary exposure, outcome and unmeasured confounder}\label{sec:sen_effect_heter_binary}
As briefly discussed at the end of Sections \ref{sec:potential} and \ref{sec:sen_strength}, 
we focus mainly on the causal odds ratio for units with a particular unmeasured confounding level.  
Understanding the variation of treatment effects across different levels of the unmeasured confounding could be challenging
due to lack of information in the test-negative design.
Below we will consider the treatment effect heterogeneity from a different perspective. Specifically, we will instead impose constraint on the treatment effect heterogeneity across different unmeasured confounding levels, to improve the bounds on $\COR_{0c}$. 
For convenience, we first consider the case where the unmeasured confounder, exposure and outcome are all binary, and then consider the extension to general case in the next subsection. 

\begin{assumption}\label{asmp:effect_heter}
    For some $\xi_c \ge 1$, we have  
    $1/\xi_c \le \COR_{0c}/\COR_{1c} \le \xi_c$. 
\end{assumption}

In Assumption \ref{asmp:effect_heter}, $\xi_c$ is a sensitivity parameter that we need to specify in data analysis and can depend on the value of the observed covariates.
Below we give two remarks regarding Assumption \ref{asmp:effect_heter}. 
First, 
unlike Assumptions \ref{asmp:bound_u_not_0} or \ref{asmp:bound_density} under which the causal odds ratio $\COR_{0c}$ becomes identifiable when $\delta_c=0$ or $\Gamma_c=1$, 
the $\COR_{0c}$ is not identifiable even if Assumption \ref{asmp:effect_heter} holds with $\xi_c = 1$.
This is due to the well-known noncollapsibility issue of the odds ratios. 
Second, 
the bounds on $\COR_{0c}$ under Assumptions \ref{asmp:bound_u_not_0} and \ref{asmp:effect_heter} are the same as that in Theorem \ref{thm:simple_bound_cor} under Assumption \ref{asmp:bound_u_not_0} alone. 
The reason is that the information on odds ratios cannot rule out extreme configurations of contingency tables. 
Specifically, 
for any $\Lambda>0$, 
there always exists 
$p_{zy\mid 1c}$s such that the corresponding $\COR_{1c}$ is arbitrarily close to $\Lambda$ and the joint distribution places most of its mass at one point so that $\COR_{0c}$ can be arbitrarily close to its bounds in Theorem \ref{thm:simple_bound_cor}. 
For example, 
we can set 
$(p_{11\mid 1c}, p_{00\mid 1c}, p_{10\mid 1c}, p_{01\mid 1c}) = (1-\Lambda\varepsilon^2-2\varepsilon, \Lambda\varepsilon^2, \varepsilon, \varepsilon)$ 
for an arbitrarily small $\varepsilon>0$. 
Nevertheless, 
if additionally Assumption \ref{asmp:bound_density} holds with $\Gamma_c<\infty$, Assumption \ref{asmp:effect_heter} on treatment effect heterogeneity can help narrow the bounds on $\COR_{0c}$; see Section \ref{sec:app} for examples. 
This also explains why we focus on Assumptions \ref{asmp:bound_u_not_0} and \ref{asmp:bound_density} for the sensitivity analyses in Sections \ref{sec:sen_partial}--\ref{sec:integrate}.

Under Assumptions \ref{asmp:U}, \ref{asmp:bound_u_not_0}, \ref{asmp:bound_density} and \ref{asmp:effect_heter} for some $\delta_c\in [0,1]$ and $\Gamma_c, \xi_c \ge 1$,
it is challenging to derive the closed-form bounds on the causal odds ratio $\COR_{0c}$. 
Fortunately, we can 
formulate 
the sensitivity analysis of $\COR_{0c}$ as a quadratically constrained quadratic programming problem. 
Quadratic programming is in general difficult. 
Fortunately, our numerical experience shows that the quadratic programming for the sensitivity analysis of $\COR_{0c}$
can often be efficiently solved 
using some standard software such as Gurobi \citep{gurobi}.

\begin{theorem}\label{thm:sen_3_constr}
Suppose that Assumptions \ref{asmp:U}, \ref{asmp:bound_u_not_0}, \ref{asmp:bound_density} and \ref{asmp:effect_heter} hold.  
\begin{enumerate}[label={(\roman*)}, topsep=1ex,itemsep=-0.3ex,partopsep=1ex,parsep=1ex]
\item The sharp lower bound on the causal odds ratio $\COR_{0c}$ is 
the solution to the following  quadratic programming problem:  
minimize $t_1 t_3$ subject to: 
\begin{align}\label{eq:qp}
    & \pi_{zy\mid c} = p_{zy\mid 0c}(1-w) + p_{zy\mid 1c}w, 
    \\ 
    &\sum_{z,y} {p_{zy\mid 0c}} = 1,
    \quad 
    0\leq p_{zy\mid 0c}, p_{zy\mid 1c} \leq 1,
    \label{eq:qp_prob}
    \\
    & { 0\le \ } w\le \delta_c, \quad 
    r_{zy} p_{zy\mid 1c} = p_{zy\mid 0c}, \quad 
    1/\Gamma_c \leq r_{zy} \leq \Gamma_c, 
    \label{eq:qp_u0_density}
    \\
    & 
    s_{1} = r_{11}r_{00} \quad 
    s_2 = r_{10}r_{10} \quad 
    s_2s_3 = 1, \quad 
    1/\xi_c \leq s_1s_3 \leq  \xi_c,
    \label{eq:qp_heter}
    \\
    & 
    t_{1} = p_{11\mid 0c}p_{00\mid 0c}, \quad 
    t_2 = p_{01\mid 0c}p_{10\mid 0c}, \quad 
    t_2t_3 = 1.  
    \label{eq:qp_objective}
\end{align}
\item The sharp upper bound on the causal odds ratio $\COR_{0c}$ is the solution to the quadratic programming problem: maximize $t_1 t_3$ subject to the same constraints as in \eqref{eq:qp}--\eqref{eq:qp_objective}. 
\end{enumerate}
\end{theorem}

In the quadratic programming in Theorem \ref{thm:sen_3_constr}, 
we introduce $w$ to denote $\Pr(U\ne 0 \mid C=c, T=1)$, which is bounded by $\delta_c$ under Assumption \ref{asmp:bound_u_not_0}, as shown in \eqref{eq:qp_u0_density}.  
In addition, we introduce $r_{zy}$ to denote the probability ratio $p_{zy\mid 0c}/p_{zy\mid 1c}$ for $z,y \in {0,1}$,  $(t_1,t_2,t_3)$ to help represent the causal odds ratio $\COR_{0c} = p_{11\mid 0c}p_{00\mid 0c}/(p_{10\mid 0c}p_{01\mid 0c})$, 
and $(s_1,s_2,s_3)$ to help represent the ratio $\COR_{0c}/\COR_{1c}$ between the two causal odds ratios. 
We introduce them
to express both the objective and the constraints in quadratic form.
Below we briefly explain the constraints in the quadratic programming in \eqref{eq:qp}--\eqref{eq:qp_objective}. 
The constraints in \eqref{eq:qp} link the observed contingency table in \eqref{eq:o_zy} and the hidden ones in \eqref{eq:pzy_u}, where $w$ represents the conditional probability that $U$ is not zero. 
The constraints in \eqref{eq:qp_prob} put natural requirement on the probabilities $p_{zy\mid uc}$s.
The constraints in \eqref{eq:qp_u0_density} impose Assumptions \ref{asmp:bound_u_not_0} and \ref{asmp:bound_density}. The constraints in \eqref{eq:qp_heter} impose Assumption \ref{asmp:effect_heter}, because $\COR_{uc} = p_{11\mid uc} p_{00\mid uc}/(p_{10\mid uc} p_{01\mid uc})$ as discussed after Theorem \ref{thm:cor_bound_density}. 
The constraints in \eqref{eq:qp_objective} are for representing the objective function $\COR_{0c}$.

\subsection{Categorical exposure and outcome with a general unmeasured confounder}\label{sec:cat_3_constraints}

We now consider the same setting in Section \ref{sec:categorical} with categorical exposure,  categorical outcome and a general unmeasured confounder. 
As detailed below, by modifying the constraints and interpretation, 
the sensitivity analysis for the binary case in Theorem \ref{thm:sen_3_constr} also works for the general case. 
This enhances the applicability of the
proposed sensitivity analysis.

Compared with 
Assumptions \ref{asmp:bound_u_not_0} and \ref{asmp:bound_density}, 
the extension of Assumption \ref{asmp:effect_heter} to a general $U$ is more involved. 
Intuitively, a natural extension of Assumption \ref{asmp:effect_heter} is to consider bounds on the ratio between the causal effect $\COR_{uc}$s defined in \eqref{eq:COR} and \eqref{eq:COR_cat} at different levels of the unmeasured confounder. 
However, due to the noncollapsibility of the odds ratio, the resulting sensitivity analysis bounds on \(\COR_{0c}\) may vary with the number of possible values of $U$,
which complicates the sensitivity analysis and makes it harder to interpret. 
Therefore, we instead extend Assumption \ref{asmp:effect_heter} to restrict the heterogeneity of odds ratio among tested units with zero and nonzero unmeasured confounders. 
Similar to the discussion after Theorem \ref{thm:cor_bound_density}, 
the causal effect
$\COR_{uc}$ for units with unmeasured confounding $U=u$ can be equivalently written as 
$\COR_{uc} = \tilde{p}_{11\mid uc}\tilde{p}_{00\mid uc}/(\tilde{p}_{10\mid uc}\tilde{p}_{01\mid uc})$, 
recalling the definition of $\tilde{p}_{zy\mid uc}$s in \eqref{eq:p_tilde}. 
Define 
$
\tilde{p}_{zy\mid  \ne 0,c} = \P(Z=z, Y=y\mid U \ne 0, C = c, \tilde{T}=1)
$
for $z,y\in \{0,1\}$  
and \(\OR_{\ne 0, c} = \tilde{p}_{11\mid \ne 0,c} \tilde{p}_{00\mid \ne 0,c}/(\tilde{p}_{10\mid \ne 0,c} \tilde{p}_{01\mid \ne 0,c})\). They denote, respectively, the joint distribution of the exposure and outcome and the corresponding odds ratio for tested units with unmeasured confounding $U\ne 0$ and exposure and outcome falling in the categories of interest.

\begin{assumption}\label{asmp:categorial_U_effect_heter}
    For some $\xi_c\ge 1$, we have \(\xi_c^{-1} \le \OR_{\ne 0, c}/\COR_{0c}\le \xi_c\).
\end{assumption}

\begin{theorem}\label{thm:cat_ZYU_supp}
The sharp bounds on the causal effect $\COR_{0c}$ 
under  Assumptions \ref{asmp:U}, \ref{asmp:categorical_ZY_not_0}, \ref{asmp:categorical_ZY_density_ratio} and \ref{asmp:categorial_U_effect_heter}
can be obtained in the same way as those in Theorem \ref{thm:sen_3_constr} under Assumptions \ref{asmp:U},
\ref{asmp:bound_u_not_0}, \ref{asmp:bound_density} and \ref{asmp:effect_heter}, 
with $\pi_{zy\mid c}$s there replaced by $\tilde{\pi}_{zy\mid c}$s and the same specification for $(\delta_c, \Gamma_c, \xi_c)$.    
\end{theorem}

\section{Practical implementation}\label{sec:prac}

\subsection{Specification of the sensitivity parameters}\label{sec:specify}

The general sensitivity analysis in Theorem \ref{thm:sen_3_constr} involves three sensitivity parameters $(\delta_c, \Gamma_c, \xi_c)$ from Assumptions \ref{asmp:bound_u_not_0}, \ref{asmp:bound_density} and \ref{asmp:effect_heter}, which characterize different aspects of the unmeasured confounding as discussed before. 
The choice of these sensitivity parameters 
is a researcher choice and 
is important for the proposed sensitivity analysis. 
We first discuss the relation among these three sensitivity parameters. 
The parameters $\delta_c$ and $(\Gamma_c, \xi_c)$ are free in the sense that one does not restrict the other. 
Mathematically, $\delta_c$ is about the conditional distribution of $U$ given $(C, T)$, whereas $(\Gamma_c, \xi_c)$ is about the conditional distribution of $(Z,Y)$ given $(U,C, T)$.
On the contrary, Assumption \ref{asmp:bound_density} for a given $\Gamma_c$ implies Assumption \ref{asmp:effect_heter} with $\xi_c = \Gamma_c^4$. However, we may believe more restrictive bounds on the treatment effect heterogeneity, which can be useful for further narrowing the causal bounds.

Below we discuss the specification of these sensitivity parameters in practice. 
If prior or domain knowledge is available, then we can specify reasonable ranges for $(\delta_c, \Gamma_c, \xi_c)$ and investigate the possible range of the causal odds ratio. 
In general, we can explore a wide range of sensitivity parameters. Specifically, supposing that the observed odds ratio is less than 1, we can try different values of $\delta_c$, and create heatmaps to display the upper bounds of the causal odds ratio under various combinations of $(\Gamma_c, \xi_c)$. 
In addition, we can plot the boundary curve that delineates the minimum amount of unmeasured confounding, expressed in terms of $(\Gamma_c, \xi_c)$, needed to nullify the observed association (e.g., to make the upper bound of the causal odds ratio equal to $1$).
This is illustrated in the application in Section \ref{sec:app}.

We can also gain insights into plausible ranges of the sensitivity parameters through the observed covariates. Similar ideas have been used in other sensitivity analyses; see, e.g., \citet{zhang2020calibrated} and \citet{lu2023flexible}.  
Note that both sensitivity parameters $\Gamma_c$ and $\xi_c$ depend on the heterogeneity of the distribution of $(Z, Y)$ across different levels of $U$, conditional on the observed covariates and being tested (i.e., $C=c$ and $T=1$). 
Intuitively, the heterogeneity of the distribution of $(Z, Y)$ with respect to the observed covariates can provide useful guidance for calibrating and specifying the sensitivity parameters $\Gamma_c$ and $\xi_c$. 
For example, we can investigate this heterogeneity across different levels of a subset of observed covariates, conditional on the remaining observed covariates and being tested. 
In addition, when some of the observed covariates can serve as a good proxy for the unmeasured $U$, for example by capturing the health-care-seeking behavior, 
then it can also be used to calibrate the sensitivity parameter $\delta_c$. 
We illustrate such calibration in the application in Section \ref{sec:app}.

\begin{remark}
The sensitivity analyses in Theorems \ref{thm:simple_bound_cor}-- \ref{thm:cor_bound_additional}
involve less sensitivity parameters compared to the most general one in Theorem \ref{thm:sen_3_constr}. 
Equivalently, they set some parameters to their largest possible values, i.e., $\delta_c=1$, $\Gamma_c = \infty$ or $\xi_c = \infty$. 
There are trade-offs among these sensitivity analysis methods. 
Sensitivity analysis with fewer sensitivity parameters requires less effort for specifying these parameters, 
but may provide wider bounds on the causal effect of interest. 
It is worth noticing that once we set $\Gamma_c = \infty$, smaller value of $\xi_c$ is not helpful for narrowing the bounds on $\COR_{0c}$, as discussed in Section \ref{sec:sen_effect_heter_binary}.  
\end{remark}

\subsection{Confidence bounds}\label{sec:conf_bound}

The sensitivity analysis methods proposed in Sections \ref{sec:sen_partial}--\ref{sec:sen_effect_heter} all presume that the true observed frequencies $\pi_{zy\mid c}$s 
under the test-negative design are known exactly without any uncertainty. 
This, however, is rarely the case in practice. 
Suppose instead that we have confidence sets for $\pi_{zy\mid c}$ under suitable models for the conditional distribution of $(Z,Y)$ given $C=c$ and $T=1$; see the supplementary materials for examples, including fully saturated models for categorical $C$ and parametric models for general $C$.
We 
can then construct confidence bounds on the true causal odds ratio $\COR_{0c}$ using the proposed sensitivity analysis, 
taking into account the uncertainty in estimating the $\pi_{zy\mid c}$s from the observed data. 
We briefly explain the main idea here, 
and relegate the details to the supplementary material.

Specifically, we first construct a confidence set for the $\pi_{zy\mid c}$s based on the observed data for any prespecified covariate value $c$. 
In particular, the confidence set can be asymptotically valid either point-wise at a given $c$ or uniformly across all levels of $c$; see the supplementary material for details on its construction. 
We then search for the worst-case causal bounds over all possible $\pi_{zy\mid c}$s in the confidence set. 
We propose closed-form confidence bounds on $\COR_{0c}$ under Assumptions \ref{asmp:bound_u_not_0} and \ref{asmp:bound_density} for any given sensitivity parameters $(\delta_c, \Gamma_c)$. 
With additional constraint on treatment effect heterogeneity as in Assumption \ref{asmp:effect_heter}, 
we can compute the confidence bounds 
again through quadratic programming, by augmenting \eqref{eq:qp}--\eqref{eq:qp_objective} with the constraint that $\pi_{zy\mid c}$s belong to the confidence set.
We present the computational details in the supplementary material. 

\begin{remark}
Our analysis conditions on the observed covariates and thus allows us to examine effect heterogeneity across different covariate levels.
In practice, especially when covariates are continuous, it may be desirable to summarize the results into a single causal estimand.
One approach is to take a weighted average of the $\COR_{0c}$s with weights given by the marginal distribution of the covariates; the corresponding confidence bounds can be obtained by applying the same weighting to the simultaneous confidence bounds on the $\COR_{0c}$s.
Moreover, we may specify different sensitivity parameters for different covariate levels using the calibration strategy in Section~\ref{sec:specify}. Alternatively, we can adopt common sensitivity parameters across all covariate levels and examine how the resulting inference varies with these values.
\end{remark}

\section{Illustration}\label{sec:app}

\begin{figure}[htbp]
    \centering
    \includegraphics[width=1\linewidth]{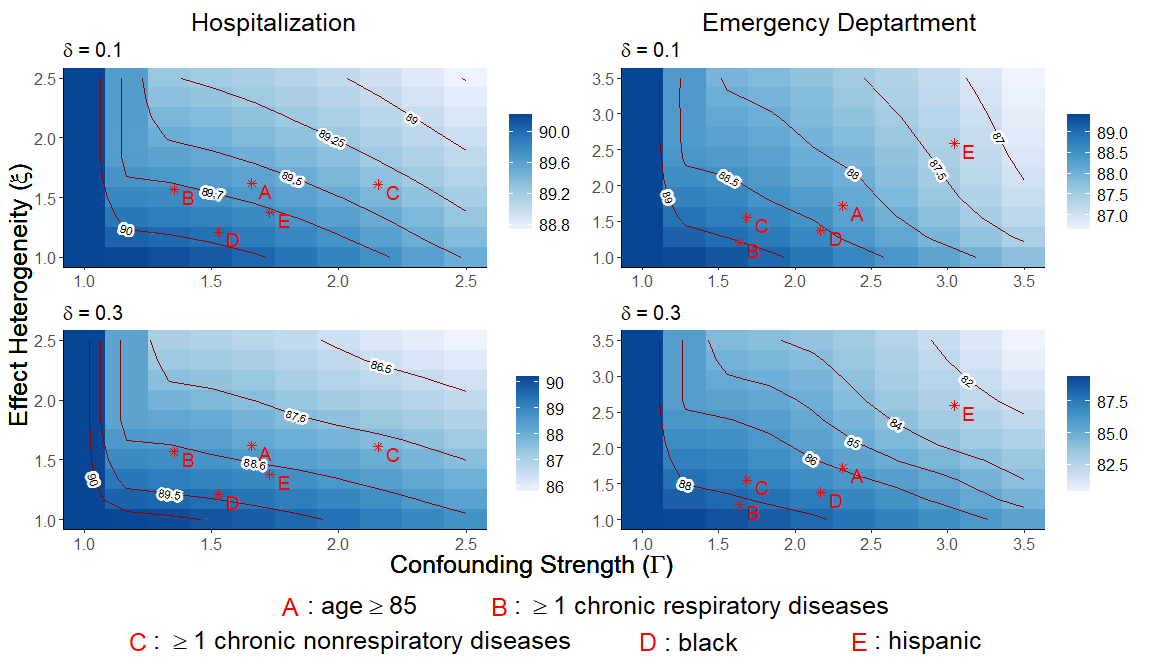}
    \caption{
    (i) Each panel presents the lower limits of the 95\% confidence bounds for the efficacy of full two-dose mRNA vaccination under various combinations of sensitivity parameters $(\delta, \Gamma, \xi)$ and two selection criteria—hospitalization and emergency department/urgent care visits. The $x$- and $y$-axes represent $\Gamma$ and $\xi$, respectively. The top (bottom) panels correspond to $\delta = 0.1$ ($\delta = 0.3$), and the left (right) panels correspond to hospitalization (emergency department/urgent care visits) as the selection criterion.  
    (ii) The brown contour lines indicate combinations of $(\Gamma, \xi)$, for fixed $\delta = 0.1$ or $0.3$, at which the lower limits of the confidence bounds attain the annotated values.  
    (iii) The five labeled points, ``A''--``E'', illustrate the heterogeneity in the joint probabilities and odds ratios of the exposure--outcome distribution across different covariate strata under the test-negative design. Specifically, ``A''--``E'' correspond, respectively, to the binary covariates: age $\geq 85$ years, presence of at least one chronic respiratory condition, presence of at least one chronic nonrespiratory condition, Black, and Hispanic.
    }
    \label{fig:heatmap}
\end{figure}

We illustrate the proposed sensitivity analysis methods using data from the study in \citet{Thompson2021} that involves adults aged 50 years or older with Covid-19–like illness. The study includes 41,552 hospital admissions 
and 21,522 emergency departments or urgent care clinics visits 
in multiple states across the US from January 1 through June 22, 2021, 
and investigates the effectiveness of three popular vaccines developed by Pfizer, Moderna, and Johnson~\&~Johnson, respectively. 
The study employed a test-negative design to estimate vaccine effectiveness by comparing the odds of testing positive for SARS-COV-2 infection between vaccinated and unvaccinated patients.

{\rev 
In the study, each medical visit was assigned an index date, defined as either the date of respiratory specimen collection associated with the most recent test result prior to the medical visit, or the date of the medical visit if testing occurred only after the admission or visit date.
For each individual, 
COVID-19 vaccination status was defined by whether vaccination occurred prior to the index date, based on data from state immunization registries, electronic health records, and claims data.
Specifically, an individual can be fully vaccinated (receiving a single dose of the Johnson \& Johnson vaccine or the second dose of the Pfizer or Moderna vaccine at least 14 days before the index date), partially vaccinated, or unvaccinated. Here we focus on comparing two levels of the exposure $Z$: fully vaccinated versus unvaccinated.

Following \citet{Thompson2021}, we consider three selection criteria, where the inclusion indicator $T=1$ if an individual is hospitalized for more than 24 hours, admitted to the ICU (a subset of hospitalization), or visits an emergency department or urgent care clinic, respectively. 
Accordingly, as discussed in Sections \ref{sec:cat_connect} and \ref{sec:categorical}, we can define the outcome $Y$ as a categorical variable encoding both vaccination status and the symptoms leading to hospitalization, ICU admission, or an emergency department or urgent care clinic visit. We focus on two levels of $Y$: fully vaccinated with the corresponding symptoms and unvaccinated with the corresponding symptoms.

The study also collected covariate information, including age, race/ethnicity, and medical conditions. However, some potential confounders, such as health-care-seeking behavior, may be unmeasured.}
In the main paper, we perform the analysis without conditioning on any observed covariates, unless otherwise stated; we will thus omit the conditioning on $c$ in this section. 
In the supplementary material, 
we also perform analysis conditioning on some observed covariates, such as age, race/ethnic group, and medical conditions. 
However, because individual-level data are not available, 
we are not able to perform analyses conditioning on all these covariates simultaneously.

We focus first on the effectiveness of full two-dose mRNA vaccination, considering either hospitalization or emergency department/urgent care visits as the selection criterion. 
When the test-negative design perfectly controls for unmeasured confounding due to, for example, health-care-seeking behavior, vaccine efficacy, defined as 
$1-\COR_0$ in our notation, 
can be estimated by $1$ minus the observed odds ratio (i.e., $1-\OR$). Under the two selection criteria, these estimates are $92.02\%$ and $91.80\%$, with corresponding 95\% confidence intervals $(91.01\%,93.02\%)$ and $(90.54\%,93.05\%)$. Both suggest a significant and substantial protective effect of vaccination against COVID-19 infection.  

Nevertheless, the validity of the test-negative design may be questioned because it might not fully control for the unmeasured health-care-seeking behavior. To address this concern, we conduct the proposed sensitivity analysis. Figure \ref{fig:heatmap} presents the lower limits of the 95\% confidence bounds on vaccine efficacy (i.e., $1-\COR_0$) under a range of specifications for the sensitivity parameter  $(\delta, \Gamma, \xi)$, following the framework in Section \ref{sec:sen_effect_heter}. In the heatmap, the intensity of the blue shading represents the lower confidence limits under varying sensitivity parameter values, while the brown contour lines depict combinations of $(\Gamma, \xi)$ at fixed $\delta$ such that the lower limit reaches the annotated values. 
It is worth noting that, as discussed in Section \ref{sec:specify}, 
the effect heterogeneity in Assumption \ref{asmp:effect_heter} is bounded between $1/\min\{\Gamma^4, \xi\}$ and $\min\{\Gamma^4, \xi\}$, implying that $\xi$ plays no role when $\Gamma$ is small.
From Figure \ref{fig:heatmap}, under both selection criteria, vaccination remains highly effective across a broad range of sensitivity parameter values.  

To further inform the choice of sensitivity parameters, we apply the strategy in Section \ref{sec:specify}. Specifically, \citet{Thompson2021} reported aggregate data for subgroups in the test-negative design defined by five covariates, separately: age $\geq 85$ years, presence of at least one chronic respiratory condition, presence of at least one chronic nonrespiratory condition, Black, and Hispanic. We compute the heterogeneity in joint probabilities and odds ratios of the exposure–outcome distribution across the two levels of each covariate, and plot these benchmarks in Figure \ref{fig:heatmap} to calibrate the sensitivity parameters $(\Gamma, \xi)$, which have analogous meanings but for the unmeasured confounder. These values serve as useful references for gauging the potential influence of the unmeasured confounding. For instance, if the influence of unmeasured confounding does not exceed that of the strongest observed covariate, and the test-negative design controls confounding for at least $1-0.3=70\%$ of individuals, then vaccine efficacy remains at least about $88\%$ and $82\%$ under the two selection criteria, respectively.  

The observed data also provide guidance for specifying $\delta$. The proportions of individuals with at least one chronic respiratory condition are $65.97\%$ and $33.78\%$ under the two selection criteria, both likely exceeding the prevalence in the general U.S. population. For example, the prevalence of chronic obstructive pulmonary disease among adults aged 50 years or older was approximately $9.88\%$ in 2020 \citep{Meng2024}. Intuitively, individuals with chronic respiratory conditions are more likely to have positive health-care-seeking behavior, particularly for respiratory infectious diseases like COVID-19. If we posit that all such individuals exhibit positive health-care-seeking behavior (i.e., $U=0$ in our notation), then $\delta$ is at most approximately $0.3$ when hospitalization is used as the selection criterion, which motivates our choice of $\delta$ in Figure~\ref{fig:heatmap}.

\begin{figure}[htb]
    \centering
    \includegraphics[width=1\linewidth]{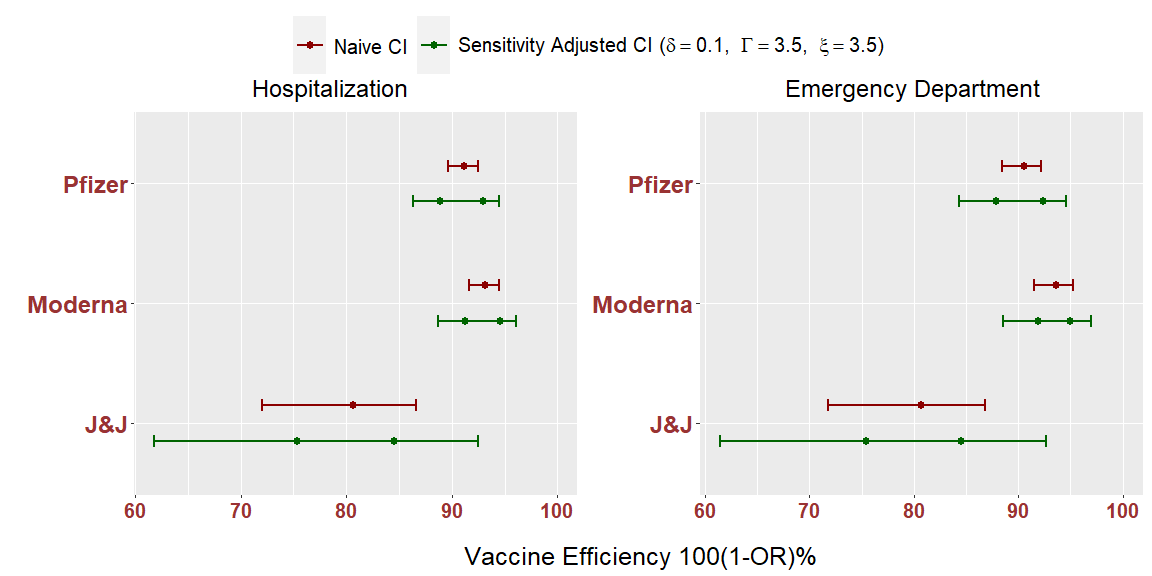}
    \caption{The two panels show analysis results from test-negative designs using two selection criteria: hospitalization and emergency department/urgent care visits, respectively.
Specifically, they display the vaccine efficacy estimates derived from one minus the observed odds ratio (brown points) and the corresponding $95\%$ confidence intervals (brown intervals), assuming perfect test-negative designs.
They also present the causal bounds (green points) and the corresponding $95\%$ confidence-level causal bounds (green intervals) accounting for sampling uncertainty, under the proposed sensitivity analysis with $(\delta, \Gamma, \xi) = (0.1, 3.5, 3.5)$.}
    \label{fig:interval_each_vaccine}
\end{figure}

We next evaluate the effectiveness of full vaccination for each of the three vaccines developed by Pfizer, Moderna, and Johnson \& Johnson. The sensitivity parameters are set to $(\delta, \Gamma, \xi) = (0.1, 3.5, 3.5)$, guided by the calibration from observed covariates described earlier. Figure \ref{fig:interval_each_vaccine} displays the observed odds ratios with their associated confidence intervals under the assumption of a perfect test-negative design, along with the estimated sensitivity bounds and confidence bounds from our sensitivity analysis, under the two selection criteria. The results show that the estimated effects of the three vaccines are highly consistent across the two selection criteria, both in terms of the observed odds ratios and their robustness to unmeasured confounding. This consistency provides stronger evidence for the high efficacy of all three vaccines.

Moreover, the mRNA-based Pfizer and Moderna vaccines demonstrate stronger efficacy than the adenovirus vector-based Johnson \& Johnson vaccine, with significantly different observed odds ratios in the test-negative design. However, this apparent difference can be sensitive to imperfections in the design. In particular, with $95\%$ confidence level and $(\delta, \Gamma, \xi) = (0.1, 3.5, 3.5)$, the true causal efficacy of the Johnson \& Johnson vaccine could range from 
$62\%$
to 
$92\%$
and from 
$61\%$ 
to 
$93\%$
under the two selection criteria, respectively. This wider range is partly due to the smaller sample size for the Johnson \& Johnson vaccine at the time of the study.

\section{Conclusion}\label{sec:conclusion}
 
We proposed sensitivity analysis methods for the test-negative design, which has been used routinely in vaccine studies. 
We focused on the causal odds ratio for units with positive health-care-seeking behavior, 
and considered three sensitivity analysis constraints on the unmeasured confounding: 
(i) the proportion of units with negative health-care-seeking behavior in the design, 
(ii) the difference in the exposure-outcome distributions for units with different health-care-seeking behaviors,
and 
(iii) the causal effect heterogeneity between units with different health-care-seeking behaviors. 
The proposed sensitivity analysis accommodates categorical exposures and outcomes as well as general unmeasured confounders, and it provides confidence bounds that account for sampling uncertainty.

\bibliographystyle{abbrvnat}
\bibliography{paper-ref}

\newpage

\setcounter{equation}{0}
\setcounter{section}{0}
\setcounter{figure}{0}
\setcounter{example}{0}
\setcounter{proposition}{0}
\setcounter{corollary}{0}
\setcounter{theorem}{0}
\setcounter{lemma}{0}
\setcounter{table}{0}
\setcounter{condition}{0}
\setcounter{page}{1}
\setcounter{assumption}{0}
\begin{center}
	\bf \LARGE 
	Supplementary Material 
\end{center}

\renewcommand {\theproposition} {A\arabic{proposition}}
\renewcommand {\theexample} {A\arabic{example}}
\renewcommand {\thefigure} {A\arabic{figure}}
\renewcommand {\thetable} {A\arabic{table}}
\renewcommand {\theequation} {A\arabic{equation}}
\renewcommand {\thelemma} {A\arabic{lemma}}
\renewcommand {\thesection} {A\arabic{section}}
\renewcommand {\thetheorem} {A\arabic{theorem}}
\renewcommand {\thecorollary} {A\arabic{corollary}}
\renewcommand {\thecondition} {A\arabic{condition}}
\renewcommand{\theassumption}{A\arabic{assumption}}

\renewcommand {\thepage} {A\arabic{page}}

\doublespacing

\allowdisplaybreaks

\startcontents[sections]
\printcontents[sections]{l}{1}{\setcounter{tocdepth}{2}}

\section{An alternative identification strategy}\label{sec:alternate identification}

In Section \ref{sec:comp_ident}, we discuss the identification strategy for the causal risk ratio in \citet{Jackson2013} and compare it with the identification strategy for the causal odds ratio. 
Here we discuss the alternative identification strategy for the causal risk ratio in \citet{yu2023test}. 
This alternative identification strategy also requires Assumptions \ref{asmp:U}(i) and \ref{asmp:U_T}, which requires that $U$ is the unmeasured confounder between exposure and outcome and that the test-negative design contains only units with positive health-care-seeking behavior, respectively. 
Moreover, we now view $Y$ as a categorical variable 
with at least 
two levels, 
and focus on the two levels of $Y$ defined in the following way: 
$Y=1$ represents individuals with both symptom and infection, 
and $Y=0$ represents individuals with symptom but without infection; see also the related discussion in Sections \ref{sec:cat_connect} and  \ref{sec:categorical} of the main paper. 
Note that, if a test-negative design contains only symptomatic units, then the outcomes of all units in the design must belong in these two levels. 
Let $I$ denote the binary indicator for infection and $S$ denote the binary indicator for symptom. 
We then have $Y=1$ if and only if $(I, S) = (1,1)$, 
and $Y=0$ if and only if $(I,S)=(0,1)$. 
\citet{yu2023test} are interested in 
    \begin{align*}
        \CRR_{0,c}
        & \equiv 
        \frac{\P(Y(1) = 1 \mid U=0,C=c)}{\P(Y(0) = 1 \mid U=0,C=c)}
        = 
        \frac{\P(I(1)=1, S(1) = 1 \mid U=0,C=c)}{\P(I(0)=1, S(0) = 1 \mid U=0,C=c)}. 
    \end{align*}
They 
imposed the following assumption.
\begin{assumption}\label{asmp:yu2023_asmp}
\begin{enumerate}[label={(\roman*)}, topsep=1ex,itemsep=-0.3ex,partopsep=1ex,parsep=1ex]
    \item  $I \ind T \mid Z, S = 1, U=0, C=c$.
    \item $\P(I=0, S=1 \mid Z=0, U=0, C=c) = \P(I=0, S=1 \mid Z=1, U=0, C=c).$
\end{enumerate}
\end{assumption}

Below we give some intuition for Assumption \ref{asmp:yu2023_asmp}. 
Assumption \ref{asmp:yu2023_asmp}(i) assumes that the infection and the choice of testing is conditionally independent given vaccination status and that the individual has symptom, positive health-care-seeking behavior and observed covariates $C=c$. 
Assumption \ref{asmp:yu2023_asmp}(ii) assumes that given positive health-care-seeking behavior and observed covariates, 
whether an individual is having symptom but not infected with the target virus, or equivalently having symptom due to other viruses, is conditionally independent of the vaccination status.

We will show in the proof of Proposition A1 that
the identification result in \citet{yu2023test} still holds when we replace Assumption \ref{asmp:yu2023_asmp}(i) by the following 
Assumption \ref{asmp:yu2023_alter_asmp}.
It 
is almost the same as Assumption \ref{asmp:U}(ii), noting that $Y$ has been redefined and contains the information in both $I$ and $S$. 
\begin{assumption}\label{asmp:yu2023_alter_asmp}
    $Z \ind T \mid I, S = 1, U=0, C=c$. 
\end{assumption}

Figure \ref{fig:dag_BW} shows a causal diagram to help understanding this setting and the assumptions. Again, the diagram is only for illustration purpose, and the assumptions are stated in terms of potential outcomes and conditional independence. 
Figure \ref{fig:dag_BW}(a) is helpful for understanding Assumption  \ref{asmp:yu2023_asmp}(i). It is similar to \citet[][Figure 1]{yu2023test}, with three slight differences discussed below. 
First, \citet{yu2023test} introduced another variable to denote the reason for testing and considered only units who test for the disease of interest due to symptoms in their identification. 
Figure \ref{fig:dag_BW}(a) can accommodate this consideration by redefining $T=1$ as for units who test due to symptoms. 
Second, under Assumption \ref{asmp:U_T}, $T=1$ implies that $U=0$. Moreover, with the redefined $T$, $T=1$ also implies that $S=1$. Thus, similar to \citet[][Figure 1]{yu2023test}, we can put boxes around $T=1$, $U=0$ and $S=1$. 
Third, we have the direct path from $U$ to $S$, which is excluded in \citet[][Figure 1]{yu2023test}. 
Figure \ref{fig:dag_BW}(b) is helpful for understanding Assumption \ref{asmp:yu2023_alter_asmp}. It is almost the same as Figure \ref{fig:dag_BW}(a), except that $Z$ has no direct edge to $T$ and $I$ can have a direct edge to $T$.

\begin{figure}[htb]
	\begin{subfigure}[htbp]{0.5\textwidth}
		\centering
		\includegraphics[width=0.8\textwidth]{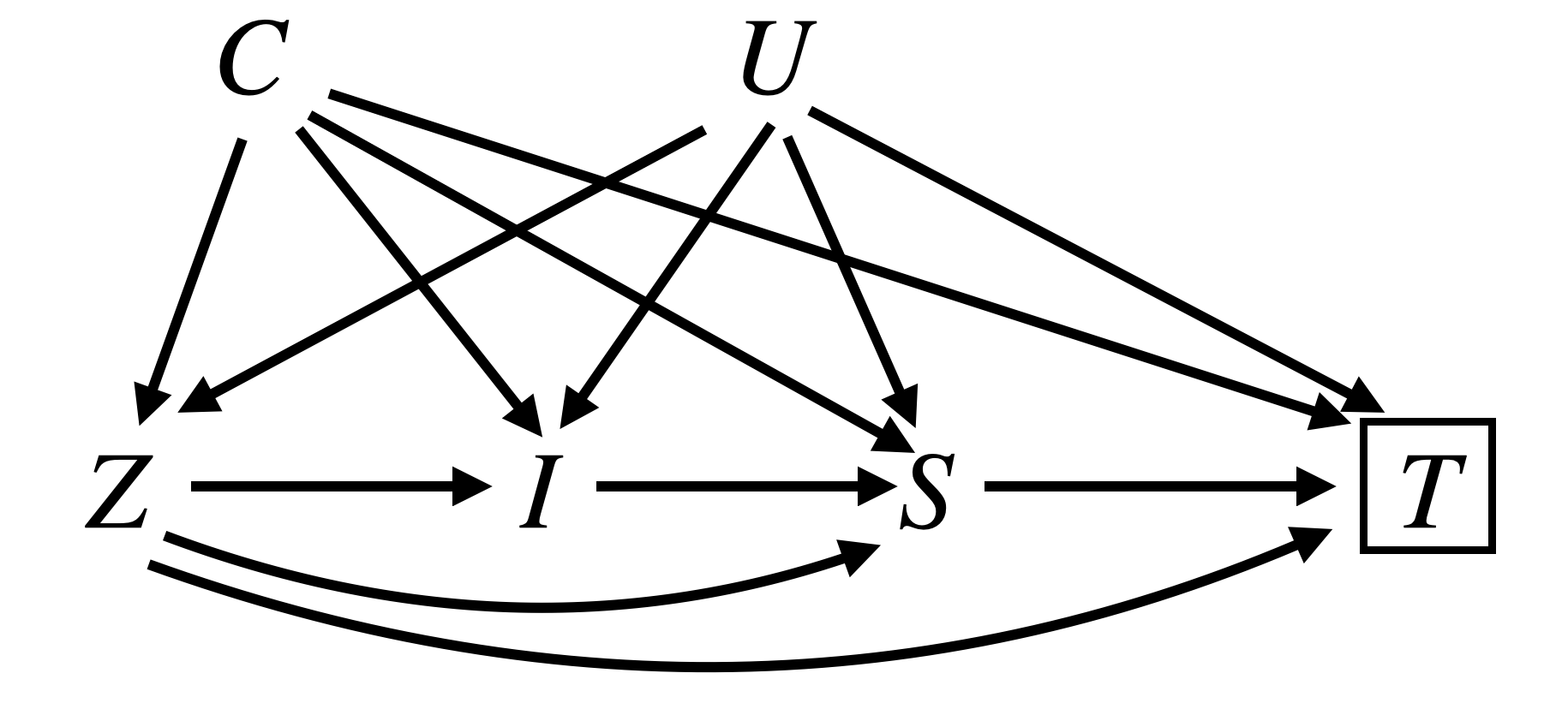}
		\caption{Assumption \ref{asmp:yu2023_asmp}(i)}
	\end{subfigure}%
	\begin{subfigure}[htbp]{0.5\textwidth}
		\centering
		\includegraphics[width=0.8\textwidth]{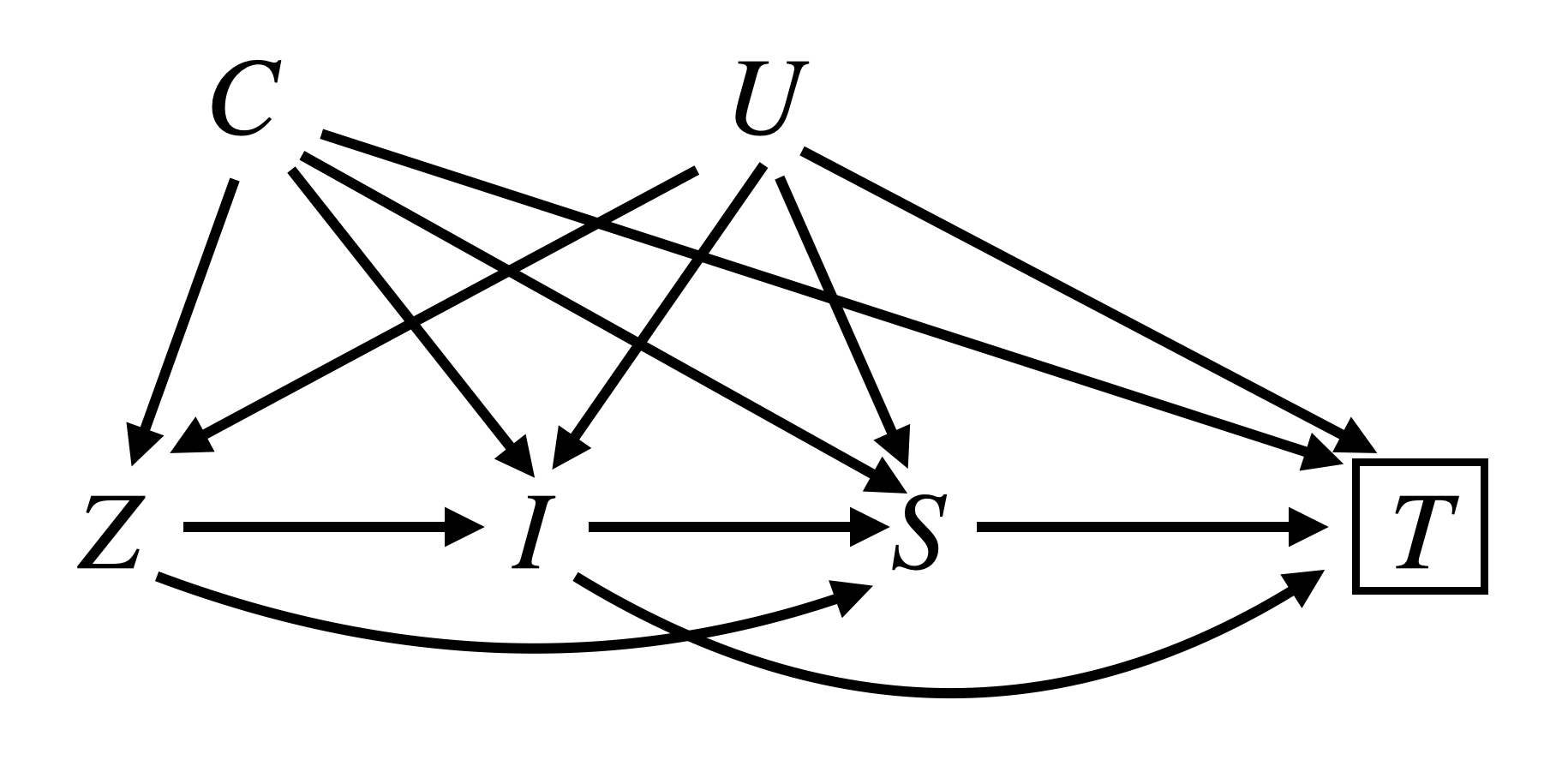}
		\caption{Assumption \ref{asmp:yu2023_alter_asmp}}
	\end{subfigure}
	\caption{
		Causal diagrams for the test-negative design. 
		$Z$ denotes the exposure, $I$ denotes the infection, $S$ denotes the symptom, $C$ denotes the observed covariates, $U$ denotes the unmeasured confounder, 
		and $T$ denotes the testing indicator. 
		We put a box around $T$ since the design involves selection on $T=1$. 
	}\label{fig:dag_BW}
\end{figure}

The proposition below shows the identification of the causal risk ratio for getting symptomatic infection. 
The identification result in (b) is similar to that discussed in Section \ref{sec:cat_connect}.
We give the proof below. 
\begin{proposition}\label{prop:Yu_iden}
	\begin{enumerate}[label={(\alph*)}, topsep=1ex,itemsep=-0.3ex,partopsep=1ex,parsep=1ex]
		\item Under Assumptions \ref{asmp:U}(i), \ref{asmp:U_T} and \ref{asmp:yu2023_asmp}, 
		we have $\CRR_{0,c} = \OR_c$. 
		
		\item     Under Assumptions \ref{asmp:U}(i), \ref{asmp:U_T}, \ref{asmp:yu2023_asmp}(ii)
        and \ref{asmp:yu2023_alter_asmp},
        we have  
		$\CRR_{0,c} = \OR_c$. 
\end{enumerate}

\end{proposition}

\begin{proof}[Proof of Proposition \ref{prop:Yu_iden}]
We have
\begin{align*}
    & \ \quad \OR_c\\
    & = \frac{\P(Y=1, T=1 \mid Z=1, U=0,C=c)/\P(Y = 0, T=1 \mid Z=1, U=0,C=c)}{\P(Y=1, T=1 \mid Z=0, U=0,C=c)/\P(Y = 0, T=1 \mid Z=0, U=0,C=c)}
    \\
    & \quad \ \text{(\textit{by Assumption \ref{asmp:U_T} and the same logic as the proof of Proposition \ref{prop:iden_jackson}})}
    \\
    & = \frac{\P(I=1, S = 1, T=1 \mid Z=1, U=0, C=c)/\P(I = 0, S=1, T=1 \mid Z=1, U=0, C=c)}{\P(I=1, S=1, T=1 \mid Z=0, U=0, C=c)/\P(I = 0, S=1, T=1 \mid Z=0, U=0, C=c)}\\
    & \quad \ \text{(\textit{by the definition of $Y$ in this context})}\\
    & = \frac{\P(I=1, S = 1 \mid Z=1, U=0, C=c)/\P(I = 0, S=1 \mid Z=1, U=0, C=c)}{\P(I=1, S=1 \mid Z=0, U=0, C=c)/\P(I = 0, S=1 \mid Z=0, U=0, C=c)}
    \\
    & \quad \ \text{(\textit{by Assumption \ref{asmp:yu2023_asmp}(i) or \ref{asmp:yu2023_alter_asmp}})}\\
    & = \frac{\P(I=1, S = 1 \mid Z=1, U=0, C=c)}{\P(I=1, S=1 \mid Z=0, U=0, C=c)}\\
    & \quad \ \text{(\textit{by Assumption \ref{asmp:yu2023_asmp}(ii)})}\\
    & = \frac{\P(I(1)=1, S(1) = 1 \mid U=0,C=c)}{\P(I(0)=1, S(0) = 1 \mid U=0,C=c)}\\
    & \quad \ \text{(\textit{by Assumption \ref{asmp:U}(i)})}\\
    & = \CRR_{0,c}, \\
    & \quad \ \text{(\textit{by definition})}
\end{align*}
where the reason for each equality is explained in the parentheses. 
\end{proof}

Assumption \ref{asmp:yu2023_asmp}(ii) is the key for the identification strategy in \cite{yu2023test}. 
Similar to \cite{Jackson2013}, it assumes that, among units with positive health-care-seeking behavior, 
the probability of experiencing symptoms from infections of viruses other than the one under consideration is independent of vaccination status. 
Similar to the discussion in Section \ref{sec:comp_ident}, Assumption \ref{asmp:yu2023_asmp}(ii) may also fail when there is virus interference.

{\rev 
\section{Additional numerical results}

\subsection{Varying causal bounds under fixed odds ratio}\label{sec:numerical}

\begin{figure}
     \flushleft
     \begin{subfigure}{0.4\textwidth}
         \includegraphics[width=1.2\linewidth]{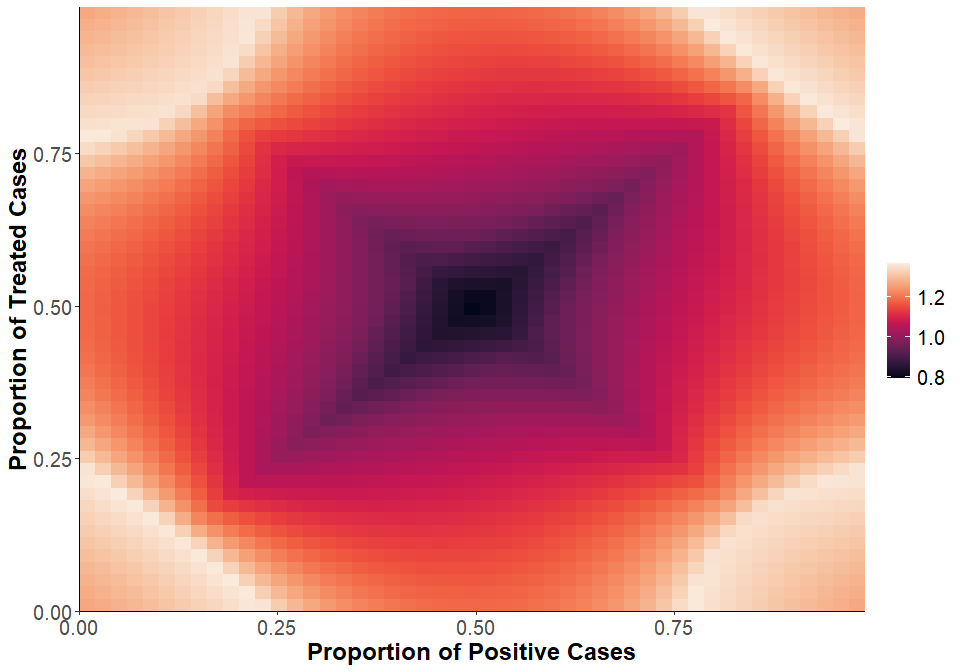}
         \caption{$(\delta,\Gamma,\xi) = (0.1, 5, \infty)$}
         \label{fig: heatmap same odds ratio xi=inf}
     \end{subfigure}
     \hspace{2.5cm}
     \begin{subfigure}{0.4\textwidth}
         \includegraphics[width=1.2\linewidth]{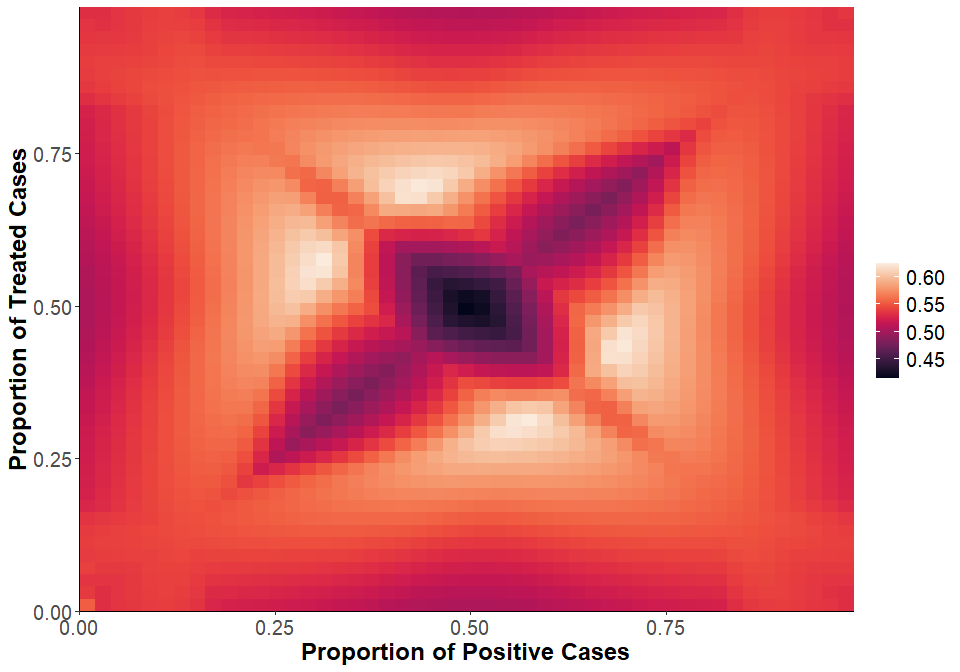}
         \caption{$(\delta,\Gamma,\xi) = (0.1, 5, 2)$}
         \label{fig: heatmap same odds ratio xi=2}
     \end{subfigure}
        \caption{\rev Each plot shows the logarithm of the difference between the upper and lower bounds for various values of \(\pi_{zy\mid c}s\), while keeping the odds ratio constant. The x-axis denotes the proportion of positive cases (\(\pi_{01\mid c}+\pi_{11\mid c}\)), and the y-axis represents the proportion of treated cases (\(\pi_{10\mid c} + \pi_{11\mid c}\)).}
        \label{fig: heatmap same odds ratio}
\end{figure}

As discussed in Sections \ref{sec:sen_U_not_0} and \ref{sec:integrate}, the causal bounds on $\COR_{0c}$ from the proposed sensitivity analysis depend on the observed data distribution 
$\pi_{zy\mid c}$s, not merely the observed odds ratio $\OR_c$ that is often used to estimate the causal effect. 
We conduct a numerical experiment to further illustrate this fact. 
We show 
that different contingency tables with the same odds ratio can lead to quite different causal bounds 
under the proposed sensitivity analysis.

Fixing the odds ratio ($\OR_c$) along with the fact that the proportions $\pi_{zy\mid c}$s should sum to $1$ gives us the following two constraints:
\begin{align*}
    \pi_{01\mid c} + \pi_{10\mid c} + \pi_{00\mid c} + \pi_{11\mid c} = 1 \quad \text{and} \quad \frac{\pi_{00}\pi_{11}}{\pi_{10}\pi_{01}} = \OR_c.
\end{align*}
Thus, when the observed odds ratio is fixed, we have two free variables in the $\pi_{zy\mid c}$s, and the resulting causal bounds could vary as the contingency table changes. This phenomenon is visually depicted in Figure \ref{fig: heatmap same odds ratio}. The x-axis represents the proportion of treated cases ($\pi_{10\mid c} + \pi_{11\mid c}$), and the y-axis represents the proportion of positive cases ($\pi_{01\mid c} + \pi_{11\mid c}$). Each point on the graph corresponds to a unique contingency table with a fixed odds ratio $\OR_c=0.5$. 
The heatmap shows the difference between the logarithms of the upper and lower bounds for each of these contingency tables. The two plots in Figure \ref{fig: heatmap same odds ratio} correspond to two distinct values of sensitivity parameters.

From Figure \ref{fig: heatmap same odds ratio}, the causal bounds 
tend to be wider near the corners of the plots, 
i.e., when most of the population are treated or when most of the population test positive.
As a side note, 
the comparison between Figures \ref{fig: heatmap same odds ratio}(a) and (b) also shows the importance of the effect heterogeneity constraint quantified by the parameter $\xi$. Given the same choice of $\delta$ and $\Gamma$, decreasing $\xi$ from $\infty$ to $2$ substantially shrinks the bounds.

\subsection{Additional results for the Covid-19 vaccines}

Figure \ref{fig:vaccine_hos} presents the lower limits of the $95\%$ confidence bounds for the efficacy of full two-dose mRNA vaccination across different values of $(\Gamma, \xi)$ at $\delta = 0.1$, using hospitalization as the selection criterion.
Each panel in Figure \ref{fig:vaccine_hos} corresponds to a specific covariate stratum, as indicated at the top of each plot.
Figure \ref{fig:vaccine_emergency} shows the analogous results using emergency department/urgent care visits as the selection criterion.
In both figures, benchmarks for the sensitivity parameters are provided in the same way as in Figure \ref{fig:heatmap}.
From both figures, vaccination shows substantial protective effects. 
If the influence of unmeasured confounding does not exceed that of the strongest observed covariate, and the test-negative design controls for confounding in at least $90\%$ of individuals, then the vaccine efficacy remains at least about $80\%$ under the two selection criteria across all covariate strata. 
The two exceptions are the strata corresponding to Black and Hispanic individuals when using emergency department/urgent care clinic visits as the selection criterion, likely due in part to their smaller sample sizes in the design.

\begin{figure}
    \centering
    \includegraphics[width=0.8\linewidth]{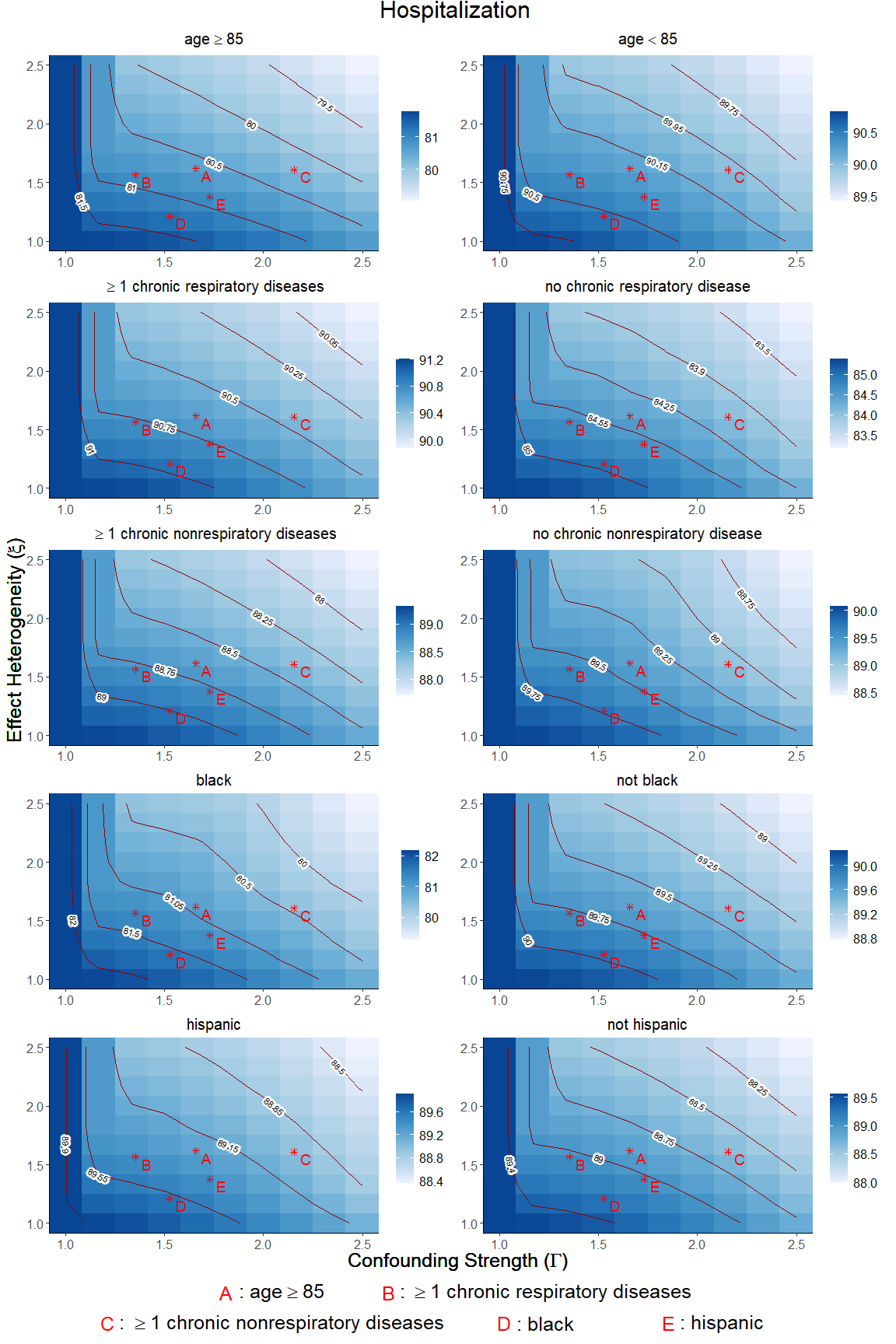}
    \caption{\rev Each panel presents analogous results as Figure \ref{fig:interval_each_vaccine} with $\delta=0.1$, but focusing on a particular covariate stratum (as indicated at the top of each plot). The points ``A''--``E'' are benchmark values for the sensitivity parameters. They are calculated using the observed covariates in the same way as in Figure \ref{fig:heatmap}. The results here all use hospitalization as the selection criterion.}
    \label{fig:vaccine_hos}
\end{figure}

\begin{figure}
    \centering
    \includegraphics[width=0.8\linewidth]{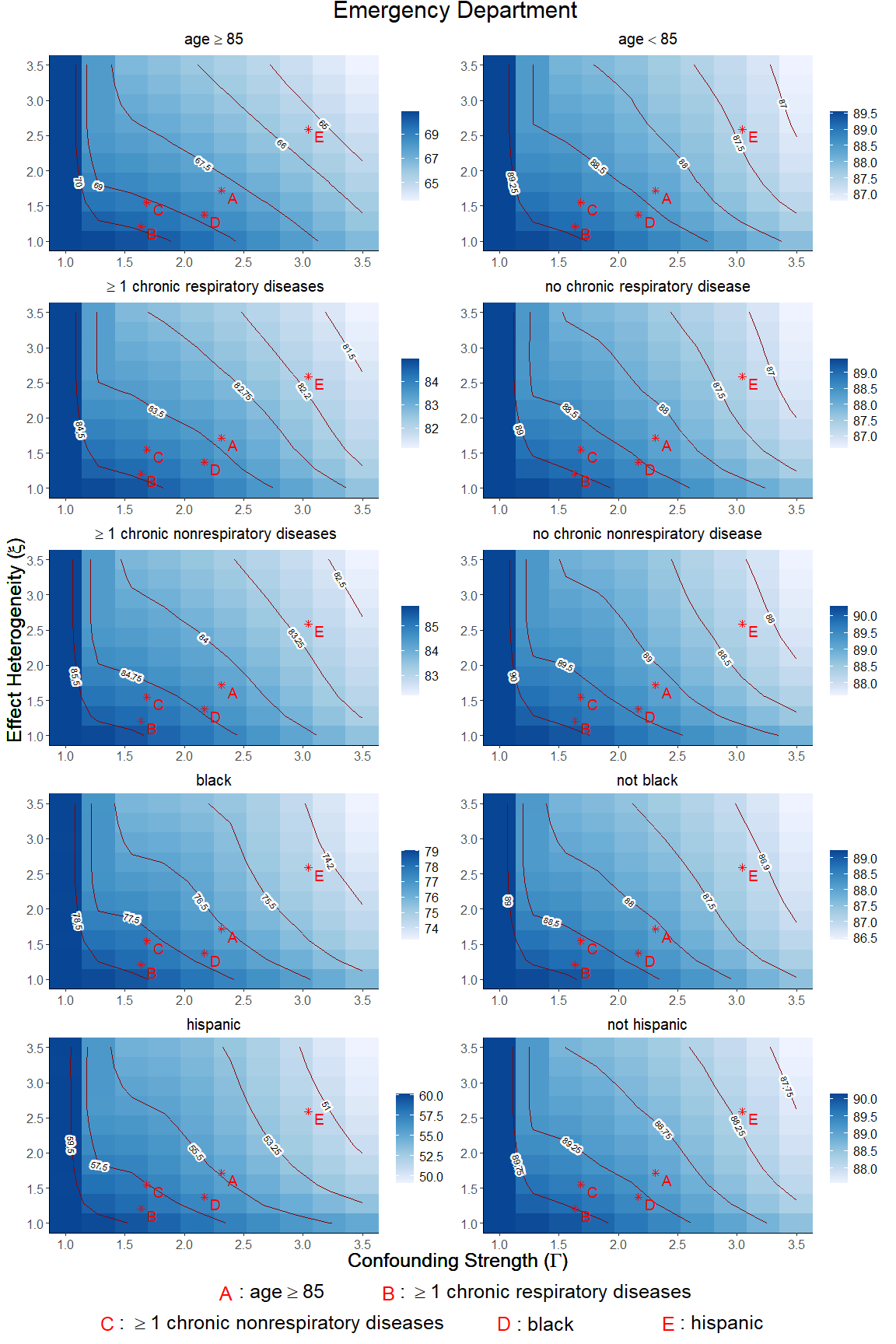}
    \caption{\rev This figure presents results analogous to those in Figure \ref{fig:vaccine_hos}, but using emergency department/urgent care visits as the selection criterion.}
    \label{fig:vaccine_emergency}
\end{figure}

}

\section{Confidence bounds}\label{sec: confidence bounds}

In this section, we study confidence bounds for the causal odds ratios under the proposed sensitivity analysis. For descriptive convenience, we focus on the case with binary exposure, outcome and unmeasured confounder. 
By the same logic as in Sections \ref{sec:categorical} and \ref{sec:cat_3_constraints}, these results can also be applied to settings with categorical exposure and outcome, and a general unmeasured confounder.

In Sections \ref{sec:cs_o} and \ref{sec:con_gen_c}, we discuss confidence sets for the true $\pi_{zy\mid c}$s with discrete and general observed covariates. 
In Section \ref{sec:closed_ci}, we construct closed-form confidence bounds for the causal odds ratio under the sensitivity analysis in Theorem \ref{thm:cor_bound_additional} with Assumptions \ref{asmp:bound_u_not_0} and \ref{asmp:bound_density}. 
In Section \ref{sec:cb_qp}, 
we construct general confidence bounds using quadratic programming under the sensitivity analysis in Theorem \ref{thm:sen_3_constr} with Assumptions \ref{asmp:bound_u_not_0}, \ref{asmp:bound_density} and \ref{asmp:effect_heter}. 
In Section \ref{sec:comp_cb}, we compare different forms of causal bounds under the proposed sensitivity analysis. 
In Section \ref{sec:num_cont}, we conduct a numerical illustration for the simultaneous causal bounds with a continuous observed covariate. 

\subsection{Confidence sets for the $\pi_{zy\mid c}$s with discrete covariates $C$}\label{sec:cs_o}

In this subsection, we consider the setting where the observed covariates $C$ are discrete and analyze the causal effects of the exposure within each level (or stratum) of $C$. Focusing on a particular level $C=c$, the observed data typically consist of counts in a $2\times 2$ contingency table defined by the two levels of exposure and outcome. These counts are commonly assumed to follow a Multinomial distribution with (unknown) probabilities $\pi_{zy\mid c}$. Specifically, let $n_c$ denote the number of units with observed covariate value $c$ under the test-negative design, and let $\hat{\pi}_{zy\mid c}$ denote the corresponding observed proportion of units with exposure $z$ and outcome $y$, for $z,y\in \{0,1 \}$.

\begin{assumption}\label{asmp:Multinomial}
    $n_c \times (\hat{\pi}_{00\mid c}, \hat{\pi}_{10\mid c}, \hat{\pi}_{01\mid c}, \hat{\pi}_{11\mid c}) \sim \textrm{Multinomial}(n_c, (\pi_{00\mid c}, \pi_{10\mid c}, \pi_{01\mid c}, \pi_{11\mid c})).$
\end{assumption}

Let $\hat{\pi}_{**\mid c}=(\hat{\pi}_{00\mid c}, \hat{\pi}_{10\mid c}, \hat{\pi}_{01\mid c}, \hat{\pi}_{11\mid c})^\top$ and $\pi_{**\mid c} = (\pi_{00\mid c}, \pi_{10\mid c}, \pi_{01\mid c}, \pi_{11\mid c})^\top$. 
From Assumption \ref{asmp:Multinomial} and by the standard central limit theorem, we have 
\begin{equation}\label{eq:clt}
\sqrt{n_c}(\hat{\pi}_{**\mid c} - \pi_{**\mid c}) \converged N\left(0, \Sigma_c \right), \quad \text{where } \  \Sigma_c = \textrm{diag}(\pi_{**\mid c}) - \pi_{**\mid c}\pi_{**\mid c}^\top.
\end{equation}
We consider the following three $1-\alpha$ asymptotic confidence sets for $\pi_{**\mid c}$: 
\begin{align}
    \mathcal{S}_{c, \alpha}^{\textsc{Q}} & = \{
    q_{**} : 
    n(\hat{\pi}_{**\mid c} - q_{**}) \hat{\Sigma}^{\dagger}_c (\hat{\pi}_{**\mid c} - q_{**}) \leq \chi^2_{3, 1-\alpha}
    \}, 
    \label{eq:C_Q}
    \\
    \mathcal{S}_{c, \alpha}^{\textsc{N}} & = \{
    q_{**}: 
    \{\hat{\pi}_{zy\mid c}(1-\hat{\pi}_{zy\mid c})\}^{-1/2}
    \sqrt{n}|\hat{\pi}_{zy\mid c} - q_{zy}| \le \hat{d}_{1-\alpha} 
    \text{ for } z,y\in\{ 0,1 \}
    \}, 
    \label{eq:C_N}
    \\
    \mathcal{S}_{c, \alpha}^{\textsc{T}} & = \{
    q_{**}: 
    \sqrt{n}| \sin^{-1}(2\hat{\pi}_{zy\mid c}-1) - \sin^{-1}(2q_{zy}-1) |
    \le \hat{d}_{1-\alpha}
    \text{ for } z,y\in\{ 0,1 \}
    \},
    \label{eq:C_T}
\end{align}
where $q_{**} = (q_{00}, q_{10}, q_{01}, q_{11})^\top$, and $\hat{\Sigma}^{\dagger}_c, \chi^2_{3, 1-\alpha}$ and $\hat{d}_{1-\alpha}$ will be defined below. 
The confidence set in \eqref{eq:C_Q} is quadratic in $\pi_{zy\mid c}$s, while the other two are linear in the $\pi_{zy\mid c}$s.  In particular, $\mathcal{C}_{c,\alpha}^{\textsc{Q}}$ is an elliptical set,  whereas $\mathcal{C}_{c,\alpha}^{\textsc{N}}$ and $\mathcal{C}_{c,\alpha}^{\textsc{T}}$ are rectangular sets. 

Below we explain these confidence sets and their asymptotic validity in detail. 
The asymptotic validity of the confidence set $\mathcal{C}_{c,\alpha}^{\textsc{Q}}$ in \eqref{eq:C_Q} follows from the asymptotic Gaussian distribution in \eqref{eq:clt}. 
Specifically, 
$\hat{\Sigma}_c$ is a consistent estimator of $\Sigma_c$ with $\pi_{zy\mid c}$s estimated by $\hat{\pi}_{zy\mid c}$s, $\hat{\Sigma}^{\dagger}_c$ denotes its pseudoinverse, 
and $\chi^2_{3, 1-\alpha}$ denotes the $1-\alpha$ quantile of the $\chi^2$ distribution with degrees of freedom \(3\). 

The asymptotic validity of the confidence sets in \eqref{eq:C_N} and \eqref{eq:C_T} follows from, respectively, the following asymptotic distributions: 
\begin{align*}
    \sqrt{n} 
    \begin{pmatrix}
        (\hat{\pi}_{00\mid c} - \pi_{00\mid c})/\sqrt{\pi_{00\mid c}(1-\pi_{00\mid c})}\\
        (\hat{\pi}_{10\mid c} - \pi_{10\mid c})/\sqrt{\pi_{10\mid c}(1-\pi_{10\mid c})}\\
        (\hat{\pi}_{01\mid c} - \pi_{01\mid c})/\sqrt{\pi_{01\mid c}(1-\pi_{01\mid c})}\\
        (\hat{\pi}_{11\mid c} - \pi_{11\mid c})/\sqrt{\pi_{11\mid c}(1-\pi_{11\mid c})}
    \end{pmatrix}
    (\hat{\pi}_{**\mid c} - \pi_{**\mid c}) \converged N\left(0, 
    \Omega_c
    \right),
\end{align*}
and 
\begin{align*}
    \sqrt{n}
    \begin{pmatrix}
        \sin^{-1}(2\hat{\pi}_{00\mid c} - 1) - \sin^{-1}(2\pi_{00\mid c} - 1)\\
        \sin^{-1}(2\hat{\pi}_{10\mid c} - 1) - \sin^{-1}(2\pi_{10\mid c} - 1)\\
        \sin^{-1}(2\hat{\pi}_{01\mid c} - 1) - \sin^{-1}(2\pi_{01\mid c} - 1)\\
        \sin^{-1}(2\hat{\pi}_{11\mid c} - 1) - \sin^{-1}(2\pi_{11\mid c} - 1)
    \end{pmatrix}
    \converged
    \mathcal{N}(0, \Omega_c), 
\end{align*}
where the common covariance matrix $\Omega_c$ has the 
diagonal elements being $1$ and the off-diagonal elements being 
$-\sqrt{\pi_{zy\mid c} \pi_{z'y'\mid c}/\{(1-\pi_{zy\mid c})(1-\pi_{z'y'\mid c})\}}$. 
Let $\hat{\Omega}_c$ be an estimator of $\Omega_c$ with $\pi_{zy\mid c}$s estimated by $\hat{\pi}_{zy\mid c}$s. 
The threshold $\hat{d}_{1-\alpha}$ is the $1-\alpha$ quantile of $\max_i |\xi_i|$, 
where $(\xi_1, \xi_2, \xi_3, \xi_4)^\top$ follows a Gaussian distribution with mean zero and covariance matrix $\hat{\Omega}_c$.

\begin{remark}
    When the causal odds ratio can be identified using the observed odds ratio, as in Propositions \ref{prop:cor_perfect_test} and \ref{prop:iden_no_conf_strength},  we can consistently estimate $\COR_{0\mid c} = \OR_{c}$ using $\widehat{\OR}_c = (\hat{\pi}_{11\mid c} \hat{\pi}_{00\mid c})/(\hat{\pi}_{10\mid c} \hat{\pi}_{01\mid c})$, which has the following asymptotic distribution: 
    \begin{align}\label{eq:or_clt}
    \sqrt{n}\{
    \log(\widehat{\OR_c}) - \log (\OR_c)
    \}
    \converged 
    \mathcal{N}\left(0, \ \sum_{z,y=0}^1 \pi_{zy\mid c}^{-1} \right).
    \end{align}
The above asymptotic distribution follows from \eqref{eq:clt} and the Delta method. 
We can then construct Wald-type confidence intervals based on \eqref{eq:or_clt}, with the asymptotic variance estimated by $\sum_{z,y=0}^1 \hat{\pi}_{zy\mid c}^{-1} $. 
\end{remark}

\begin{remark}
We can adjust the critical values in the confidence sets in \eqref{eq:C_Q}–\eqref{eq:C_T} to obtain confidence sets for $\pi_{zy\mid c}$ that are simultaneously valid across all levels of $C$, which in turn yield simultaneously valid confidence sets for $\COR_{0c}$ across all levels of $C$. Conditional on the observed covariates for all units in the design, $(\hat{o}_{00\mid c}, \hat{\pi}_{10\mid c}, \hat{\pi}_{01\mid c}, \hat{\pi}_{11\mid c})$ is mutually independent across different levels of $C$. Consequently, the $(1-\alpha)^{1/K}$ confidence set for $(\hat{o}_{00\mid c}, \hat{\pi}_{10\mid c}, \hat{\pi}_{01\mid c}, \hat{\pi}_{11\mid c})$ achieves simultaneous coverage probability $1-\alpha$, where $K$ denotes the number of levels of $C$.
\end{remark}

{\rev 
\subsection{Confidence sets for the $\pi_{zy\mid c}$s with general covariates $C$}\label{sec:con_gen_c}

In this subsection we consider general observed covariates $C\in \mathbb{R}^p$, which, for example, can include continuous covariates. Again, we will analyze the causal effects of the exposure within each level (or stratum) of $C$. To estimate or infer the conditional probabilities $\pi_{zy\mid c}$s within each level of $C$, we will impose some parametric model linking the probabilities $\pi_{zy\mid c}$s to the covariate value $c$. Specifically, we consider the following general parametric model:
\begin{align}\label{eq:gen_model_pi}
    \pi_{**\mid c} \equiv
    \left( 
        \pi_{00\mid c}, 
        \pi_{01\mid c}, 
        \pi_{10\mid c}, 
        \pi_{11\mid c}
    \right)^\top 
     = g(\beta, c)^\top, 
\end{align}
where $\beta \in \mathbb{R}^m$ denotes the unknown coefficient vector and $g$ is a map from $\mathbb{R}^m \times \mathbb{R}^p$ to $\mathbb{R}^4$. A common example is the Multinomial logistic regression with  
\begin{align}\label{eq:mnlogit}
    g(\beta,  c) 
    = 
    \left\{1 +  \exp( \tilde{c}^\top \beta_{01} ) + \exp( \tilde{c}^\top \beta_{10} ) + \exp( \tilde{c}^\top \beta_{11} ) \right\}^{-1}
    \begin{pmatrix}
    1\\
    \exp( \tilde{c}^\top \beta_{01} )\\
    \exp( \tilde{c}^\top \beta_{10} )\\
    \exp( \tilde{c}^\top \beta_{11} )
    \end{pmatrix}, 
\end{align}
where $\tilde{c} = (1, c^\top)^\top$ and $\beta = (\beta_{01}^\top, \beta_{10}^\top, \beta_{11}^\top)^\top$. 

Suppose we have an estimator $\hat{\beta}_n$ based on observed data of size $n$, obtained for example by the maximum likelihood estimation. Furthermore, we assume $\hat{\beta}_n$ follows an asymptotic Gaussian distribution and its asymptotic covariance matrix can be consistently estimated; these can often be justified by the standard asymptotic theory. 
We can then apply the Delta method to obtain the asymptotic Gaussian distribution of $\hat{\pi}_{**\mid c}$s and consistently estimate the corresponding asymptotic covariance matrix. 
Finally, we can construct confidence sets for $\pi_{**\mid c}$ in a similar way as that in Section \ref{sec:cs_o}, which can further lead to confidence bounds for the true causal effect under our sensitivity analysis. Such confidence bounds from our sensitivity analysis will be point-wise valid at any given covariate value $c$.

Below we further consider simultaneous confidence sets for the $\pi_{**\mid c}$s across all possible values of $c$, which can lead to simultaneous confidence bounds for the causal effects $\COR_{0c}$s under our sensitivity analysis. 
The following theorem constructs simultaneous confidence sets for the $\pi_{**\mid c}$s under some regularity conditions on the parametric model $g$ in \eqref{eq:gen_model_pi}. 

\begin{theorem}\label{thm:uniform_CI_model_pi}
Consider a parametric model in \eqref{eq:gen_model_pi} for the joint distribution of exposure and outcome conditional on the observed covariates. 
Suppose that 
\(\hat{\beta}_n\) is an estimator of the true parameter \(\beta_0\in \mathbb{R}^m\) based on samples of size $n$ and satisfies that 
$
\sqrt{n}(\hat{\beta}_n - \beta_0) 
\converged
\mathcal{N}(0,\Sigma)
$
for some positive definite matrix \(\Sigma\), where the asymptotic covariance matrix \(\Sigma\) can be consistently estimated by $\hat{\Sigma}_n$.
Let $\mathcal{C}$ be the set of possible covariate values, and $g_i$ denote the $i$th element of $g$ for $1\le i \le 4$.  
Assume that the following regularity conditions hold: 
\begin{itemize}
    \item[(i)] $g_i(c, \beta)$ is twice differentiable in \(\beta \in \mathbb{R}^m\) for all \(c\in \mathcal{C}\) and $1\le i \le 4$, 
    
    \item[(ii)] the nonzero singular values of $D_c \equiv \frac{\partial g(\beta, c)}{\partial \beta} \mid_{\beta = \beta_0} \in \mathbb{R}^{4\times m}$ are uniformly bounded above and below for all $c\in \mathcal{C}$,

    \item[(iii)] for all $1\le i \le 4$ and a sufficiently small $\eta > 0$ , the largest singular value of $\frac{\partial^2 g_i(\beta, c)}{\partial \beta \partial \beta^\top }$ is uniformly bounded for all $c\in \mathcal{C}$ and all $\beta \in \mathbb{R}^m$ such that $\|\beta - \beta_0\|\le \eta$. 
\end{itemize}
Then, for any \(\alpha\in (0,1)\),
\begin{align*}
    \mathcal{S}_{c, \alpha}^g 
    = 
    \left\{ 
    q_{**}\in \mathbb{R}^4: 
    \{ g(\hat{\beta}_n, c) - q_{**} \}^\top \left(D_c \hat{\Sigma}_n D_c^\top  \right)^{\dagger} \{ g(\hat{\beta}_n, c) - q_{**}\} \leq n^{-1} \chi^2_{m, 1-\alpha}
    \right\}
\end{align*}
is a $1-\alpha$ simultaneous confidence set for 
$\pi_{**\mid c} \equiv ( \pi_{00\mid c}, 
\pi_{01\mid c}, 
\pi_{10\mid c}, 
\pi_{11\mid c} )^\top$ across all $c\in \mathcal{C}$, in the sense that 
$
    \liminf_{n\rightarrow \infty} \Pr\{ \pi_{**\mid c} \in \mathcal{S}_{c, \alpha}^g \text{ for all } c \in \mathcal{C} \} \ge 1-\alpha, 
$
where 
\(\chi^2_{m, 1-\alpha}\) denotes the \((1-\alpha)\)th quantile of the chi-squared distribution with degrees of freedom $m$, and $A^\dagger$ denotes the pseudoinverse of a matrix $A$.   
\end{theorem}

The regularity conditions on $g(\beta, c)$ in Theorem \ref{thm:uniform_CI_model_pi} can be simplified when $g(\beta, c)$ represents the usual Multinomial logistic regression model as in \eqref{eq:mnlogit}. We summarize it in the following proposition. 

\begin{proposition}\label{prop:multinom_cond}
    If $g(\beta, c)$ takes the form in \eqref{eq:mnlogit}, 
    and the covariates are uniformly bounded over $\mathcal{C}$, then conditions (i)--(iii) in Theorem \ref{thm:uniform_CI_model_pi} must hold.
\end{proposition}

}

\subsection{Closed-form confidence bounds}\label{sec:closed_ci}

We consider sensitivity analysis in Theorem \ref{thm:cor_bound_additional}, which involves Assumptions \ref{asmp:bound_u_not_0} and \ref{asmp:bound_density} for the unmeasured confounding. This also covers Theorems \ref{thm:simple_bound_cor} and \ref{thm:cor_bound_density} as special cases. 

To facilitate the computation, we first consider confidence sets for the true $\pi_{zy\mid c}$s of the following rectangular form: 
\begin{align}\label{eq:conf_C}
    \mathcal{S}_{c, \alpha} = \{(q_{00}, q_{10}, q_{01}, q_{11}): \hat{\underline{\pi}}_{zy\mid c} \le q_{zy} \le \hat{\overline{\pi}}_{zy \mid c} \text{ for } z,y=0,1
    \},
\end{align}
where $0\le \hat{\underline{\pi}}_{zy\mid c}  \le \hat{\overline{\pi}}_{zy \mid c} \le 1$ for all $z,y$ and they depend on the observed data and the desired confidence level. 
For example, $\mathcal{S}_{c, \alpha}$ can be the confidence set in either \eqref{eq:C_N} or \eqref{eq:C_T}.

\begin{theorem}\label{thm:confidence_bound}
Assume Assumptions \ref{asmp:U}, \ref{asmp:bound_u_not_0},  \ref{asmp:bound_density} and \ref{asmp:Multinomial}
for some $\delta_c \in [0,1]$ and $\Gamma_c \ge 1$, 
and let $\mathcal{S}_{c, \alpha}$ in \eqref{eq:conf_C} be an asymptotic $1-\alpha$   confidence set for the true $\pi_{**\mid c} = (\pi_{00\mid c}, \pi_{10\mid c}, \pi_{01\mid c}, \pi_{11\mid c})$. 
Define 
\begin{align}\label{eq:lu_hat}
    \hat{l}_{zy \mid c} \equiv \max \left\{ \frac{\hat{\underline{\pi}}_{zy \mid c}}{\delta_c\Gamma_c+(1-\delta_c)},  \  \frac{\hat{\underline{\pi}}_{zy \mid c}-\delta_c}{1-\delta_c} \right\}, 
    \ \ 
    \hat{u}_{zy\mid c} \equiv \min \left\{\frac{\hat{\overline{\pi}}_{zy\mid c}  \Gamma_c }{\delta_c+(1-\delta_c)\Gamma_c}, \  1 \right\}, 
    \ \ 
    (z,y\in\{ 0,1 \})
\end{align}
where $\hat{\underline{\pi}}_{zy\mid c}$s and $\hat{\overline{\pi}}_{zy \mid c}$s are from the confidence set in \eqref{eq:conf_C}. 
Define the lower and upper bounds $\widehat{\underline{\COR}}_{0c}$ and  $\widehat{\overline{\COR}}_{0c}$ in the same way as in Theorem \ref{thm:cor_bound_density} but with $l_{zy\mid c}$s and $u_{zy\mid c}$s there replaced by $\hat{l}_{zy\mid c}$s and $\hat{u}_{zy\mid c}$s in \eqref{eq:lu_hat}. 
Then $[\widehat{\underline{\COR}}_{0c}, \widehat{\overline{\COR}}_{0c}]$ is an asymptotic $1-\alpha$ confidence interval for the true causal odds ratio $\COR_{0c}$. 
\end{theorem}

\begin{remark}\label{rmk:simultaneous1}
    In Theorem \ref{thm:confidence_bound} as well as the later Theorem \ref{thm:cb_sen_3_constr}, if further $\mathcal{S}_{c, \alpha}$ is a simultaneous asymptotic $1-\alpha$ confidence set for the true $\pi_{**\mid c}$ over all $c\in \mathcal{C}$, then $[\widehat{\underline{\COR}}_{0c}, \widehat{\overline{\COR}}_{0c}]$ is also a simultaneous asymptotic $1-\alpha$ confidence interval for the true causal odds ratio $\COR_{0c}$ over all $c\in \mathcal{C}$. 
\end{remark}

The confidence bounds in Theorem \ref{thm:confidence_bound} can be further improved by incorporating the constraint that the true $\pi_{zy\mid c}$s sum up to $1$.
However, the computation of the resulting confidence bounds is more demanding. Fortunately, it can often be efficiently solved through quadratic programming, as discussed in detail in the next subsection.

\subsection{Confidence bounds from quadratic programming}\label{sec:cb_qp}

We can modify the quadratic programming in \eqref{eq:qp}--\eqref{eq:qp_objective} to obtain confidence bounds for the true causal odds ratio $\COR_{0c}$ under the sensitivity analysis proposed in Theorem \ref{thm:sen_3_constr}. 

\begin{theorem}\label{thm:cb_sen_3_constr}
Suppose that Assumptions \ref{asmp:U}, \ref{asmp:bound_u_not_0}, \ref{asmp:bound_density}, \ref{asmp:effect_heter} and \ref{asmp:Multinomial} hold.  
Let $\widehat{\underline{\COR}}_{0c}$ be the solution to the following quadratic programming problem: minimize $t_1 t_3$ subject to: 
\begin{align}\label{eq:qp_cs}
    & \pi_{zy\mid c} = p_{zy\mid 0c}(1-w) + p_{zy\mid 1c}w, 
    \\ 
    &\sum {p_{zy\mid 0c}} = 1,
    \quad 
    0\leq p_{zy\mid 0c}, p_{zy\mid 1c} \leq 1, & \textup{(Probabilities)}
    \nonumber
    \\
    & { 0\le \ } w\le \delta_c, \quad 
    r_{zy} p_{zy\mid 1c} = p_{zy\mid 0c}, \quad 
    1/\Gamma \leq r_{zy} \leq \Gamma, 
    & 
    \textup{(Assumptions \ref{asmp:bound_u_not_0} and \ref{asmp:bound_density})}
    \nonumber
    \\
    & 
    s_{1} = r_{11}r_{00} \quad 
    s_2 = r_{10}r_{10} \quad 
    s_2s_3 = 1, \quad 
    1/\xi \leq s_1s_3 \leq  \xi,
    & \textup{(Assumption \ref{asmp:effect_heter})}
    \nonumber
    \\
    & 
    t_{1} = p_{11\mid 0c}p_{00\mid 0c}, \quad 
    t_2 = p_{01\mid 0c}p_{10\mid 0c}, \quad 
    t_2t_3 = 1.  
    & 
    \textup{(Objective)}
    \nonumber
    \\
    & (\pi_{00\mid c}, \pi_{10\mid c}, \pi_{01\mid c}, \pi_{11\mid c})^\top \in \mathcal{S}_{c, \alpha},
    & 
    \textup{(Confidence set)}
    \nonumber
\end{align}
and $\widehat{\overline{\COR}}_{0c}$ be the solution to the quadratic programming problem: maximize $t_1 t_3$ subject to the same constraints as in \eqref{eq:qp_cs}, 
where $\mathcal{S}_{c, \alpha}$ is an asymptotic $1-\alpha$ confidence set for the true $\pi_{zy\mid c}$s.  
Then $[\widehat{\underline{\COR}}_{0c}, \widehat{\overline{\COR}}_{0c}]$ is an asymptotically valid $1-\alpha$ confidence interval for the true causal odds ratio $\COR_{0c}$. 
\end{theorem}

The problem in \eqref{eq:qp_cs} differs from that in \eqref{eq:qp}--\eqref{eq:qp_objective} only in the last constraint. 
As discussed in Sections \ref{sec:cs_o} and \ref{sec:con_gen_c}, the constraint in the last line of 
\eqref{eq:qp_cs} is either quadratic or linear in the $\pi_{zy\mid c}$s. 
Consequently, the problem in \eqref{eq:qp_cs} is still a quadratic programming problem.

\begin{figure}[htb]
    \centering
    \includegraphics[width=\linewidth]{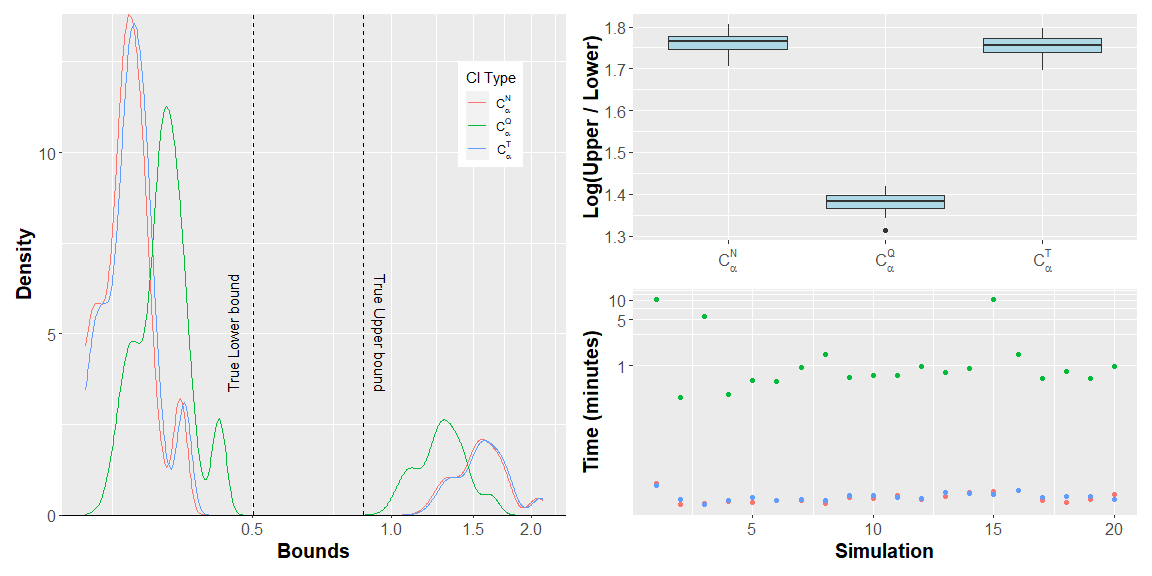}
    \caption{
    The left plot shows the distributions of upper and lower confidence bounds for a fixed \(\pi_{zy\mid c}\)s with $(\pi_{00\mid c}, \pi_{10\mid c}, \pi_{01\mid c}, \pi_{11\mid c})=(0.1, 0.2, 0.3, 0.4)$. The right plots show (i) the boxplots for 
    the difference in the logarithms of the upper and lower confidence bounds
    and (ii) the time needed to compute the bounds using the three forms of confidence sets, with randomly generated \(\pi_{zy\mid c}\)s from a Dirichlet distribution with parameters being $(1,1,1,1)$. 
    For each plot, 
    we run the simulation with 20 iterations, and set the confidence level to be $95\%$. 
    The 
    $\hat{\pi}_{zy\mid c}$s are generated as in Assumption \ref{asmp:Multinomial} with $n=1000$.     
    }
    \label{fig: CI}
\end{figure}

\subsection{Numerical comparison of confidence sets of different forms}\label{sec:comp_cb}%

We use simulation to compare the performance of the confidence bounds in Theorem \ref{thm:cb_sen_3_constr} based on the three confidence sets in \eqref{eq:C_Q}--\eqref{eq:C_T} for a particular level or stratum $c$ of discrete covariates.

We first generate \(\hat{\pi}_{zy\mid c}\)s from a Multinomial distribution as in Assumption \ref{asmp:Multinomial} with $(\pi_{00\mid c}, \pi_{10\mid c}, \pi_{01\mid c}, \pi_{11\mid c})=(0.1, 0.2, 0.3, 0.4)$ and $n_c=1000$. We compute the \(95\%\) confidence bounds using each simulated \(\hat{\pi}_{zy\mid c}\) under our sensitivity analysis with \((\delta_c, \Gamma_c, \xi_c) = (0.1, 5, 2)\). The left plot of Figure \ref{fig: CI} shows the estimated densities of the confidence lower and upper bounds under each of the three confidence sets in \eqref{eq:C_Q}--\eqref{eq:C_T}, across 20 simulations.
The confidence set \(\mathcal{C}_{c, \alpha}^{\textsc{Q}}\) tends to give the shortest confidence bounds.
This is likely due to the fact that the elliptical confidence set, compared with the rectangular ones, contains less extreme values of individual \(\pi_{zy\mid c}\)s.

We then conduct a similar simulation but with the true \(\pi_{zy\mid c}\)s generated
from the Dirichlet distribution with all parameters being \(1\).
The top right plot shows the boxplots of the difference in the logarithms of the upper and lower confidence bounds, using each of the confidence sets in \eqref{eq:C_Q}--\eqref{eq:C_T}. 
Similarly, the confidence set \(\mathcal{C}_{c, \alpha}^{\textsc{Q}}\) tends to give the shortest confidence bounds for $\COR_{0c}$. 
However, this benefit comes at the expense of higher computational cost. The bottom right plot displays the boxplots of the computation time (in seconds) using each of the three confidence sets. 
The computational time using the confidence set \(\mathcal{C}_{c, \alpha}^{\textsc{Q}}\), which imposes a quadratic constraint on the true \(\pi_{zy\mid c}\)s, is larger than that using the confidence sets \(\mathcal{C}_{c, \alpha}^{\textsc{N}}\) and \(\mathcal{C}_{c, \alpha}^{\textsc{T}}\), both of which impose linear constraints on the true \(\pi_{zy\mid c}\)s.

{\rev 
\subsection{Numerical illustration with a continuous observed covariate}\label{sec:num_cont}

We now consider simulation with a univariate continuous observed covariate $C$. 
We generate the observed covariate, exposure and outcome in the test negative design as independent and identically distributed  samples from the following model: 
 \begin{equation}\label{eq: model continuous simulation}
    C \sim \mathrm{Uniform}(0,1), 
    \quad 
    (Z,Y)\mid C = c \sim \textrm{Multinomial}(1, ({\pi}_{00\mid c}, {\pi}_{10\mid c}, {\pi}_{01\mid c}, {\pi}_{11\mid c})),
    \end{equation}
where 
\begin{equation}\label{eq: pi continuous simulation}
    \pi_{zy\mid c} 
    = 
    \frac{\exp( \tilde{c}^\top \beta_{zy})}{\sum_{z',y'\in \{0,1\}} \exp( \tilde{c}^\top  \beta_{z'y'} )}, 
    \quad 
    \text{for } z, y \in \{0,1\}
\end{equation}
and 
\begin{align}\label{eq:simu_beta_value}
\beta_{00} = (0,0)^\top, \ \ 
\beta_{10} = (0.5, 0.5)^\top, \ \ 
\beta_{01} = (1.3, -1.3)^\top, \ \ 
\beta_{11} = (-0.1, -0.3)^\top. 
\end{align}

For each simulation iteration, we generate $50000$ samples from the model in \eqref{eq: model continuous simulation}--\eqref{eq:simu_beta_value}, fit a Multinomial logistic regression model for $\pi_{zy\mid c}$s, and obtain the maximum likelihood estimate for the coefficients $\beta_{zy}$s and the estimate for its asymptotic covariance matrix. 
We then use Theorem \ref{thm:uniform_CI_model_pi} to construct simultaneous $95\%$ confidence intervals for $\pi_{zy\mid c}$s over $c\in [0,1]$, and further use them to construct simultaneous $95\%$ confidence bounds for $\COR_{0c}$s over $c\in [0,1]$ under our proposed sensitivity analysis (with sensitivity parameters $\delta_c=0.1$, $\Gamma_c=5$, $\xi_c=2$ for all $c\in [0,1]$). 
For comparison, we also use Theorem \ref{thm:uniform_CI_model_pi} to construct simultaneous confidence bounds for the observable odds ratios $\OR_{c}$s.

\begin{figure}[htb]
    \centering
    \includegraphics[width=\linewidth]{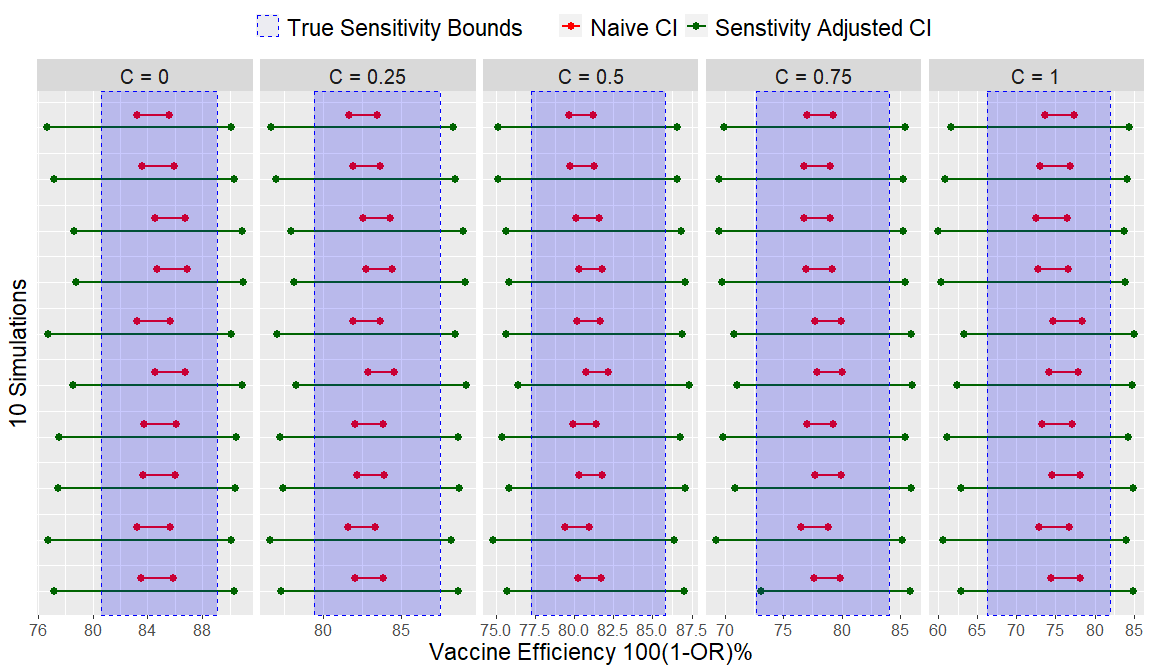}
    \caption{\rev Simultaneous confidence bounds with a continuous confounder across 10 simulation iterations. For each value of $C$, the blue shaded region denotes the true causal bounds, the red line segments show the simultaneous confidence intervals for the odds ratio, and the green line segments represent the sensitivity adjusted simultaneous confidence intervals with sensitivity parameters $\delta = 0.1$, $\Gamma = 5$, and $\xi = 2$.}
    \label{fig: continuous simulation}
\end{figure}

Figure \ref{fig: continuous simulation} shows the results across 10 simulation iterations. The blue shaded regions depict the true causal bounds under the sensitivity analysis, computed using the true values of $\pi_{zy\mid c}$. The red lines show the simultaneous confidence interval for the odds ratio, which does not cover the causal bounds. By contrast, the green lines represent the simultaneous confidence intervals under the sensitivity analysis, which successfully cover the true causal bounds across all values of $c$ shown in the figure.

}

\section{Proofs of theorems}\label{sec: proof}

For simplicity, 
we condition on the observed covariates $C=c$ implicitly throughout the proofs in this section, 
and use $a_{**}$ to denote $(a_{00}, a_{01}, a_{10}, a_{11})$.

\subsection{A useful lemma}

Below we present Lemma \ref{lemma:useful}, which is useful for proving the sharpness of our sensitivity analysis; see Remark \ref{rmk:sharp}. 
We first give some intuition for the purpose of this lemma, and then formally present this lemma and its proof. 

Consider first the case with binary exposure and outcome.  
Using Lemma \ref{lemma:useful}, we can know that,  
as long as the conditional distributions of 
\begin{align}\label{eq:cond_binary}
    U\mid T=1, \quad (Z,Y) \mid U, T=1
\end{align}
are coherent with the observed data, 
any specification for them 
is feasible. 
That is, there exists underlying data generating process satisfying Assumption \ref{asmp:U} that leads to the observed data distribution and the specified conditional distributions for \eqref{eq:cond_binary}. 
We will show later that 
both our causal estimand and our sensitivity analysis constraints can be represented by the conditional distributions in \eqref{eq:cond_binary}. 
Therefore, to derive the sharp bounds under the proposed sensitivity analysis, it suffices to consider all possible conditional distributions in \eqref{eq:cond_binary} that are coherent with the observed data. 

Consider then the case with general categorical exposure and outcome. The implication of Lemma \ref{lemma:useful} is similar. 
As long as the conditional distributions of 
\begin{align}\label{eq:cond_cat}
    U\mid \tilde{T}=1 & \ \sim U \  \mid Z \in \{0,1\}, Y \in \{0, 1\}, T=1, 
    \nonumber
    \\ 
    (Z,Y) \mid U, \tilde{T}=1 & \ \sim \ (Z,Y) \mid U, 
    Z \in \{0,1\}, Y \in \{0, 1\}, T=1
\end{align}
are coherent with the observed distribution of $(Z, Y)\mid \tilde{T}=1 \sim (Z, Y)\mid Z \in \{0,1\}, Y \in \{0, 1\}, T=1$, 
any specification for them 
is feasible.
That is, there exists underlying data generating process satisfying Assumption \ref{asmp:U} that leads to the observed data distribution of $(Z, Y) \mid T=1$ and the specified conditional distributions in \eqref{eq:cond_cat}. 
We will show later that 
both the estimand and the sensitivity analysis constraints can be represented by the conditional distributions in \eqref{eq:cond_cat}. Therefore, to derive the sharp bounds under the proposed sensitivity analysis, it suffices to consider all possible conditional distributions in \eqref{eq:cond_cat} that are coherent with the observed data.

\begin{lemma}\label{lemma:useful}
    Let $\bar{Z}\in \{0, 1, \ldots, I\}, \bar{Y} \in \{0, 1, 2\ldots, J\}$ be two categorical random variables, and $\bar{T}\in \{0, 1\}$ be a binary random variable, where $I\ge 1$ and $J\ge 1$.
    Consider any set $\mathcal{U}$, any probability measure $\omega$ on the set $\mathcal{U}$,
    and any $(q_{00\mid u}, q_{01\mid u}, q_{10\mid k}, q_{11\mid u})$s for $u \in \mathcal{U}$ such that 
    \begin{align}
        q_{00\mid u} +  q_{01\mid u} + q_{10\mid u} + q_{11\mid u} = 1 \ \ \text{ for all } u \in \mathcal{U}, 
        \label{eq:sumq}
        \\
        \int_{\mathcal{U}} q_{ij\mid u} \text{d} \omega(u)  = \Pr(\bar{Z} = i, \bar{Y} = j  \mid \bar{Z} \in \{0,1\}, \bar{Y} \in \{0, 1\}, \bar{T}=1) \ \ \text{ for all } i,j\in \{0,1\}.
        \label{eq:sumwq}
    \end{align}  
    Then there always exist random variables $Z\in \{0, 1, \ldots, I\}, Y(0), \ldots, Y(I) \in \{0, 1, 2\ldots, J\},$ $U\in \mathcal{U}$ and  $T \in \{0, 1\}$ such that, 
    with $Y = \sum_{i=0}^I \I(Z=i) Y(i)$, 
    \begin{enumerate}[label={(\roman*)}, topsep=1ex,itemsep=-0.3ex,partopsep=1ex,parsep=1ex]

    \item 
    The conditional distribution of $U \mid (Z \in \{0,1\},\, Y \in \{0,1\},\, T=1)$ induces the same probability measure as $\omega$;

    \item $\Pr(Z = i, Y = j  \mid U=u, Z \in \{0,1\}, Y \in \{0, 1\}, T=1) = q_{ij\mid u}$ for all $i,j\in \{0,1\}$ and $u \in \mathcal{U}$; 

     \item $Z, Y \mid T = 1  \ \sim \ \bar{Z}, \bar{Y} \mid  \bar{T}=1$; 
     
    \item 
    $Z \ind T \mid Y, U$; 
    
    \item
    $Z \ind Y(z) \mid U$ for all $0\le z\le I$.
    \end{enumerate}
\end{lemma}

For clarity, we split the proof of Lemma \ref{lemma:useful} in into the following two parts. 

\begin{proof}[\bf Proof of Lemma \ref{lemma:useful}: part I]
    We show there exist random variables $Z\in \{0, 1, \ldots, I\},$ $Y \in \{0, 1, 2\ldots, J\}, U \in \mathcal{U}$ and $T \in \{0, 1\}$ such that properties (i)--(iv) in Lemma \ref{lemma:useful} hold. 
    We construct random variables $(Z, Y, U, T)$ with  their distributions specified in the following steps. 
    In steps (a) and (b), we specify the conditional distribution of $(Z,Y,U)$ given $T=1$. 
    In steps (c) and (d), we specify the conditional distribution of $(Z,Y,U)$ given $T=0$. 
    In step (e), we specify the marginal distribution of $T$.
    Let $\nu$ denote the usual counting measure on natural numbers. 
    Let $f_{Z, Y, U\mid T=1}$ denote the density for the conditional distribution of $(Z,Y,U)$ given $T=1$ with respect to $\nu\times \nu \times \omega$, 
    $f_{Y, U\mid T=0}$ denote the density for the conditional distribution of $(Y,U)$ given $T=0$ with respect to $\nu\times \omega$, 
    and 
    $f_{Z\mid Y=j,U=u,T=t}$ denote the density for the conditional distribution of $Z$ given $(Y=j,U=u,T=t)$ with respect to $\nu$. 
    \begin{enumerate}[label={(\alph*)}, topsep=1ex,itemsep=-0.3ex,partopsep=1ex,parsep=1ex]
    \item For any $i,j\in \{0,1\}$ and $u \in \mathcal{U}$, define 
    \begin{align*}
        f_{Z, Y, U\mid T=1}(i,j,u)
        & = 
        q_{ij\mid u} 
        \cdot
        \Pr(\bar{Z} \in \{0,1\}, \bar{Y} \in \{0, 1\} \mid  \bar{T}=1 ). 
    \end{align*}
    
    \item For any $0\le i\le I$ and $0\le j\le J$ with $(i,j)\notin \{(0,0), (1,0), (0,1), (1,1)\}$, and any $u \in \mathcal{U}$, define $f_{Z, Y, U\mid T=1}(i,j, u)$ such that 
    \begin{align*}
        \int_{\mathcal{U}} f_{Z, Y, U\mid T=1}(i,j, u) \text{d} \omega(u) = \Pr(\bar{Z} =i, \bar{Y} =j \mid  \bar{T}=1 ).
    \end{align*}
    For example, we can set 
    $f_{Z, Y, U\mid T=1}(i,j, u) = \Pr(\bar{Z} =i, \bar{Y} =j \mid  \bar{T}=1 )$.

    \item For any $0\le j\le J$ and $u \in \mathcal{U}$, 
    define $f_{Y, U\mid T=0}(j, u)$ arbitrarily. For example, we can set 
    $f_{Y, U\mid T=0}(j, u) = (J+1)^{-1}$. 
    
    \item For any $0\le i\le I$, $0\le j\le J$ and $u \in \mathcal{U}$, define 
    \begin{align*}
        f_{Z\mid Y=j, U=u, T=0}(i) & 
        = f_{Z\mid Y=j, U=u, T=1}(i)
        = 
        \frac{
        f_{Z,Y,U\mid T=1}(i,j,u)
        }{\sum_{z=0}^I f_{Z,Y,U\mid T=1}(z,j,u)}, 
    \end{align*}
    where $f_{Z,Y,U\mid T=1}$ is defined in step (a), and $f_{Z\mid Y=j, U=u, T=1}(i)$ is the implied density for the conditional distribution of $Z$ given $(Y=j, U=u, T=1)$.

    \item Define $\Pr(T=1)$ arbitrarily. For example, we can set $\Pr(T=1) = 0.5$.     
    \end{enumerate}
    
    We first prove that the above specification for the distribution of $(Z,Y,U,T)$ is well-defined. By the construction, it suffices to verify that the defined conditional distribution of $(Z,Y,U)$ given $T=1$ is valid. From step (a) and \eqref{eq:sumwq}, we have, for any $i,j\in \{0,1\}$, 
    \begin{align}\label{eq:construct_zy}
        \Pr(Z=i, Y=j\mid T=1)
        & = 
        \int_{\mathcal{U}} f_{Z, Y, U\mid T=1}(i,j,u) \text{d} \omega(u)
        \nonumber
        \\
        & = 
        \int_{\mathcal{U}} 
        q_{ij\mid u} 
        \cdot
        \Pr(\bar{Z} \in \{0,1\}, \bar{Y} \in \{0, 1\} \mid  \bar{T}=1 )
        \text{d} \omega(u)
        \nonumber
        \\
        & = 
        \int_{\mathcal{U}} 
        q_{ij\mid u} 
        \text{d} \omega(u)
        \cdot
        \Pr(\bar{Z} \in \{0,1\}, \bar{Y} \in \{0, 1\} \mid  \bar{T}=1 ) 
        \nonumber
        \\
        & = \Pr(\bar{Z} = i, \bar{Y} = j  \mid \bar{Z} \in \{0,1\}, \bar{Y} \in \{0, 1\}, \bar{T}=1) 
        \nonumber
        \\
        & \quad \  \cdot
        \Pr(\bar{Z} \in \{0,1\}, \bar{Y} \in \{0, 1\} \mid  \bar{T}=1 )
        \nonumber
        \\
        & = \Pr(\bar{Z} = i, \bar{Y} = j  \mid \bar{T}=1). 
    \end{align}
    By the construction in step (b), we can then know that 
    $
    \sum_{i=0}^I \sum_{j=0}^J \int_{\mathcal{U}} f_{Z, Y, U\mid T=1}(i,j,u) \text{d} \omega(u) 
    $
    equals $1$, 
    and $f_{Z, Y, U\mid T=1}(i,j,u) $ is nonnegative for all $i,j, u$. 
    Thus, the conditional distribution of $(Z,Y,U)$ given $T=1$ is well-defined. 
    
    We then verify properties (i)--(iv) in Lemma \ref{lemma:useful}. 
    From step (a) and \eqref{eq:sumq}, for any $u \in \mathcal{U}$, 
    the conditional density of $U$ given $(Z \in \{0,1\}, Y \in \{0, 1\}, T=1)$ with respect to the measure $\omega$, evaluated at $u$, is 
    \begin{align*}
        \frac{\sum_{i=0}^1 \sum_{j=0}^1 f_{Z, Y, U\mid T=1}(i,j,u)}{
        \int_{\mathcal{U}} \sum_{i=0}^1 \sum_{j=0}^1 f_{Z, Y, U\mid T=1}(i,j,u') \text{d} \omega(u')
        }
        = 
        \frac{\sum_{i=0}^1 \sum_{j=0}^1 q_{ij\mid u}}{
        \int_{\mathcal{U}} \sum_{i=0}^1 \sum_{j=0}^1 q_{ij\mid u'} \text{d} \omega(u')
        }
        = \frac{1}{ \int_{\mathcal{U}} \text{d} \omega(u')} = 1. 
    \end{align*}
    Thus, property (i) holds. 
    For all $i,j\in \{0,1\}$ and $u \in \mathcal{U}$, 
    from step (a) and \eqref{eq:sumq}, we have 
    \begin{align*}
        & \quad \ \Pr(Z = i, Y = j  \mid U=u, Z \in \{0,1\}, Y \in \{0, 1\}, T=1)\\
        & = 
        \frac{ f_{Z, Y, U\mid T=1}(i,j,u)}{
        \sum_{z=0}^1 \sum_{y=0}^1 f_{Z, Y, U\mid T=1}(z,y,u)
        }
        = \frac{q_{ij\mid u}}{ \sum_{z=0}^1 \sum_{y=0}^1 q_{zy\mid u} } = q_{ij\mid u}. 
    \end{align*}
    Thus, property (ii) holds.
    From \eqref{eq:construct_zy} and the construction in step (b), we can know that property (iii) holds.
    By the construct in step (d), property (iv) holds.
\end{proof}

\begin{proof}[\bf Proof of Lemma \ref{lemma:useful}: part II]
We show that for any random variables $(\check{Z}, \check{Y}, \check{U}, \check{T})$, there exist random variables $(Z, Y(0), Y(1), \ldots, Y(I), U, T)$ such that 
    \begin{align}\label{eq:proper_joint}
        Z \ind Y(z) \mid U \text{ for all } 0\le z\le I, 
        \ \ \text{and} \ \ 
        ( \check{Z}, \check{Y}, \check{U}, \check{T}) \sim 
        \left(Z, \sum_{i=0}^I \I(Z=i)Y(i), U, T \right). 
    \end{align}
    We construct random variables $(Z, Y(0), Y(1), \ldots, Y(I), U, T)$ with their distributions specified in the following way. 
    \begin{enumerate}[label={(\alph*)}, topsep=1ex,itemsep=-0.3ex,partopsep=1ex,parsep=1ex]
    \item Let $U$ has the same marginal distribution as $\check{U}$.

    \item Let $Z$ and $(Y(0), Y(1), \ldots, Y(I))$ be conditionally independent given $U$. 

    \item Let $Z\mid U=u$ follow the same distribution as $\check{Z}\mid \check{U} = u$, for all $u \in \mathcal{U}$. 

    \item Let $Y(z)\mid U=u$ follow the same distribution as $\check{Y} \mid \check{Z} = z, \check{U}=u$, for $0\le z \le I$ and all $u \in \mathcal{U}$. 
    The joint distribution of $(Y(0), Y(1), \ldots, Y(I))$ given $U$ can be defined arbitrarily as long as their marginal distributions follow the specified form. For example, we can set $Y(0)$, $\ldots,$ and $Y(I)$ to be conditionally independent given $U$. 

    \item Define the conditional distribution of $T$ given $(Z, Y(0), Y(1), \ldots, Y(I), U)$ such that, for any $u\in \mathcal{U}$, $z\in \{0, 1, \ldots, I\}$ and $j_0, \ldots, j_I \in \{0, 1, \ldots, J\}$, 
    \begin{align*}
        & \quad \ \Pr(T=1 \mid U=u, Z = z, Y(0) = j_0, Y(1) = j_1, \ldots, Y(I) = j_I)\\
        & = \Pr(\check{T} = 1 \mid \check{U} = u, \check{Z} = z, \check{Y} = j_z ). 
    \end{align*}
    \end{enumerate}
    By the construction in step (b), the first property in \eqref{eq:proper_joint} holds. 
    By the construction in steps (a) and (c), we know that $(Z, U) \sim (\check{Z}, \check{U})$. By the construction in steps (b) and (d), we have, for any $0\le z\le I$ and $u\in \mathcal{U}$, 
    \begin{align}\label{eq:Y_obs_z}
        \sum_{i=0}^I \I(Z=i)Y(i) \mid Z=z, U = u 
        & \ \sim \  
        Y(z) \mid Z=z, U = u 
        \ \sim \ Y(z) \mid U = u
        \nonumber
        \\
        & \ 
        \sim \ \check{Y} \mid \check{Z} = z, \check{U}=u.
    \end{align}
    This then implies that $(\sum_{i=0}^I \I(Z=i)Y(i), Z, U) \sim (\check{Y}, \check{Z}, \check{U})$. 
    By the construction in steps (d) and (e), we have, for any $u\in \mathcal{U}$, $0\le z\le I$ and $0\le j\le J$,  
    \begin{align*}
        & \quad \ 
        \Pr(T = 1, \sum_{i=0}^I \I(Z=i)Y(i) = j \mid Z=z, U = u)\\
        & = \Pr(T = 1, Y(z) = j \mid Z=z, U = u)
        \\
        & = 
        \sum_{(j_0, \ldots, j_I)\in \{0, \ldots, J\}^I: j_z = j}
        \Pr(T=1, Y(0) = j_0, Y(1) = j_1, \ldots, Y(I) = j_I \mid U=u, Z = z)\\
        & = 
        \sum_{(j_0, \ldots, j_I)\in \{0, \ldots, J\}^I: j_z = j}
        \left\{ \Pr(T=1 \mid Y(0) = j_0, Y(1) = j_1, \ldots, Y(I) = j_I, U=u, Z = z) \cdot \right.
        \\
        & \quad \ \left. 
        \Pr(Y(0) = j_0, Y(1) = j_1, \ldots, Y(I) = j_I \mid U=u, Z = z) \right\}\\
        & = 
        \Pr(\check{T}=1 \mid \check{Y} = j, \check{U}=u, \check{Z} = z) \cdot
        \\
        & \quad \ \sum_{(j_0, \ldots, j_I)\in \{0, \ldots, J\}^I: j_z = j}
        \Pr(Y(0) = j_0, Y(1) = j_1, \ldots, Y(I) = j_I \mid U=u, Z = z)\\
        & = \Pr(\check{T}=1 \mid \check{Y} = j, \check{U}=u, \check{Z} = z) \cdot \Pr(Y(z) = j \mid U=u, Z = z)\\
        & = \Pr(\check{T}=1 \mid \check{Y} = j, \check{U}=u, \check{Z} = z) \cdot \Pr(\check{Y} = j \mid \check{U}=u, \check{Z} = z)\\
        & = \Pr(\check{T}=1, \check{Y} = j \mid \check{U}=u, \check{Z} = z),
    \end{align*}
    where the forth equality follows from the construction in (e) and the second last equality follows from \eqref{eq:Y_obs_z}. 
    Consequently, 
    \begin{align*}
        T \mid \sum_{i=0}^I \I(Z=i)Y(i) = j, Z=z, U = u
        \ \sim \ 
        \check{T}=1 \mid \check{Y} = j,  \check{U}=u, \check{Z} = z
    \end{align*}
    Therefore, 
    the second property in \eqref{eq:proper_joint} must hold. 

    From the above, Lemma \ref{lemma:useful} holds. 
\end{proof}

\subsection{Proof of Theorem \ref{thm:simple_bound_cor}}

To prove Theorem \ref{thm:simple_bound_cor}, we need the following lemma, where we also make the conditioning on $c$ implicit. 

\begin{lemma}\label{lemma: COR p_xy}
    Recall the definition of $\COR_u$ in \eqref{eq:COR} and $p_{zy\mid u}$s in \eqref{eq:pzy_u}. 
    Under Assumption \ref{asmp:U}, we have
    \begin{equation*}
        \COR_{u} = \frac{p_{11\mid u}p_{00\mid u}}{p_{10\mid u}p_{01\mid u}}.
    \end{equation*}
\end{lemma}

\begin{proof}[\bf Proof of Lemma \ref{lemma: COR p_xy}]
Under Assumption \ref{asmp:U}, we have 
\begin{align*}
    \COR_{u} & = \frac{\P(Y=1\mid Z=1, U = u)/\P(Y = 0 \mid Z=1, U = u)}{\P(Y=1\mid Z=0, U = u)/\P(Y = 0 \mid Z=0, U = u)} 
    \nonumber
    \\
    & \quad \ \text{(\textit{by Assumption \ref{asmp:U}(i)})}
    \nonumber
    \\
    & = 
    \frac{\P(Z=1\mid Y=1, U = u)/\P(Z=0 \mid Y = 1, U = u)}{\P(Z=1\mid Y=0, U = u)/\P(Z=0 \mid Y=0, U = u)} 
    \nonumber
    \\
    & \quad \ \text{(\textit{by invariance of the odds ratio})}
    \nonumber
    \\
    & = 
    \frac{\P(Z=1\mid Y=1, T = 1, U = u)/\P(Z=0 \mid Y = 1, T=1, U = u)}{\P(Z=1\mid Y=0, T=1, U = u)/\P(Z=0 \mid Y=0, T= 1, U = u)} 
    \nonumber
    \\
    & \quad \ \text{(\textit{by Assumption \ref{asmp:U}(ii)})}
    \nonumber
    \\
    & = 
    \frac{\P(Z=1, Y=1\mid T = 1, U = u)/\P(Z=0, Y = 1\mid T=1, U = u)}{\P(Z=1, Y=0\mid T=1, U = u)/\P(Z=0, Y=0 \mid T= 1, U = u)} 
    \nonumber
    \\
    & \quad \ \text{(\textit{by some algebra})}
    \nonumber\\
    & = \frac{p_{11\mid u}p_{00\mid u}}{p_{10\mid u}p_{01\mid u}}, 
    \\
    & \quad \ \text{(\textit{by definition})}
\end{align*}
with the reason for each equality given in parentheses. 
Thus, Lemma \ref{lemma: COR p_xy} holds. 
\end{proof}

\begin{proof}[\bf Proof of Theorem \ref{thm:simple_bound_cor}]
Let $w = \P(U \ne 0 \mid T=1)$. 
By the law of total probability, the observed distribution satisfies that, for $z,y \in \{ 0,1 \}$, 
\begin{align}\label{eq:o_wp}
    \pi_{zy} & = \sum_{u=0}^1 \P(Z=z, Y=y, U = u \mid T=1) 
    \nonumber
    \\
    & = \sum_{u=0}^1 \P(Z=z, Y=y \mid U = u, T=1) \P(U = u \mid T=1) \nonumber
    \\
    & = (1-w) p_{zy\mid 0} + w p_{zy\mid 1}.  
\end{align}
From Assumption \ref{asmp:U} and Lemma \ref{lemma: COR p_xy}, we have $\COR_0 = p_{11\mid 0}p_{00\mid 0}/p_{01\mid 0}p_{10\mid 0}$.

We first show that, for any $0\le w \le 1$, 
\begin{align}\label{eq:COR_w}
    \min\left\{ \frac{(\pi_{11}-w)_{+}\pi_{00}}{\pi_{10} \pi_{01}}, \  \frac{\pi_{11}(\pi_{00}-w)_{+}}{\pi_{10} \pi_{01}} \right\} 
    \le \COR_0 \le 
    \max\left\{ \frac{\pi_{11}\pi_{00}}{(\pi_{10}-w)_{+}\pi_{01}}, \frac{\pi_{11}\pi_{00}}{\pi_{10}(\pi_{01}-w)_{+}} \right\}. 
\end{align}
When $w = 0$, the inequalities in \eqref{eq:COR_w} hold obviously by Proposition \ref{prop:cor_perfect_test}. 
When $w = 1$, the inequalities in \eqref{eq:COR_w} also hold obviously since the lower and upper bounds are, respectively, $0$ and $\infty$.  
Below we consider only the case in which $w\in (0,1)$. 
From \eqref{eq:o_wp}, we have 
\begin{align*}
    \COR_{0} = 
    \frac{p_{11\mid 0} p_{00 \mid 0}}{p_{10\mid 0} p_{01\mid 0}}
    = 
    \frac{(1-w)p_{11\mid 0} (1-w)p_{00 \mid 0}}{(1-w)p_{10\mid 0} (1-w)p_{01\mid 0}}
    = 
    \frac{(\pi_{11} - w p_{11\mid 1})(\pi_{00} - w p_{00\mid 1})}{(\pi_{10} - w p_{10\mid 1})(\pi_{01} - w p_{01\mid 1})},
\end{align*}
which can then be bounded by 
\begin{align*}
    \frac{(\pi_{11} - w p_{11\mid 1})(\pi_{00} - w p_{00\mid 1})}{\pi_{10} \pi_{01}}
    \le \COR_{0} \le \frac{\pi_{11} \pi_{00}}{(\pi_{10} - w p_{10\mid 1})(\pi_{01} - w p_{01\mid 1})}. 
\end{align*}
The key is to bound $(\pi_{11} - w p_{11\mid 1})(\pi_{00} - w p_{00\mid 1})$ in the lower bound and $(\pi_{10} - w p_{10\mid 1})(\pi_{01} - w p_{01\mid 1})$ in the upper bound above.
\begin{enumerate}[label={(\roman*)}, topsep=1ex,itemsep=-0.3ex,partopsep=1ex,parsep=1ex]
\item Bound $(\pi_{11} - w p_{11\mid 1})(\pi_{00} - w p_{00\mid 1})$. 
Note that 
\begin{align*}
    & \quad \ (\pi_{11} - w p_{11\mid 1})(\pi_{00} - w p_{00\mid 1}) 
    \\
    & \ge (\pi_{11} - w p_{11\mid 1}) \{\pi_{00} - w (1-p_{11\mid 1})\}_{+}
    = 
    \big[ (\pi_{11} - w p_{11 \mid 1}) \{\pi_{00} - w (1-p_{11\mid 1})\} \big]_{+}
    \\
    & = \big\{ -(  w p_{11\mid 1} - \pi_{11}) (w p_{11\mid 1}  + \pi_{00} - w ) \big\}_{+}. 
\end{align*}
Because $-(  w p_{11\mid 1} - \pi_{11}) (w p_{11\mid 1}  + \pi_{00} - w )$ is a concave function of $p_{11\mid 1}$, it must achieve its minimum value at $p_{11\mid 1} = 0$ or $1$. 
Consequently, 
\begin{align*}
    (\pi_{11} - w p_{11\mid 1})(\pi_{00} - w p_{00\mid 1}) 
    & \ge 
    \min\left[
    \big\{ -(  w 1 - \pi_{11}) (w 1  + \pi_{00} - w ) \big\}_{+}, 
    \big\{ -( - \pi_{11}) ( \pi_{00} - w ) \big\}_{+}
    \right]\\
    & = 
    \min\left[
    \{ (\pi_{11} -  w ) \pi_{00} \}_{+}, 
    \{ \pi_{11} ( \pi_{00} - w ) \}_{+}
    \right]\\
    & = 
    \min\left\{
    (\pi_{11} -  w )_{+} \pi_{00}, 
    \pi_{11} ( \pi_{00} - w )_{+}
    \right\}. 
\end{align*}
\item 
Bound $(\pi_{10} - w p_{10\mid 1} ) (\pi_{01} - w p_{01\mid 1})$. 
By the same logic as before, 
\begin{align*}
    (\pi_{10} - w p_{10\mid 1} ) (\pi_{01} - w p_{01\mid 1})
    \ge 
    \min\left\{
    (\pi_{10} -  w )_{+} \pi_{01}, 
    \pi_{10} ( \pi_{01} - w )_{+}
    \right\}. 
\end{align*}
\end{enumerate}
From (i) and (ii), we have 
\begin{align*}
    \COR_{0}
    \ge  
    \frac{(\pi_{11} - b_{11} w)(\pi_{00} - b_{00} w)}{\pi_{10} \pi_{01}} 
    \ge 
    \min\left\{
    \frac{(\pi_{11} -  w )_{+} \pi_{00}}{\pi_{10} \pi_{01}}, 
    \frac{\pi_{11} ( \pi_{00} - w )_{+}}{\pi_{10} \pi_{01}}
    \right\}, 
\end{align*}
and 
\begin{align*}
    \COR_0 
    \le 
    \frac{\pi_{11} \pi_{00}}{(\pi_{10} - b_{10} w) (\pi_{01} - b_{01} w)}
    \le 
    \max\left\{
    \frac{\pi_{11} \pi_{00}}{(\pi_{10} -  w )_{+} \pi_{01}}, 
    \frac{\pi_{11} \pi_{00}}{\pi_{10} ( \pi_{01} - w )_{+}}
    \right\},
\end{align*}
i.e., the inequalities in \eqref{eq:COR_w} hold.

We then consider the bounds on $\COR_0$ under Assumption \ref{asmp:bound_u_not_0}. 
Note that the lower bound of $\COR_0$ in \eqref{eq:COR_w} is decreasing in $w$ and the upper bound of $\COR_0$ in \eqref{eq:COR_w} is increasing in $w$. 
These immediately imply the bounds on $\COR_0$ in Theorem \ref{thm:simple_bound_cor}. 

Finally, we show that the bounds on $\COR_0$ in Theorem \ref{thm:simple_bound_cor} are sharp. 
From Lemma \ref{lemma:useful}, it suffices to show that $\COR_0$ can attain any of the following four values:
\begin{align*}
    \frac{(\pi_{11} -  \delta )_{+} \pi_{00}}{\pi_{10} \pi_{01}}, \ \ 
    \frac{\pi_{11} ( \pi_{00} - \delta )_{+}}{\pi_{10} \pi_{01}}, \ \ 
    \frac{\pi_{11} \pi_{00}}{(\pi_{10} -  \delta )_{+} \pi_{01}}, \ \ 
    \frac{\pi_{11} \pi_{00}}{\pi_{10} ( \pi_{01} - \delta )_{+}}, 
\end{align*}
at some values of $(w, p_{**\mid 0}, p_{**\mid 0})$ that are coherent with the observed data distribution. 
We prove only that $\COR_0$ can attain the value $(\pi_{11} -  w )_{+} \pi_{00}/(\pi_{10} \pi_{01})$; the proof for the other three follow by the same logic. 
Let $p_{11\mid 1} = 1$, $p_{10\mid 1} = p_{01\mid 1} = p_{00\mid 1} = 0$, 
and $w = \min\{\delta, \pi_{11}\}< 1$, recalling that all the $\pi_{zy}$s are assumed to be positive. 
From \eqref{eq:o_wp}, we then have 
$p_{10\mid 0} = \pi_{10}/(1-w), p_{01\mid 0} = \pi_{01}/(1-w), p_{00\mid 0} = \pi_{00}/(1-w)$, 
and $p_{11\mid 0} = (\pi_{11} - w)/(1-w)$. 
Consequently, 
\begin{align*}
    \COR_0 = \frac{p_{11\mid 0} p_{00 \mid 0}}{p_{10\mid 0} p_{01\mid 0}}
    = \frac{(\pi_{11} - w) \pi_{00}}{\pi_{10} \pi_{01}}
    = \begin{cases}
        \frac{(\pi_{11} - \delta) \pi_{00}}{\pi_{10} \pi_{01}}, & \text{ if } \delta \le \pi_{11} \\
        0 & \text{ if } \delta > \pi_{11}
    \end{cases}
    = \frac{(\pi_{11} -  \delta )_{+} \pi_{00}}{\pi_{10} \pi_{01}}. 
\end{align*}
We can verify that the above values of $w$, $p_{zy\mid 1}$s and $p_{zy\mid 0}$s are all plausible given the observed data distribution $\pi_{zy}$s. 
Therefore, $\COR_0$ can attain the value $(\pi_{11} -  w )_{+} \pi_{00}/(\pi_{10} \pi_{01})$, 
and $\COR_0$ attains this value when one of the $p_{zy\mid 1}$s is $1$ and $w$ is as close as possible to $\delta$. 

From the above, Theorem \ref{thm:simple_bound_cor} holds. 
\end{proof}

\begin{proof}[\bf Proof of Proposition \ref{prop:iden_no_conf_strength}]
From Lemma \ref{lemma: COR p_xy}, for any $u$, the causal odds ratio in \eqref{eq:COR} among units with unmeasured confounder $U=u$ has the following equivalent forms:
\begin{align*}
    \COR_{u}
    & = \frac{p_{11\mid u} p_{00 \mid u}}{p_{10\mid u} p_{01 \mid u}} = \frac{\pi_{11} \pi_{00}}{\pi_{10} \pi_{01}} 
    = \OR,
\end{align*}
where the second equality follows from the fact that $p_{zy\mid 1}=p_{zy\mid 0} = \pi_{zy}$ for all $z,y$ under Assumption \ref{asmp:bound_density} with $\Gamma=1$, and the last equality follows from definition. 
Therefore, Proposition \ref{prop:iden_no_conf_strength} holds. 
\end{proof}

\subsection{Proofs of Theorems \ref{thm:cor_bound_density} and \ref{thm:cor_bound_additional}}
Because Theorem \ref{thm:cor_bound_density} is a special case of Theorem \ref{thm:cor_bound_additional} with \(\delta=1\), 
we only prove Theorem \ref{thm:cor_bound_additional}. 
We need the following three lemmas.
Note that, 
from Lemma \ref{lemma: COR p_xy}, the causal odds ratio $\COR_0$ is equivalently $p_{11\mid u}p_{00\mid u}/(p_{10\mid u}p_{01\mid u})$. 
Lemmas \ref{lemma:cond_w_gamma_lu} and \ref{lemma:cond_delta_gamma_lu}
show that Assumptions \ref{asmp:bound_u_not_0} and \ref{asmp:bound_density} equivalently impose bounds on the $p_{zy\mid 0}$s, 
and Lemma \ref{lemma:optim} derives the bounds on $\COR_0 = p_{11\mid 0}p_{00\mid 0}/(p_{10\mid 0}p_{01\mid 0})$ under given bounds on the $p_{zy\mid 0}$s.

\begin{lemma}\label{lemma:cond_w_gamma_lu}
For any given $\pi_{**}$, $w = \Pr(U\ne 0\mid T=1) \ge 0$ and $\Gamma\ge 1$, 
define $\mathcal{A}_w$ as the set consisting of all possible values of $(p_{**\mid 0}, p_{**\mid 1})$ satisfying that, for $z,y\in \{0,1\}$, 
\begin{align*}
    & \pi_{zy} = (1-w) p_{zy\mid 0} + w p_{zy\mid 1}, 
    \quad 
    \sum_{z,y} p_{zy\mid 0} = \sum_{z,y} p_{zy\mid 1} = 1, \quad
    p_{zy\mid 0} \ge 0, \quad  p_{zy\mid 1} \ge 0, \\
    & 1/\Gamma \le p_{zy\mid 1}/p_{zy\mid 0} \le \Gamma, 
\end{align*}
and define $\mathcal{B}_w$ as the set consisting of all possible values of $(p_{**\mid 0}, p_{**\mid 1})$ satisfying that, for $z,y\in \{0,1\}$,  
\begin{align}\label{eq:bound_p0_w}
    & \pi_{zy} = (1-w) p_{zy\mid 0} + w p_{zy\mid 1}, 
    \quad 
    \sum_{z,y} p_{zy\mid 0} = 1, 
    \nonumber
    \\
    & 
    \max\left\{ \frac{\pi_{zy}}{w\Gamma + (1-w)}, \ \frac{\pi_{zy} - w}{1-w} \right\}
    \le p_{zy\mid 0} \le 
    \min\left\{ \frac{\pi_{zy}\Gamma}{w + (1-w)\Gamma},\  1 \right\}. 
\end{align}
Then $\mathcal{A}_w = \mathcal{B}_w$ for any $0<w\le 1$. 
\end{lemma}
\begin{proof}[\bf Proof of Lemma \ref{lemma:cond_w_gamma_lu}]
    Below we first present some facts that are useful for the proof of Lemma \ref{lemma:cond_w_gamma_lu}. 
    Consider any $w>0$ and any $\pi_{**}$s with $\sum_{z=0}^1 \sum_{y=0}^1 \pi_{zy} = 1$ and $\pi_{zy} > 0$ for all $z,y$. 
    Suppose that   
    $p_{**\mid 0}$ and $p_{**\mid 1}$ satisfy that $\pi_{zy} = (1-w) p_{zy\mid 0} + w p_{zy\mid 1}$ for all $z,y\in \{0,1\}$. 
    We have the following equivalence relations: 
    \begin{align}\label{eq:equiv_ratio}
        \frac{1}{\Gamma} \le \frac{p_{zy\mid 1}}{p_{zy\mid 0}} \le \Gamma 
        & \Longleftrightarrow 
        \frac{1}{\Gamma} \le \frac{w p_{zy\mid 1}}{w p_{zy\mid 0}} \le \Gamma
        \nonumber
        \\
        & \Longleftrightarrow 
        \frac{1}{\Gamma} \le \frac{\pi_{zy} - (1-w) p_{zy\mid 0}}{w p_{zy\mid 0}} \le \Gamma
        \nonumber
        \\
        & \Longleftrightarrow 
        \frac{w}{\Gamma} + (1-w) \le \frac{\pi_{zy}}{p_{zy\mid 0}} \le w\Gamma + (1-w)
        \nonumber
        \\
        & \Longleftrightarrow 
        \frac{\pi_{zy}}{w\Gamma + (1-w)} \le p_{zy\mid 0} \le \frac{\pi_{zy}\Gamma}{w + (1-w)\Gamma},
    \end{align}     
    and 
    \begin{align}\label{eq:equi_bound_p_01}
        0\le p_{zy\mid 1} \le 1 
        & \Longleftrightarrow
        0\le \frac{\pi_{zy} - (1-w) p_{zy\mid 0}}{w} \le 1
        \nonumber
        \\
        & \Longleftrightarrow
        \pi_{zy} - w \le (1-w)p_{zy\mid 0} \le \pi_{zy}
        \nonumber
        \\
        & \Longleftrightarrow
        \frac{\pi_{zy} - w}{1-w} \le p_{zy\mid 0} \le \frac{\pi_{zy}}{1-w},
    \end{align}
    where $(\pi_{zy} - w)/(1-w)$ and $\pi_{zy}/(1-w)$ are defined as $0$ and $\infty$, respectively, when $w=1$. 

    We then prove that $\mathcal{A}_w \subset \mathcal{B}_w$. 
    Let $(p_{**\mid 0}, p_{**\mid 1})$ be any element of $\mathcal{A}_w$. 
    By the definition of the set $\mathcal{A}_w$
    and from \eqref{eq:equiv_ratio} and \eqref{eq:equi_bound_p_01}, 
    we have 
    \begin{align*}
        \max\left\{ \frac{\pi_{zy}}{w\Gamma + (1-w)}, \ \frac{\pi_{zy} - w}{1-w} \right\}
    \le p_{zy\mid 0} \le 
    \min\left\{ \frac{\pi_{zy}\Gamma}{w + (1-w)\Gamma},\  \frac{\pi_{zy}}{1-w} \right\}.
    \end{align*}
    Note that 
    ${\pi_{zy}\Gamma}/\{w + (1-w)\Gamma\} \le \pi_{zy}/(1-w)$ and   
    $p_{zy\mid 0}$ must be in $[0,1]$. We can derive the inequalities in \eqref{eq:bound_p0_w}. 
    Therefore, $(p_{**\mid 0}, p_{**\mid 1})$ must be in $\mathcal{B}_{w}$. 
    
    Second, we prove that $\mathcal{B}_w \subset \mathcal{A}_w$. 
    Let $(p_{**\mid 0}, p_{**\mid 1})$ be any element of $\mathcal{B}_w$. 
    From \eqref{eq:bound_p0_w}, we can know that the last inequalities in both \eqref{eq:equiv_ratio} and \eqref{eq:equi_bound_p_01} hold. This then implies that, for $z,y\in \{0,1\}$, 
    $1/\Gamma \le p_{zy\mid 1}/p_{zy\mid 0} \le \Gamma$ and $0\le p_{zy\mid 1} \le 1$. Furthermore, 
    because $1 = \sum_{z,y} \pi_{zy} = (1-w) 
    \sum_{z,y} p_{zy\mid 0} + w \sum_{z,y} p_{zy\mid 1} = 1-w+ w\sum_{z,y} p_{zy\mid 1}$, we must have $\sum_{z,y} p_{zy\mid 1} = 1$. 
    Therefore, $(p_{**\mid 0}, p_{**\mid 1})$ must be in $\mathcal{A}_{w}$. 
    
    From the above, Lemma \ref{lemma:cond_w_gamma_lu} holds. 
\end{proof}

\begin{lemma}\label{lemma:cond_delta_gamma_lu}
For any given $\delta \in [0,1]$ and $\Gamma\ge 1$, 
define $\mathcal{A}_w$ the same as in Lemma \ref{lemma:cond_w_gamma_lu}, and $\tilde{l}_{**}$ and $\tilde{u}_{**}$ the same as in \eqref{eq:lu_delta_Gamma}.
Define further  
\begin{align*}
    \mathcal{A} & = \{(w, p_{**\mid 0}, p_{**\mid 1}):  (p_{**\mid 0}, p_{**\mid 1}) \in \mathcal{A}_w, 0\le w \le \delta \}, \\
    \mathcal{B} & = \Big\{ p_{**\mid 0}:  
    \tilde{l}_{zy}
    \le p_{zy\mid 0} \le 
    \tilde{u}_{zy} \text{ for } z,y\in \{ 0,1 \}, 
    \text{ and }
    \sum_{z,y} p_{zy\mid 0} = 1
    \Big\}.
\end{align*}
Then 
\begin{align}\label{eq:equi_form_sup}
    \inf_{(w, p_{**\mid 0}, p_{**\mid 1})\in \mathcal{A}} \frac{p_{11\mid 0} p_{00\mid 0}}{p_{10\mid 0}p_{01\mid 0}}
    = 
    \inf_{p_{**\mid 0}\in \mathcal{B}} \frac{p_{11\mid 0} p_{00\mid 0}}{p_{10\mid 0}p_{01\mid 0}}.
\end{align}
\end{lemma}
\begin{proof}[\bf Proof of Lemma \ref{lemma:cond_delta_gamma_lu}]
First, we prove that the left hand side of \eqref{eq:equi_form_sup} is  greater than or equal to the right hand side. 
Let $(w, p_{**\mid 0}, p_{**\mid 1})$ be any element in $\mathcal{A}$. We prove that the corresponding $p_{11\mid 0} p_{00\mid 0}/( p_{10\mid 0}p_{01\mid 0})$ must be greater than or equal to the right hand side of \eqref{eq:equi_form_sup}. We consider the following two cases, depending on whether $w$ is positive. 
\begin{enumerate}[label={(\alph*)}, topsep=1ex,itemsep=-0.3ex,partopsep=1ex,parsep=1ex]
    \item 
    $w>0$. 
    From Lemma \ref{lemma:cond_w_gamma_lu}, we must have $(p_{**\mid 0}, p_{**\mid 1}) \in \mathcal{B}_w$, with $\mathcal{B}_w$ defined the same as in Lemma \ref{lemma:cond_w_gamma_lu}. 
    Consequently, $p_{**\mid 0}$ must satisfy the inequalities in \eqref{eq:bound_p0_w}. 
    Note that the upper bound in \eqref{eq:bound_p0_w} is increasing in $w$ and the lower bound is decreasing in $w$. 
    Given that $w\le \delta$, we must have $\tilde{l}_{zy} \le p_{zy\mid 0} \le \tilde{u}_{zy}$ for $z,y\in \{0,1\}$. 
    Thus, $p_{**\mid 0} \in \mathcal{B}$, and $p_{11\mid 0} p_{00\mid 0}/( p_{10\mid 0}p_{01\mid 0})$ must be greater than or equal to the right hand side of \eqref{eq:equi_form_sup}.
    
    \item 
    $w=0$. By the definition of $A_w$, in this case, $p_{**\mid 0} = \pi_{**}$. 
    We can verify that $\pi_{**}\in \mathcal{B}$. Thus, $p_{**\mid 0} \in \mathcal{B}$, and $p_{11\mid 0} p_{00\mid 0}/( p_{10\mid 0}p_{01\mid 0})$ must be greater than or equal to the right hand side of \eqref{eq:equi_form_sup}.
\end{enumerate}
    
    Second, we prove that the right hand side of \eqref{eq:equi_form_sup} is greater than or equal to the left hand side. 
    Let $p_{**\mid 0}$ be any element in $\mathcal{B}$. 
    Below we consider two cases, depending on whether $\delta$ is positive. 
    \begin{enumerate}[label={(\alph*)}, topsep=1ex,itemsep=-0.3ex,partopsep=1ex,parsep=1ex]
        \item 
        $\delta>0$. 
        Define $p_{zy\mid 1} = \{\pi_{zy} - (1-\delta) p_{zy\mid 0}\}/\delta$ for $z,y=0,1$. 
        From Lemma \ref{lemma:cond_w_gamma_lu}, we can know that $(p_{**\mid 0}, p_{**\mid 1}) \in \mathcal{B}_{\delta} = \mathcal{A}_{\delta}$. 
        Thus, $(\delta, p_{**\mid 0}, p_{**\mid 1}) \in \mathcal{A}$, and $p_{11\mid 0} p_{00\mid 0}/( p_{10\mid 0}p_{01\mid 0})$ must be greater than or equal to the left hand side of \eqref{eq:equi_form_sup}.
        
        \item 
        $\delta=0$. We can verify that $\tilde{l}_{**} = \tilde{u}_{**} = \pi_{**}$. 
        Thus, $p_{**\mid 0} = \pi_{**}$. 
        Obviously, $(0, \pi_{**}, \pi_{**}) \in \mathcal{A}$. 
        Thus, $p_{11\mid 0} p_{00\mid 0}/( p_{10\mid 0}p_{01\mid 0})$ must be greater than or equal to the left hand side of \eqref{eq:equi_form_sup}.
    \end{enumerate}
    
    From the above, Lemma \ref{lemma:cond_delta_gamma_lu} holds. 
\end{proof}

\begin{lemma}\label{lemma:optim}
Consider any 
$\underline{c}_{**}$ and $\overline{c}_{**}$ satisfying that 
$0\le \underline{c}_{zy} \le \overline{c}_{zy}$ for all $z,y\in \{0,1\}$, 
$\sum_{z,y} \underline{c}_{zy} \le 1$ and $\sum_{z,y} \overline{c}_{zy} \ge 1$. 
        Define 
        \begin{align*}
            m = \inf_{c_{**}\in \mathcal{C}} \frac{c_{11}c_{00}}{c_{10}c_{01}} 
            \quad 
            \text{with}
            \quad 
            \mathcal{C} = \left\{ c_{**}:  \underline{c}_{zy} \le c_{zy} \le \overline{c}_{zy} \textup{ for all } z,y=0,1, \textup{ and } \sum_{z,y} c_{zy} = 1 \right\}. 
        \end{align*}
        If $\underline{c}_{11}+\underline{c}_{00}+\overline{c}_{01} + \overline{c}_{10}\geq 1$, 
        then 
        $m = q_{11} q_{00} / (q_{10} q_{01})$ with
        $q_{11} = \underline{c}_{11}$, $q_{00} = \underline{c}_{00}$, $q_{01} = 1 - q_{11} - q_{00} - q_{10}$, and 
        \begin{align*}
                q_{10}
                = 
                \min\left\{ \max\left\{ \underline{c}_{10}, (1 - \underline{c}_{11}-\underline{c}_{00}-\overline{c}_{01}), 
                \frac{1-\underline{c}_{11}-\underline{c}_{00}}{2} \right\}, 
                \overline{c}_{10}, 
                (1 - \underline{c}_{11}-\underline{c}_{00}-\underline{c}_{01})
                \right\};
        \end{align*}
        otherwise, 
        $m = \min \{ q_{11}^{(1)} q_{00}^{(1)} / (q_{10}^{(1)} q_{01}^{(1)}), q_{11}^{(2)} q_{00}^{(2)} / (q_{10}^{(2)} q_{01}^{(2)}) \}$ with 
        $q_{10}^{(j)} = \overline{c}_{10}$, $q_{01}^{(j)} = \overline{c}_{01}$, 
        $q_{00}^{(j)} = 1 -q_{10}^{(j)} - q_{01}^{(j)} - q_{11}^{(j)}$, and 
        \begin{align*}
            q_{11}^{(j)} =  
            \begin{cases}
                \max\left\{ \underline{c}_{11}, 1-\overline{c}_{10}-\overline{c}_{01}-\overline{c}_{00} \right\}, & \text{if } j = 1, \\
                \min\left\{ \overline{c}_{11}, 1-\overline{c}_{10}-\overline{c}_{01}-\underline{c}_{00} \right\}, & \text{if } j = 2. 
            \end{cases}
        \end{align*}
\end{lemma}
\begin{proof}[\bf Proof of Lemma \ref{lemma:optim}]
We first prove that the infimum of $c_{11} c_{00}/( c_{10}c_{01})$ over $c_{**} \in\mathcal{C}$ is the same as that over $c_{**} \in\mathcal{C}'$ with 
\begin{align*}
    \mathcal{C}' \equiv \mathcal{C}\cap\left\{
    c_{**}: c_{11}+c_{00} = \underline{c}_{11} + \underline{c}_{00} \text{ or } c_{10}+c_{10} = \overline{c}_{01}+\overline{c}_{10}
    \right\}. 
\end{align*}
Let $c_{**}$ be any element in $\mathcal{C} \setminus \mathcal{C}'$. 
By definition, we must have $c_{11}+c_{00} > \underline{c}_{00} + \underline{c}_{11}$ and  $c_{10}+c_{10} < \overline{c}_{01}+\overline{c}_{10}$. 
By decreasing $c_{11}$ and $c_{00}$ and increasing $c_{10}$ and $c_{10}$, 
we can then find $\tilde{c}_{**} \in \mathcal{C}'$ such that, for $z,y\in \{0,1\}$, 
\begin{align*}
    \underline{c}_{zy} \le \tilde{c}_{zy} \le c_{zy} \le \overline{c}_{zy} \quad (\text{if } z=y); 
    \qquad
    \underline{c}_{zy}  \le c_{zy} \le \tilde{c}_{zy} \le \overline{c}_{zy} \quad (\text{if } z\ne y).  
\end{align*}
Consequently, 
$c_{11} c_{00}/( c_{10}c_{01}) \ge \tilde{c}_{11} \tilde{c}_{00}/( \tilde{c}_{10} \tilde{c}_{01})$. 
Therefore, 
the infimum of $c_{11} c_{00}/( c_{10}c_{01})$ over $c_{**} \in\mathcal{C}$ must be the same as that over $c_{**} \in\mathcal{C}'$. 

We then consider the infimum of $c_{11} c_{00}/( c_{10}c_{01})$ over $c_{**} \in\mathcal{C}'$. We consider the following three cases depending on whether $\underline{c}_{11}+\underline{c}_{00}+\overline{c}_{10}+\overline{c}_{01}$ is greater than, less than, or equal to $1$. 

\begin{enumerate}[label={(\alph*)}, topsep=1ex,itemsep=-0.3ex,partopsep=1ex,parsep=1ex]
    \item 
    {
        Consider the case where $\underline{c}_{11}+\underline{c}_{00}+\overline{c}_{10}+\overline{c}_{01} = 1$. 
        Let $(q^{(\text{a})}_{11}, q^{(\text{a})}_{00}, q^{(\text{a})}_{10}, q^{(\text{a})}_{01}) = (\underline{c}_{11}, \underline{c}_{00}, \overline{c}_{10}, \overline{c}_{01})$. Obviously, $q^{(\text{a})}_{**}\in \mathcal{C}'$. Moreover, 
        \begin{align*}
            \inf_{c_{**} \in \mathcal{C}'}
            \frac{c_{11} c_{00}}{c_{10}c_{01}}
            \ge 
            \frac{\underline{c}_{11}\underline{c}_{00}}{\overline{c}_{10}\overline{c}_{01}} = \frac{q_{11}^{(\text{a})} q_{00}^{(\text{a})}}{
            q_{10}^{(\text{a})} q_{01}^{(\text{a})} 
            }.
        \end{align*}
        Thus, the infimum of $c_{11} c_{00}/( c_{10}c_{01})$ over $c_{**} \in\mathcal{C}$ is achieved at $q_{**}^{(\text{a})}$. 
    }

    \item 
    {
        Consider the case where $\underline{c}_{11}+\underline{c}_{00}+\overline{c}_{10}+\overline{c}_{01} > 1$. 
        Note that, for any $c_{**}\in \mathcal{C}'$,
        \begin{align*}
            c_{10} + c_{01} = 1 - c_{11} - c_{00} 
            \le 1 - \underline{c}_{11} - \underline{c}_{00} < \overline{c}_{10} + \overline{c}_{01}.
        \end{align*}
        We then have $c_{**}\in \mathcal{C}'$ if and only if $c_{11} = \underline{c}_{11}$, $c_{00} = \underline{c}_{00}$, $c_{01} = 1 - \underline{c}_{11} - \underline{c}_{00} -c_{10}$, and 
        \begin{align*}
            \underline{c}_{10}'\equiv \max\left\{ \underline{c}_{10}, \ 1 - \underline{c}_{11} - \underline{c}_{00} - \overline{c}_{01} \right\}
            \le 
            c_{10} \le \min\left\{ \overline{c}_{10}, \  1 - \underline{c}_{11} - \underline{c}_{00} - \underline{c}_{01} \right\} \equiv \overline{c}_{10}',
        \end{align*}
        From the fact that $\underline{c}_{zy}\leq \overline{c}_{zy}$, $\sum_{z,y}\underline{c}_{zy}\le 1$, and $\underline{c}_{11}+\underline{c}_{00}+\overline{c}_{10}+\overline{c}_{01} > 1$
        in this case, 
        we can verify that $\underline{c}_{10}'\le \overline{c}_{10}'$, indicating that $\mathcal{C}'\ne \emptyset$. 
        Consequently, 
        \begin{align*}
            \inf_{c_{**} \in \mathcal{C}'}
            \frac{c_{11} c_{00}}{c_{10}c_{01}}
            = 
            \inf_{c_{10} \in [\underline{c}_{10}', \overline{c}_{10}']}
            \frac{\underline{c}_{11} \underline{c}_{00}}{c_{10}(1-\underline{c}_{11}-\underline{c}_{00}-c_{10})}
        \end{align*}
        Because $c_{10}(1-\underline{c}_{11}-\underline{c}_{00}-c_{10})$ is a concave function of $c_{10}$, its maximum on $[\underline{c}_{10}', \overline{c}_{10}']$ must be achieved at 
        \begin{align*}
            q_{10}^{(\text{b})} & = 
            \begin{cases}
                (1-\underline{c}_{11}-\underline{c}_{00})/2, & \text{if }
                \underline{c}_{10}' \le (1-\underline{c}_{11}-\underline{c}_{00})/2
                \le \overline{c}_{10}', \\
                \underline{c}_{10}', & \text{if } (1-\underline{c}_{11}-\underline{c}_{00})/2 < \underline{c}_{10}',\\
                \overline{c}_{10}', & \text{if } (1-\underline{c}_{11}-\underline{c}_{00})/2 > \overline{c}_{10}'.
            \end{cases}
            \\
            & = 
            \min\{ \max\{
            \underline{c}_{10}', 
            (1-\underline{c}_{11}-\underline{c}_{00})/2
            \}, \overline{c}_{10}' \}.
        \end{align*}
        Let $q^{(\text{b})}_{11} = \underline{c}_{11}, q^{(\text{b})}_{00} = \underline{c}_{00}$ and $q^{(\text{b})}_{01} = 1 - q^{(\text{b})}_{11} - q^{(\text{b})}_{00} - q^{(\text{b})}_{10}$. 
        From the discussion before, the infimum of $c_{11} c_{00}/( c_{10}c_{01})$ over $c_{**} \in\mathcal{C}$ is achieved at $q^{(\text{b})}_{**}\in \mathcal{C}'$. 
    }

    \item 
        Consider the case where $\underline{c}_{11}+\underline{c}_{00}+\overline{c}_{10}+\overline{c}_{01} < 1$. 
        Note that, for any $c_{**}\in \mathcal{C}'$,
        \begin{align*}
            c_{11} + c_{00} = 1 - c_{10} - c_{01} 
            \ge 1 - \overline{c}_{10} - \overline{c}_{01} > \underline{c}_{11}+\underline{c}_{00}. 
        \end{align*}

        It follows that $c_{**}\in \mathcal{C}'$ if and only if $c_{10} = \overline{c}_{10}$, $c_{01} = \overline{c}_{01}$, $c_{00} = 1 - \overline{c}_{10} - \overline{c}_{01} -c_{11}$, and 
        \begin{align*}
            \underline{c}_{11}' \equiv \max\left\{ \underline{c}_{11}, \ 1 - \overline{c}_{10} - \overline{c}_{01} - \overline{c}_{00} \right\}
            \le 
            c_{11} \le 
            \min\left\{ \overline{c}_{11}, \ 1 - \overline{c}_{10} - \overline{c}_{01} - \underline{c}_{00} \right\} \equiv \overline{c}_{11}',
        \end{align*}
        Again, from the fact that  $\underline{c}_{zy}\leq \overline{c}_{zy}$, $\sum_{z,y}\overline{c}_{zy}\ge 1$, and $\underline{c}_{11}+\underline{c}_{00}+\overline{c}_{10}+\overline{c}_{01} < 1$ in this case, we can verify that  $\underline{c}_{11}' \le \overline{c}_{11}'$, indicating that $\mathcal{C}'\ne \emptyset$. 
        Consequently, 
        \begin{align*}
            \inf_{c_{**} \in \mathcal{C}'}
            \frac{c_{11} c_{00}}{c_{10}c_{01}}
            & = 
            \inf_{c_{11} \in [\underline{c}_{11}', \overline{c}_{11}']}
            \frac{c_{11} (1 - \overline{c}_{10} - \overline{c}_{01} - c_{11})}{\overline{c}_{10} \overline{c}_{01}}. 
        \end{align*}
        Because $c_{11} (1 - \overline{c}_{10} - \overline{c}_{01} - c_{11})$ is a concave function of $c_{11}$, its minimum on $[\underline{c}_{11}', \overline{c}_{11}']$ must be achieved at the boundary. 
        From the discussion before, the infimum of $c_{11} c_{00}/( c_{10}c_{01})$ over $c_{**} \in\mathcal{C}$ must be achieved at $q^{(\text{c})}_{**}\in \mathcal{C}'$, with $q^{(\text{c})}_{10} = \overline{c}_{10}, q^{(\text{c})}_{01} = \overline{c}_{01}$, $q^{(\text{c})}_{00} = 1 - \overline{c}_{10} - \overline{c}_{01} - q_{11}$, and $q^{(\text{c})}_{11} = \underline{c}_{11}'$ or $\overline{c}_{11}'$. 
\end{enumerate}
From the above, to prove Lemma \ref{lemma:optim}, it suffices to verify that $q^{(\text{b})}_{**} = q^{(\text{a})}_{**}$ when $\underline{c}_{11}+\underline{c}_{00}+\overline{c}_{10}+\overline{c}_{01} = 1$. 
By definition, $q^{(\text{b})}_{11} = \underline{c}_{11} = q^{(\text{a})}_{11}, q^{(\text{b})}_{00} = \underline{c}_{00} = q^{(\text{a})}_{00}$, and 
\begin{align*}
    q_{10}^{(\text{b})} & = 
    \min\{ \max\{
    \underline{c}_{10}, \ 1 - \underline{c}_{11} - \underline{c}_{00} - \overline{c}_{01}, \ 
    (1-\underline{c}_{11}-\underline{c}_{00})/2
    \}, \ \overline{c}_{10}, \  1 - \underline{c}_{11} - \underline{c}_{00} - \underline{c}_{01} \}
    \\
    & = 
    \min\{ \max\{
    \underline{c}_{10}, \ \overline{c}_{10}, \ 
    (1-\underline{c}_{11}-\underline{c}_{00})/2
    \}, \ \overline{c}_{10}, \  \overline{c}_{10} +  \overline{c}_{01} - \underline{c}_{01} \}\\
    & = 
    \min\{ \max\{
    \overline{c}_{10}, \ 
    (1-\underline{c}_{11}-\underline{c}_{00})/2
    \}, \ \overline{c}_{10} \}
    = \overline{c}_{10}, 
\end{align*}
where the first equality holds by definition, the second equality holds due to the condition that $\underline{c}_{11}+\underline{c}_{00}+\overline{c}_{10}+\overline{c}_{01} = 1$, 
and the last equality can be easily verified. 
Therefore, Lemma \ref{lemma:optim} holds. 
\end{proof}

\begin{proof}[\bf Proof of Theorem \ref{thm:cor_bound_additional}]
From Lemmas \ref{lemma:useful} and \ref{lemma:cond_w_gamma_lu}--\ref{lemma:optim}, 
to prove Theorem \ref{thm:cor_bound_additional}, it suffices to verify that $\sum_{z,y} \tilde{l}_{zy} \leq 1$ and $\sum_{z,y} \tilde{u}_{zy}\geq 1$. 

We first prove that  $\sum_{z,y} \tilde{u}_{zy}\geq 1$. 
If any of the $\tilde{u}_{zy}$s takes value $1$, then this inequality holds automatically; 
otherwise, 
$$\sum_{z,y} \tilde{u}_{zy} = \sum_{z,y} \frac{\Gamma \pi_{zy}}{\delta + \Gamma(1-\delta)} = \frac{\Gamma}{\delta + \Gamma(1-\delta)} \geq 1,$$
where the last inequality holds since $\Gamma\ge 1$. 

We then prove that $\sum_{z,y} \tilde{l}_{zy} \leq 1$. 
Suppose that 
$\tilde{l}_{zy} = \pi_{zy}/\{\delta \Gamma + (1-\delta)\}$
for $(z,y)\in \mathcal{I}_1$,  
and $\tilde{l}_{zy} =  (\pi_{zy}-\delta)/(1-\delta)$ for $(z,y)\in \mathcal{I}_2$. If $|\mathcal{I}_2|\geq 1$, then 
\begin{align*}
    \sum_{z,y} \tilde{l}_{zy}
    & 
    = \frac{\sum_{ (z,y)\in \mathcal{I}_1  }\pi_{zy}}{\delta \Gamma + (1-\delta)} + \frac{\sum_{(z,y)\in \mathcal{I}_2} \pi_{zy} - |\mathcal{I}_2|\delta}{1-\delta}
    =
    \frac{\sum_{ (z,y)\in \mathcal{I}_1  }\pi_{zy}}{\delta \Gamma + (1-\delta)} + \frac{\sum_{(z,y)\in \mathcal{I}_2} \pi_{zy}}{1-\delta}
    - 
    \frac{|\mathcal{I}_2|\delta}{1-\delta}
    \\
    & = 
    \frac{\delta \Gamma \sum_{(z,y) \in \mathcal{I}_2} \pi_{zy}+(1-\delta) \sum_{z,y} \pi_{zy}}{\{\delta \Gamma+(1-\delta)\}(1-\delta)} - \frac{|\mathcal{I}_2|\delta}{1-\delta}
    \le \frac{ \delta \Gamma + (1-\delta)}{\{\delta \Gamma+(1-\delta)\}(1-\delta)}- \frac{|\mathcal{I}_2|\delta}{1-\delta}
    \\
    & = \frac{1-|\mathcal{I}_2|\delta}{1-\delta}\leq 1, 
\end{align*}
where the second last inequality holds because $\sum_{(z,y) \in \mathcal{I}_2} \pi_{zy}\le \sum_{z,y} \pi_{zy} = 1$; 
otherwise, $\mathcal{I}_2 = \emptyset$, and 
    $$\sum_{z,y} \tilde{l}_{zy} = \sum_{z,y} \frac{\pi_{zy}}{\delta \Gamma + (1-\delta)} = \frac{1}{\delta \Gamma + (1-\delta)} \leq 1.$$
From the above, Theorem \ref{thm:cor_bound_additional} holds. 
\end{proof}

\subsection{Proof of Theorem \ref{thm:sen_3_constr}}

\begin{proof}[\bf Proof of Theorem \ref{thm:sen_3_constr}]
From Lemma \ref{lemma:useful}, the discussion of its implication and Lemma \ref{lemma: COR p_xy}, to derive the sharp bounds of $\COR_0 = p_{11\mid 0} p_{00\mid 0}/(p_{10\mid 0} p_{01\mid 0})$, it suffices to search over all possible $w\equiv P(U=1\mid T=1)$ and $p_{zy\mid u}$s defined in \eqref{eq:pzy_u} that are coherent with the observed data distribution and satisfy Assumptions \ref{asmp:bound_u_not_0}, \ref{asmp:bound_density} and \ref{asmp:effect_heter}. 
These equivalently impose the following constraints on $w$ and $p_{zy\mid u}$s: 
\begin{align*}
    & 0\le w \le 1, \quad 0\le p_{zy\mid u} \le 1 \text{ for all } z,y,u\in \{0,1\}, & \text{(validity of probabilities)} \\
    & \pi_{zy} = p_{zy\mid 0}(1-w) + p_{zy\mid 1}w \text{ for all } z,y\in \{0,1\}, & \text{(coherent with observed data distribution)}
    \\
    & 0\le w \le \delta, & \text{(Assumption \ref{asmp:bound_u_not_0})}\\
    & 1/\Gamma \le p_{zy\mid 1}/p_{zy\mid 0} \le \Gamma \text{ for all } z,y\in \{0,1\}, & \text{(Assumption \ref{asmp:bound_density})}\\
    & 1/\xi \le \frac{p_{11\mid 0} p_{00\mid 0}}{p_{10\mid 0} p_{01\mid 0}}/\frac{p_{11\mid 1} p_{00\mid 1}}{p_{10\mid 1} p_{01\mid 1}} \le \xi, & \text{(Assumption \ref{asmp:effect_heter})}
\end{align*}
where the last equation uses Lemma \ref{lemma: COR p_xy}. 
Thus, the sharp bounds of $\COR_0$ under the proposed sensitivity analysis with Assumptions \ref{asmp:bound_u_not_0}, \ref{asmp:bound_density} and \ref{asmp:effect_heter} are the minimum and maximum values of $p_{11\mid 0} p_{00\mid 0}/(p_{10\mid 0} p_{01\mid 0})$ over all possible $w$ and  $p_{zy\mid u}$s subject to the above constraints. 
It is not difficult to verify that this optimization is equivalent to the quadratic programming problem in \eqref{eq:qp}--\eqref{eq:qp_objective}. 
Therefore, Theorem \ref{thm:sen_3_constr} holds. 
\end{proof}

\subsection{Proofs of Theorems \ref{thm:cat_ZYU} and \ref{thm:cat_ZYU_supp}}
Because 
Theorem \ref{thm:cat_ZYU} is a special case of Theorem \ref{thm:cat_ZYU_supp} when \(\xi=\infty\), 
we only prove Theorem \ref{thm:cat_ZYU_supp}. 
We need the following two lemmas. 

\begin{lemma}\label{lemma: COR p_xy_cat_ZY_U}
    Consider the case with categorical exposure, outcome and general unmeasured confounder. 
    Recall the definition of $\COR_u$ in \eqref{eq:COR} and \eqref{eq:COR_cat} and $\tilde{p}_{zy\mid u}$s in \eqref{eq:p_tilde}. 
    Under Assumption \ref{asmp:U}, we have
    \begin{equation*}
        \COR_{u} = \frac{\tilde{p}_{11\mid u}\tilde{p}_{00\mid u}}{\tilde{p}_{10\mid u}\tilde{p}_{01\mid u}}.
    \end{equation*}
\end{lemma}

\begin{proof}[\bf Proof of Lemma \ref{lemma: COR p_xy_cat_ZY_U}]
Following the same steps as the proof of Lemma \ref{lemma: COR p_xy}, we have  
    \begin{align}\label{eq:cat_ZY_bin_U_proof}
    \COR_{u} 
    & = 
    \frac{\P(Z=1, Y=1\mid T = 1, U = u)/\P(Z=0, Y = 1\mid T=1, U = u)}{\P(Z=1, Y=0\mid T=1, U = u)/\P(Z=0, Y=0 \mid T= 1, U = u)}.
\end{align}
Note that, for any $z,y\in \{0,1\}$, 
\begin{align*}
    \tilde{p}_{zy\mid u}
    & = 
    \P(Z=z, Y=y\mid U=u, \tilde{T}=1)\\
    & 
    = \P(Z=z, Y=y\mid U=u, Z\in \{0,1\}, Y\in \{0,1\}, T=1)
    \\
    & = 
    \frac{\P(Z=z, Y=y \mid  U=u,  T=1)}{\P( Z\in \{0,1\}, Y\in \{0,1\} \mid U=u,T=1)}.
\end{align*}
Dividing each term in \eqref{eq:cat_ZY_bin_U_proof} by $\P( Z\in \{0,1\}, Y\in \{0,1\} \mid U=u,T=1)$, we then have 
$\COR_u = \tilde{p}_{11\mid u}\tilde{p}_{00\mid u}/(\tilde{p}_{10\mid u}\tilde{p}_{01\mid u}).$
Therefore, Lemma \ref{lemma: COR p_xy_cat_ZY_U} holds. 
\end{proof}

\begin{proof}[\bf Proof of Theorem \ref{thm:cat_ZYU_supp}]
From Lemma \ref{lemma:useful}, its implication, and Lemma \ref{lemma: COR p_xy_cat_ZY_U}, 
to find the sharp bounds on $\COR_0 = \tilde{p}_{11\mid 0}\tilde{p}_{00\mid 0}/(\tilde{p}_{10\mid 0}\tilde{p}_{01\mid 0})$, it suffices to search over all possible conditional distributions of $U\mid \tilde{T}=1$ and $(Z, Y)\mid U, \tilde{T}=1$ that are coherent with the observed data distribution and satisfy Assumptions \ref{asmp:categorical_ZY_not_0}, \ref{asmp:categorical_ZY_density_ratio} and \ref{asmp:categorial_U_effect_heter}.

First, consider any specification for the conditional distributions of $U\mid \tilde{T}=1$ and $(Z, Y) \mid U, \tilde{T}=1$
for 
categorical $(U, Z, Y)$
that are coherent with the observed data distribution $\tilde{\pi}_{xy}$s and satisfy Assumptions \ref{asmp:categorical_ZY_not_0}, \ref{asmp:categorical_ZY_density_ratio}  and \ref{asmp:categorial_U_effect_heter}. 
Let $\tilde{\omega}$ be the probability measure induced by the conditional distribution of $U$ given $\tilde{T}=1$, 
and let 
$\tilde{p}_{zy\mid u} = \Pr(Z=z, Y=y \mid U=u, \tilde{T}=1)$ for all $z,y\in \{0,1\}$ and $u\in \mathcal{U}$. 
Define 
\begin{align}\label{eq:cat_to_bin}
    \check{w}_0 & \equiv \tilde{\omega}(\{0\}) = \Pr(U=0\mid \tilde{T}=1), \quad \check{w}_1 \equiv 1-\tilde{\omega}(\{0\}),  
    \nonumber
    \\
    \check{p}_{zy\mid 0} &\equiv \tilde{p}_{zy\mid 0}, 
    \quad 
    \check{p}_{zy\mid 1} \equiv 
    \begin{cases}
        \frac{\tilde{\pi}_{zy}-\tilde{p}_{zy\mid 0} \cdot \tilde{\omega}(\{0\}) }{1-\tilde{\omega}(\{0\})} & \tilde{\omega}(\{0\})<1\\
        \tilde{p}_{zy\mid 0} & \tilde{\omega}(\{0\})=1, 
    \end{cases}
    \qquad (z, y\in \{0,1\}). 
\end{align}
Below we show that $\check{w}_u$s and 
$\check{p}_{zy\mid u}$s are possible specifications for the conditional distributions of $U\mid T=1$ and $(Z, Y) \mid U, T=1$ with 
binary $(U, Z, Y)$
that are coherent with the observed data distribution $\pi_{zy}=\tilde{\pi}_{zy}$s and satisfy Assumptions \ref{asmp:bound_u_not_0}, \ref{asmp:bound_density} and \ref{asmp:effect_heter} with the same values of $(\delta, \Gamma, \xi)$. 
We consider only the case where $\tilde{\omega}(\{0\})< 1$. It is not difficult to verify that the following results hold when $\tilde{\omega}(\{0\}) = 1$. 
\begin{enumerate}[label={(\alph*)}, topsep=1ex,itemsep=-0.3ex,partopsep=1ex,parsep=1ex]
    \item By the construction in \eqref{eq:cat_to_bin}, we can verify that 
    $\tilde{\pi}_{zy} = \check{p}_{zy\mid 0}\check{w}_0 + \check{p}_{zy\mid 1} \check{w}_1$ for all $z,y\in \{0,1\}$. 
    Moreover, $\check{w}_u$s and $\check{p}_{zy\mid u}$s are all nonnegative, $\check{w}_1+\check{w}_0 = 1$, 
    and $\sum_{z=0}^1 \sum_{y=0}^1 \check{p}_{zy\mid u} = 1$ for $u=0,1$. These imply that  $\check{w}_u$s and 
    $\check{p}_{zy\mid u}$s are possible specifications for the conditional distributions of $U\mid T=1$ and $(Z, Y) \mid U, T=1$ coherent with the observed data distribution. 
    
    \item From Assumption \ref{asmp:categorical_ZY_not_0}, $\check{w}_1=1-\tilde{w}(\{0\}) \le \delta$. Thus, Assumption \ref{asmp:bound_u_not_0} holds. 

    \item From Assumption \ref{asmp:categorical_ZY_density_ratio} and by definition, for any $z,y \in \{0,1\}$, 
    \begin{align*}
        \frac{\check{p}_{zy\mid 1}}{\check{p}_{zy\mid 0}}
        & = 
        \frac{\tilde{\pi}_{zy} - \tilde{p}_{zy\mid 0} \tilde{w}(\{0\})}{\tilde{p}_{zy\mid 0} \{1-\tilde{w}(\{0\}) \} }
        = 
        \frac{\int_{\mathcal{U}\setminus \{0\}}\tilde{p}_{zy\mid u} \text{d}\tilde{w}(u)}{\tilde{p}_{zy\mid 0}\int_{\mathcal{U}\setminus \{0\}}\text{d}\tilde{w}(u)}
        = 
        \frac{\int_{\mathcal{U}\setminus \{0\}}\tilde{p}_{zy\mid u} \text{d}\tilde{w}(u)}{\int_{\mathcal{U}\setminus \{0\}} \tilde{p}_{zy\mid 0} \text{d}\tilde{w}(u)}
        \\
        & \in 
        \left[ \inf_{u\in \mathcal{U}\setminus \{0\}}
        \frac{\tilde{p}_{zy\mid u}}{\tilde{p}_{zy\mid 0}}, \ 
        \sup_{u\in \mathcal{U}\setminus \{0\}} \frac{\tilde{p}_{zy\mid u}}{\tilde{p}_{zy\mid 0}} \right]
        \subset [1/\Gamma, \Gamma]. 
    \end{align*}
    Thus, Assumption \ref{asmp:bound_density} holds.

    \item By definition and the law of total probability, 
    \begin{align*}%
    \tilde{p}_{zy\mid \ne 0}
    & = \Pr(Z=z, Y=y\mid U \ne 0, \tilde{T}=1) 
    \\
    & = 
    \frac{
    \Pr(Z=z, Y=y\mid \tilde{T}=1) 
    - 
    \Pr(U = 0 \mid \tilde{T}=1) \cdot \Pr(Z=z, Y=y\mid U = 0, \tilde{T}=1)
    }{\Pr(U\ne 0 \mid \tilde{T}=1)}
    \\
    & = \frac{\tilde{\pi}_{zy}-\tilde{w}(\{0\}) \cdot \tilde{p}_{zy\mid 0}}{1-\tilde{w}(\{0\})} = \check{p}_{zy\mid 1}.  
    \end{align*}
    From Assumption \ref{asmp:categorial_U_effect_heter}, we then have 
    \begin{align*}
        \frac{\check{p}_{11\mid 1} \check{p}_{00\mid 1}}{\check{p}_{10\mid 1} \check{p}_{01\mid 1}}/\frac{\check{p}_{11\mid 0} \check{p}_{00\mid 0}}{\check{p}_{10\mid 0} \check{p}_{01\mid 0}}
        & = 
        \frac{\tilde{p}_{11\mid \ne 0} \tilde{p}_{00\mid \ne 0}}{\tilde{p}_{10\mid \ne 0} \tilde{p}_{01\mid \ne 0}}/\frac{\tilde{p}_{11\mid 0} \tilde{p}_{00\mid 0}}{\tilde{p}_{10\mid 0} \tilde{p}_{01\mid 0}}
        \in [1/\xi, \xi]. 
    \end{align*}
    Thus, Assumption \ref{asmp:effect_heter} holds. 
\end{enumerate}
Moreover, 
from Lemmas \ref{lemma: COR p_xy} and \ref{lemma: COR p_xy_cat_ZY_U}, 
the specification for the binary case leads to the same causal effect $\check{p}_{11\mid 0} \check{p}_{00\mid 0}/(\check{p}_{10\mid 0} \check{p}_{01\mid 0}) = \tilde{p}_{11\mid 0} \tilde{p}_{00\mid 0}/(\tilde{p}_{10\mid 0} \tilde{p}_{01\mid 0})$ as the original one for the categorical case.

Second, consider any specification for the conditional distributions of $U\mid T=1$ and $(Z, Y) \mid U, T=1$
for binary $(U, Z, Y)$ that are coherent with the observed data distribution $\pi_{zy}$s and satisfy Assumptions \ref{asmp:bound_u_not_0}, \ref{asmp:bound_density} and \ref{asmp:effect_heter}. 
Let $\check{w}_u = \Pr(U=u\mid T=1)$ and 
$\check{p}_{zy\mid u} = \Pr(Z=z, Y=y \mid U=u, T=1)$ for all $z,y, u\in \{0,1\}$. 
Without loss of generality, we assume $1\in \mathcal{U}$, 
define the following probability measure $\tilde{\omega}$ on $\mathcal{U}$: 
\begin{align*}
    \tilde{\omega}(\{u\}) = \check{w}_u \text{ for } u \in \{0,1\}, 
    \quad \text{and} \quad 
    \tilde{\omega}(\mathcal{U}\setminus\{0,1\}) = 0,
\end{align*}
and define further 
\begin{align*}
    & \tilde{p}_{zy\mid u} = \check{p}_{zy\mid u} \text{ for } u=0,1, \quad \text{and} \quad  \tilde{p}_{zy\mid u} = \check{p}_{zy\mid 0} \text{ for } u \in \mathcal{U}\setminus\{0,1\}. 
\end{align*}
We can straightforwardly verify that $\tilde{w}$ and 
$\tilde{p}_{zy\mid u}$s are possible specifications for the conditional distributions of $U\mid \tilde{T}=1$ and $(Z, Y) \mid U, \tilde{T}=1$ with 
categorical $(U, Z, Y)$
that are coherent with the observed data distribution $\tilde{\pi}_{zy} = \pi_{zy}$s and satisfy Assumptions \ref{asmp:categorical_ZY_not_0}, \ref{asmp:categorical_ZY_density_ratio}
and \ref{asmp:categorial_U_effect_heter} with the same values of $(\delta, \Gamma, \xi)$. 
Moreover,  
from Lemmas \ref{lemma: COR p_xy} and \ref{lemma: COR p_xy_cat_ZY_U}, 
the specification for the categorical case leads to the same causal effect $\tilde{p}_{11\mid 0} \tilde{p}_{00\mid 0}/(\tilde{p}_{10\mid 0} \tilde{p}_{01\mid 0}) = \check{p}_{11\mid 0} \check{p}_{00\mid 0}/(\check{p}_{10\mid 0} \check{p}_{01\mid 0})$ as the original one for the binary case.

From the above and Lemma \ref{lemma:useful}, we can then derive Theorem \ref{thm:cat_ZYU_supp}. 
\end{proof}

\subsection{Proofs of Theorems \ref{thm:confidence_bound} and \ref{thm:cb_sen_3_constr}}

\begin{proof}[\bf Proof of Theorem \ref{thm:confidence_bound}]
    Recall the definition of \(\tilde{l}_{zy}\)s and \(\tilde{u}_{zy}\)s in Theorem \ref{thm:cor_bound_additional} and the definition of \(\hat{l}_{zy}\)s and \(\hat{u}_{zy}\)s in Theorem \ref{thm:confidence_bound}. 
    Suppose that $\hat{\underline{\pi}}_{zy} \leq \pi_{zy} \leq \hat{\overline{\pi}}_{zy}$
    for all $z,y\in \{0,1\}$. 
    We can then verify that $\hat{l}_{zy} \leq \tilde{l}_{zy}$ and $\tilde{u}_{zy} \leq \hat{u}_{zy}$ for all $z,y\in \{0,1\}$. 
    Consequently, 
    \begin{align*}
        \mathcal{B} & \equiv \Big\{ p_{**\mid 0}:  
        \tilde{l}_{zy}
        \le p_{zy\mid 0} \le 
        \tilde{u}_{zy} \text{ for } z,y\in\{0,1\}, 
        \text{ and }
        \sum_{z,y} p_{zy\mid 0} = 1
        \Big\}.
    \end{align*}
    is a subset of 
    \begin{align*}
        \widehat{\mathcal{B}} & \equiv \Big\{ p_{**\mid 0}:  
        \hat{l}_{zy}
        \le p_{zy\mid 0} \le 
        \hat{u}_{zy} \text{ for } z,y\in \{0,1\}, 
        \text{ and }
        \sum_{z,y} p_{zy\mid 0} = 1
        \Big\}.
    \end{align*}
From Theorem \ref{thm:cor_bound_additional} and Lemma \ref{lemma:optim}, this implies that 
\begin{equation*}
    \underline{\COR}_0 =  \inf_{p_{**\mid 0}\in \mathcal{B}} \frac{p_{11\mid 0} p_{00\mid 0}}{p_{10\mid 0}p_{01\mid 0}}
    \ge 
    \inf_{p_{**\mid 0}\in \widehat{\mathcal{B}}} \frac{p_{11\mid 0} p_{00\mid 0}}{p_{10\mid 0}p_{01\mid 0}} = 
    \widehat{\underline{\COR}}_0. 
\end{equation*}
By the same logic, \( \widehat{\overline{\COR}}_0\geq \overline{\COR}_0\). 

From the above, as long as $\hat{\underline{\pi}}_{zy} \leq \pi_{zy} \leq \hat{\overline{\pi}}_{zy}$
    for all $z,y\in \{0,1\}$, 
$[\underline{\COR}_0, \overline{\COR}_0]$ will be a subset of $[\widehat{\underline{\COR}}_0, \widehat{\overline{\COR}}_0]$. 
Thus, 
\begin{align*}
    \P(\COR_0 \in [\widehat{\underline{\COR}}_0, \widehat{\overline{\COR}}_0]) 
    & \ge 
    \P([\underline{\COR}_0, \overline{\COR}_0] \subset [\widehat{\underline{\COR}}_0, \widehat{\overline{\COR}}_0])
    \\
    & \ge \Pr(\hat{\underline{\pi}}_{zy} \leq \pi_{zy} \leq \hat{\overline{\pi}}_{zy} \text{ for all } z,y\in \{0,1\})
    = \Pr(\pi_{**} \in \mathcal{S}_{\alpha}).
\end{align*}
Theorem \ref{thm:confidence_bound} then follows by the asymptotic validity of the confidence intervals for the true $\pi_{zy}$s; see Section \ref{sec:cs_o}. 
\end{proof}

\begin{proof}[\bf Proof of Theorem \ref{thm:cb_sen_3_constr}]
From Theorems \ref{thm:sen_3_constr} and \ref{thm:cb_sen_3_constr}, we can know that, as long as $\mathcal{S}_{\alpha}$ contains the true $\pi_{zy}$s, 
$[\widehat{\underline{\COR}}_0, \widehat{\overline{\COR}}_0]$ will contain the sharp bounds on the true $\COR_0$. 
Thus, 
\begin{align*}
    \P(\COR_0 \in [\widehat{\underline{\COR}}_0, \widehat{\overline{\COR}}_0])  \ge \Pr(\pi_{**} \in \mathcal{S}_{\alpha} ).
\end{align*}
    Theorem \ref{thm:cb_sen_3_constr} then follows from the asymptotic validity of the confidence set $\mathcal{C}_{\alpha}^{*}$ for the true $\pi_{**}$; see Section \ref{sec:cs_o}. 
\end{proof}

\begin{proof}[\bf Technical details for Remark  \ref{rmk:simultaneous1}]
In either Theorems \ref{thm:confidence_bound} or \ref{thm:cb_sen_3_constr}, 
if $\mathcal{S}_{c, \alpha}$ is a simultaneous asymptotic $1-\alpha$ confidence set for the true $\pi_{**\mid c}$ over all $c\in \mathcal{C}$, then we have 
\begin{align*}
    \liminf_{n\rightarrow \infty}\P(\COR_{0c} \in [\widehat{\underline{\COR}}_{0c}, \widehat{\overline{\COR}}_{0c}] \text{ for all } c \in \mathcal{C})  \ge \liminf_{n\rightarrow \infty} \Pr(\pi_{**\mid c} \in \mathcal{S}_{c, \alpha} \text{ for all } c \in \mathcal{C}) \ge 1-\alpha, 
\end{align*}
where we make the conditioning on $c$ explicit. 
Therefore, $[\widehat{\underline{\COR}}_{0c}, \widehat{\overline{\COR}}_{0c}]$ is a simultaneous asymptotic $1-\alpha$ confidence interval for the true causal odds ratio $\COR_{0c}$ over all $c\in \mathcal{C}$. 
\end{proof}

{\rev 
\subsection{Proof of Theorem \ref{thm:uniform_CI_model_pi}}
In this subsection we first prove the following proposition, which implies Theorem \ref{thm:uniform_CI_model_pi}. 
\begin{proposition}\label{prop:uniform_CI_model}
        Let 
        \(\hat{\beta}_n\) be an estimator of \(\beta_0\in \mathbb{R}^m\) based on samples of size $n$ satisfying that 
        $
        \sqrt{n}(\hat{\beta}_n - \beta_0) 
        \converged
        \mathcal{N}(0,\Sigma)
        $
        for some positive definite matrix \(\Sigma\), 
        and let $\hat{\Sigma}_n$ be a consistent estimator of $\Sigma$, in the sense that $\hat{\Sigma}_n - \Sigma = o_{\P}(1)$. 
        Let $\mathcal{C}$ be a fixed set, and \(g = (g_1, g_2, \ldots, g_r)^\top  : \mathbb{R}^m \times \mathcal{C} \to \mathbb{R}^r\) be a function. 
        Assume that the following regularity conditions hold: 
        \begin{itemize}
            \item[(i)] $g_i(\beta, c)$ is twice differentiable in \(\beta\in \mathbb{R}^m\) for all \(c\in \mathcal{C}\) and $1\le i \le r$, 
            
            \item[(ii)] the nonzero singular values of $D_c \equiv \frac{\partial g(\beta, c)}{\partial \beta} \mid_{\beta = \beta_0} \in \mathbb{R}^{r\times m}$ are uniformly bounded above and below for all $c\in \mathcal{C}$, 

            \item[(iii)] for a sufficiently small $\eta > 0$, the largest singular value of $\frac{\partial^2 g_i(\beta, c)}{\partial \beta \partial \beta^\top }$ is uniformly bounded for all $c\in \mathcal{C}$ and $\beta \in \mathbb{R}^m$ such that $\|\beta - \beta_0\|\le \eta$. 
        \end{itemize}
        Then, for any \(\alpha\in (0,1)\),
        \begin{equation*}
        \liminf_{n\rightarrow \infty}\Pr \left(\sup_{c\in \mathcal{C}}n \{ g(\hat{\beta}_n, c) - g(\beta_0, c) \}^\top \left(D_c \hat{\Sigma}_n D_c^\top  \right)^{\dagger} \{ g(\hat{\beta}_n, c) - g(\beta_0, c)\} \leq \chi^2_{m, 1-\alpha}\right) \geq 1-\alpha, 
        \end{equation*}
        where 
        \(\chi^2_{m, 1-\alpha}\) denotes the \((1-\alpha)\)th quantile of the chi-squared distribution with degrees of freedom $p$, and $A^\dagger$ denotes the pseudoinverse of a matrix $A$.   
\end{proposition}

To prove Proposition \ref{prop:uniform_CI_model}, we need the following three lemmas. 

\begin{lemma}\label{lemma:pesudo_inverse_ABA_norm}
    Suppose \(A\in \mathbb{R}^{n\times m}\) is a non-zero matrix and \(B\in \mathbb{R}^{n\times n}\) is a positive definite matrix. Then
\begin{equation*}
    ||(A^\top BA)^\dagger||\leq \{ \sigma_{\min}^+(A) \}^{-2} \{\lambda_{\min}(B)\}^{-1}
\end{equation*}
where \( \sigma_{\min}^+(A)\) is the smallest non-zero singular value of \(A\) and \(\lambda_{\min}(B)\) is the smallest eigenvalue of \(B\).
\end{lemma}

\begin{proof}[\bf Proof of Lemma \ref{lemma:pesudo_inverse_ABA_norm}]
By definition, 
$
    ||(A^\top BA)^\dagger||  = 1/\lambda_{\min}^+(A^\top BA), 
$
where \(\lambda_{\min}^+(A^\top BA)\) denotes the smallest non-zero eigenvalue of $A^\top BA$. 
Let $\Ker(A) = \{x\in \mathbb{R}^m: Ax = 0\}$ be the null space of $A$, and define analogously $\Ker(A^\top BA)$. Because $B$ is positive definite, we know that $\Ker(A^\top BA) = \Ker(A)$. 
For any $x \bot \Ker(A^\top BA)$, we then have 
\begin{align*}
    \frac{x^\top A^\top BA x}{x^\top x}
    & \ge \lambda_{\min}(B) 
    \frac{x^\top A^\top A x}{x^\top x}
    \ge \lambda_{\min}(B) \lambda_{\min}^{+}(A^\top A)
    = \lambda_{\min}(B)  \{\sigma_{\min}^+(A)\}^2, 
\end{align*}
where $\lambda_{\min}^{+}(A^\top A)$ denotes the smallest nonzero eigenvalue of $A^\top A$. 
This immediately implies that 
$\lambda_{\min}^+(A^\top BA) \ge \lambda_{\min}(B)  \{\sigma_{\min}^+(A)\}^2$.
Therefore, Lemma \ref{lemma:pesudo_inverse_ABA_norm} holds. 
\end{proof}

\begin{lemma}\label{lemma:diff_pseudo_inv}
For any matrix $A$ and $B$ with $B=A+E$, 
\begin{align*}
    \| B^\dagger - A^\dagger \| 
    \le 
    \frac{1+\sqrt{5}}{2} \cdot 
    \max\{ \|A^\dagger\|_2^2, \|B^\dagger\|_2^2 \} \cdot \|E\|. 
\end{align*}
\end{lemma}

\begin{proof}[\bf Proof of Lemma \ref{lemma:diff_pseudo_inv}]
    Lemma \ref{lemma:diff_pseudo_inv} is from \citet[][Theorem 3.3]{Stewart1977}. 
\end{proof}

\begin{lemma}\label{lemma:pesudo_inverse_ABA_norm_diff}
    Suppose \(A\in \mathbb{R}^{n\times m}\) is a non-zero matrix and \(B, \tilde{B}\in \mathbb{R}^{n\times n}\) are two positive definite matrices. Then
\begin{equation*}
    ||(A^\top \tilde{B}A)^\dagger - (A^\top BA)^\dagger||\leq \frac{1+\sqrt{5}}{2} 
        \frac{\|A\|^2}{\{ \sigma_{\min}^+(A) \}^{4}}
        \max\left\{
             \{\lambda_{\min}(\tilde{B})\}^{-2}, 
            \{\lambda_{\min}(B)\}^{-2}
        \right\}
        \| \tilde{B} - B \| . 
\end{equation*}
where \( \sigma_{\min}^+(A)\), \(\lambda_{\min}(B)\) and \(\lambda_{\min}(\tilde{B})\) are defined analogously as in Lemma \ref{lemma:pesudo_inverse_ABA_norm}
\end{lemma}

\begin{proof}[\bf Proof of Lemma \ref{lemma:pesudo_inverse_ABA_norm_diff}]
    From Lemma \ref{lemma:diff_pseudo_inv}, we have 
    \begin{align*}
        ||(A^\top \tilde{B}A)^\dagger - (A^\top BA)^\dagger||\leq \frac{1+\sqrt{5}}{2} 
        \max\left\{
            \| (A^\top \tilde{B}A)^\dagger \|^2, 
            \| (A^\top BA)^\dagger \|^2
        \right\}
        \| A^\top (\tilde{B} - B) A \|.
    \end{align*}
    From Lemma \ref{lemma:pesudo_inverse_ABA_norm},  
    \begin{align*}
        ||(A^\top \tilde{B}A)^\dagger||\leq \{ \sigma_{\min}^+(A) \}^{-2} \{\lambda_{\min}(\tilde{B})\}^{-1}, 
        \quad 
        ||(A^\top BA)^\dagger||\leq \{ \sigma_{\min}^+(A) \}^{-2} \{\lambda_{\min}(B)\}^{-1}.
    \end{align*}
    We then have 
    \begin{align*}
        & \quad \ ||(A^\top \tilde{B}A)^\dagger - (A^\top BA)^\dagger||
        \\
        & \leq \frac{1+\sqrt{5}}{2} 
        \max\left\{
            \{ \sigma_{\min}^+(A) \}^{-4} \{\lambda_{\min}(\tilde{B})\}^{-2}, 
            \{ \sigma_{\min}^+(A) \}^{-4} \{\lambda_{\min}(B)\}^{-2}
        \right\}
        \| A \|  \| \tilde{B} - B \|  \| A \|
        \\
        & \le 
        \frac{1+\sqrt{5}}{2} 
        \frac{\|A\|^2}{\{ \sigma_{\min}^+(A) \}^{4}}
        \max\left\{
             \{\lambda_{\min}(\tilde{B})\}^{-2}, 
            \{\lambda_{\min}(B)\}^{-2}
        \right\}
        \| \tilde{B} - B \| . 
    \end{align*}
    Therefore, Lemma \ref{lemma:pesudo_inverse_ABA_norm_diff} holds.  
\end{proof}

\begin{proof}[\bf Proof of Proposition \ref{prop:uniform_CI_model}
 ]
For any \(c\in \mathcal{C}\) and $1\le i \le r$, 
by Taylor's expansion, we have
\begin{align*}%
    g_i(\hat{\beta}_n, c) - g_i(\beta_0, c) & = \frac{\partial g_i(\beta, c)}{\partial \beta} 
    \mid_{\beta = \beta_0} (\hat{\beta}_{n}- \beta_0) + 
    \frac{1}{2} 
    (\hat{\beta}_{n}- \beta_0)^\top 
    \frac{\partial^2 g_i(\beta, c)}{\partial \beta \partial \beta^\top }
    \mid_{\beta = \beta^{(i)}_n}
    (\hat{\beta}_{n}- \beta_0)\\
    & = \frac{\partial g_i(\beta, c)}{\partial \beta} 
    \mid_{\beta = \beta_0} (\hat{\beta}_{n}- \beta_0) + 
    M_{ni}^{(c)}, 
\end{align*}
where $\beta^{(i)}_n$ is between $\hat{\beta}_{n}$ and $\beta_0$, 
and 
\begin{align*}
    M_{ni}^{(c)} \equiv 
    \frac{1}{2} 
    (\hat{\beta}_{n}- \beta_0)^\top 
    \frac{\partial^2 g_i(\beta, c)}{\partial \beta \partial \beta^\top }
    \mid_{\beta = \beta^{(i)}_n}
    (\hat{\beta}_{n}- \beta_0), 
    \quad (1\le i \le r). 
\end{align*}
Let $M_n^{(c)} = (M_{n1}^{(c)}, M_{n2}^{(c)}, \ldots, M_{nr}^{(c)})^\top$. 
We then have 
\begin{align}\label{eq:taylor_exp}
    g(\hat{\beta}_n, c) - g(\beta_0, c) & = D_c(\hat{\beta}_{n}- \beta_0) + 
    M_n^{(c)}.  
\end{align}
We further introduce $\Delta_c = (D_c \hat{\Sigma}_n D_c^\top)^\dagger - (D_c \Sigma D_c^\top)^\dagger$. 
Consequently, we have 
\begin{align*}
    & \quad \ \{ g(\hat{\beta}_n, c) - g(\beta_0, c) \}^\top \left(D_c \hat{\Sigma}_n D_c^\top  \right)^{\dagger} \{ g(\hat{\beta}_n, c) - g(\beta_0, c)\}
    \\
    & = 
    \{ D_c(\hat{\beta}_{n}- \beta_0) + 
    M_n^{(c)} \}^\top 
    \left\{  (D_c \Sigma D_c^\top)^\dagger + \Delta_c \right\}
    \{ D_c(\hat{\beta}_{n}- \beta_0) + 
    M_n^{(c)}\}\\
    & =  
(\hat{\beta}_{n}- \beta_0)^\top D_c^\top \left(D_c \Sigma D_c^\top  \right)^{\dagger} 
D_c(\hat{\beta}_{n}- \beta_0) 
+ 
\{M_n^{(c)}\}^\top \left(D_c \Sigma D_c^\top  \right)^{\dagger} M_n^{(c)}\\
& \quad \ 
+ 
\{ 
    M_n^{(c)}
\}^\top \left(D_c \Sigma D_c^\top  \right)^{\dagger} \{ 
D_c(\hat{\beta}_{n}- \beta_0) \}
+ 
\{ 
D_c(\hat{\beta}_{n}- \beta_0)
\}^\top \left(D_c \Sigma D_c^\top  \right)^{\dagger} 
    M_n^{(c)} \\
& \quad \   + 
\{ D_c(\hat{\beta}_{n}- \beta_0) + 
    M_n^{(c)}  \}^\top  \Delta_c \{ D_c(\hat{\beta}_{n}- \beta_0) + 
    M_n^{(c)} \}
\\
& = Q_1(c) + Q_2(c) + Q_3(c) + Q_3(c)^\top + Q_4 (c), 
\end{align*}
where 
\begin{align*}
    Q_1(c) & \equiv (\hat{\beta}_{n}- \beta_0)^\top D_c^\top \left(D_c \Sigma D_c^\top  \right)^{\dagger} 
D_c(\hat{\beta}_{n}- \beta_0), \quad 
Q_2(c)  \equiv \{M_n^{(c)}\}^\top \left(D_c \Sigma D_c^\top  \right)^{\dagger} M_n^{(c)}, \\
Q_3(c) & \equiv \{ 
    M_n^{(c)}
\}^\top \left(D_c \Sigma D_c^\top  \right)^{\dagger} \{ 
D_c(\hat{\beta}_{n}- \beta_0) \}, \\
Q_4(c) & \equiv \{ D_c(\hat{\beta}_{n}- \beta_0) + 
    M_n^{(c)}  \}^\top  \Delta_c \{ D_c(\hat{\beta}_{n}- \beta_0) + 
    M_n^{(c)} \}.
\end{align*}
Below we will first bound $Q_1(c)$, $Q_2(c)$, $Q_3(c)$ and $Q_4(c)$, and then prove the theorem.

First, we prove that 
$\sup_{c\in \mathcal{C}} \|M_{n}^{(c)}\| = \| \hat{\beta}_{n}- \beta_0 \|^2 \cdot O_{\Pr}(1)$. 
By definition, we have  
\begin{align*}
    \sup_{1\le i\le r, c\in \mathcal{C}} |M_{ni}^{(c)}| & \le 
    \frac{1}{2} 
    \| \hat{\beta}_{n}- \beta_0 \|^2 
    \cdot 
    \sup_{1\le i\le r, c\in \mathcal{C}}
    \left\| \frac{\partial^2 g_i(\beta, c)}{\partial \beta \partial \beta^\top }  \mid_{\beta = \beta^{(i)}_n} \right\|
    \\
    & \le 
    \| \hat{\beta}_{n}- \beta_0 \|^2 
    \cdot 
    \frac{1}{2}  \sup_{1\le i\le r, c\in \mathcal{C}}
    \sup_{\beta:\|\beta-\beta_0\| \le \|\hat{\beta}_n-\beta_0\|}\left\| \frac{\partial^2 g_i(\beta, c)}{\partial \beta \partial \beta^\top }   \right\|.
\end{align*}
Note that 
\begin{align*}
    & \quad \ \sup_{1\le i\le r, c\in \mathcal{C}}
    \sup_{\beta:\|\beta-\beta_0\| \le \|\hat{\beta}_n-\beta_0\|}\left\| \frac{\partial^2 g_i(\beta, c)}{\partial \beta \partial \beta^\top }   \right\|
    \\
    & \le 
    \sup_{1\le i\le r, c\in \mathcal{C}}
    \sup_{\beta:\|\beta-\beta_0\| \le \eta}\left\| \frac{\partial^2 g_i(\beta, c)}{\partial \beta \partial \beta^\top }   \right\|
    + 
    \I(\|\hat{\beta}_n - \beta_0\| > \eta) 
    \cdot 
    \sup_{1\le i\le r, c\in \mathcal{C}}
    \sup_{\beta:\|\beta-\beta_0\| \le \|\hat{\beta}_n-\beta_0\|}\left\| \frac{\partial^2 g_i(\beta, c)}{\partial \beta \partial \beta^\top }   \right\|\\
    & = O(1) + o_{\Pr}(1) =  O_{\Pr}(1), 
\end{align*}
where the second last equality follows from condition (iii) in the theorem and the fact that $\|\hat{\beta}_n-\beta_0\| = o_{\Pr}(1)$. 
Consequently, we must have  $\sup_{1\le i\le r, c\in \mathcal{C}} |M_{ni}^{(c)}| = \| \hat{\beta}_{n}- \beta_0 \|^2 \cdot O_{\Pr}(1)$.
This immediate implies that $\sup_{c\in \mathcal{C}} \|M_{n}^{(c)}\| = \| \hat{\beta}_{n}- \beta_0 \|^2 \cdot O_{\Pr}(1)$. 

Second, we prove that $\sup_{c\in \mathcal{C}}\| \left(D_c \Sigma D_c^\top  \right)^{\dagger} \| = O(1)$. 
This follows immediately from Lemma \ref{lemma:pesudo_inverse_ABA_norm} and condition (ii) in the theorem.

Third, we prove that $\sup_{c\in \mathcal{C}} \Delta_c = o_{\P}(1)$. 
From Lemma \ref{lemma:pesudo_inverse_ABA_norm_diff}, 
\begin{align*}
    \sup_{c\in \mathcal{C}} \|\Delta_c\| & 
    \leq \frac{1+\sqrt{5}}{2} 
        \sup_{c\in \mathcal{C}} \frac{\|D_c\|^2}{\{ \sigma_{\min}^+(D_c) \}^{4}}
        \cdot
        \max\left\{
             \{\lambda_{\min}(\hat{\Sigma}_n)\}^{-2}, 
            \{\lambda_{\min}(\Sigma)\}^{-2}
        \right\}
        \| \hat{\Sigma}_n- \Sigma \| 
        = o_{\P}(1), 
\end{align*}
where the last equality follows from condition (ii) in the theorem and the condition that $\hat{\Sigma}_n- \Sigma = o_{\P}(1)$. 

Fourth, by the property of pseudo inverse, 
$\Sigma^{1/2} D_c^\top \left(D_c \Sigma D_c^\top  \right)^{\dagger} 
D_c \Sigma^{1/2}$ is a projection matrix, where $\Sigma^{1/2}$ is the positive definite square root of $\Sigma$. 
Let $\Sigma^{-1/2}$ be the inverse of $\Sigma^{1/2}$. We then have, for any $c\in \mathcal{C}$, 
\begin{align*}
    Q_1(c) & = (\hat{\beta}_{n}- \beta_0)^\top \Sigma^{-1/2} \Sigma^{1/2} D_c^\top \left(D_c \Sigma D_c^\top  \right)^{\dagger} 
    D_c \Sigma^{1/2} \Sigma^{-1/2}(\hat{\beta}_{n}- \beta_0)
    \le (\hat{\beta}_{n}- \beta_0)^\top \Sigma^{-1} (\hat{\beta}_{n}- \beta_0). 
\end{align*}

Fifth, from the above first to third parts, 
and using the facts that $\sup_{c\in \mathcal{C}}
    \|
    D_c\|  = O(1)$ and $\| \hat{\beta}_{n}- \beta_0 \| = O_{\Pr}(n^{-1/2})$, 
we have 
\begin{align*}
    \sup_{c\in \mathcal{C}}|Q_2(c)|
    & \le 
    \sup_{c\in \mathcal{C}} \| 
    M_n^{(c)}\|  
    \cdot
    \sup_{c\in \mathcal{C}} \| \left(D_c \Sigma D_c^\top  \right)^{\dagger} \| 
    \cdot
    \sup_{c\in \mathcal{C}} \| M_n^{(c)}\|
    = \| \hat{\beta}_{n}- \beta_0 \|^4 \cdot O_{\Pr}(1) = O_{\P}(n^{-2}),\\
    \sup_{c\in \mathcal{C}}|Q_3(c)|
    & \le 
    \sup_{c\in \mathcal{C}} \| M_n^{(c)} \|
    \cdot
    \sup_{c\in \mathcal{C}}\| \left(D_c \Sigma D_c^\top  \right)^{\dagger} \| 
    \cdot 
    \sup_{c\in \mathcal{C}}
    \|
    D_c\| 
    \cdot \| \hat{\beta}_{n}- \beta_0 \|
    = 
    O_{\Pr}(\| \hat{\beta}_{n}- \beta_0 \|^3)= O_{\P}(n^{-3/2}),\\
    \sup_{c\in \mathcal{C}}|Q_4(c)|
    & \le \left( \sup_{c\in \mathcal{C}} \| D_c\| \| \hat{\beta}_{n}- \beta_0\| + 
    \sup_{c\in \mathcal{C}} \|M_n^{(c)}\| \right)^2  \| \sup_{c\in \mathcal{C}} \Delta_c \|
    = 
    \| \hat{\beta}_{n}- \beta_0\| ^2 \cdot o_{\P}(1)
    = o_{\P}(n^{-1}). 
    \end{align*}

From the above, we have  
\begin{align*}
    & \quad \ \sup_{c\in \mathcal{C}}n\{ g(\hat{\beta}_n, c) - g(\beta_0, c) \}^\top \left(D_c \Sigma D_c^\top  \right)^{\dagger} \{ g(\hat{\beta}_n, c) - g(\beta_0, c)\}
    \\
    & \le n \sup_{c\in \mathcal{C}} Q_1(c)
    + 
    n \sup_{c\in \mathcal{C}} Q_2(c) + n \sup_{c\in \mathcal{C}} Q_3(c) + n \sup_{c\in \mathcal{C}} Q_3(c)^\top + n \sup_{c\in \mathcal{C}} Q_4(c)\\
    & \le 
    n (\hat{\beta}_{n}- \beta_0)^\top \Sigma^{-1} (\hat{\beta}_{n}- \beta_0)
    + o_{\Pr}(1)\\
    & \converged \chi^2_m,
\end{align*}
where $\chi^2_m$ denotes a chi-squared random variable with degrees of freedom $m$, 
and the last convergence follows from the asymptotic distribution of $\hat{\beta}_{n}$ and Slutsky's theorem. 
This then implies that 
\begin{align*}
    & \quad \ \liminf_{n\rightarrow \infty}\Pr \left(\sup_{c\in \mathcal{C}}n\{ g(\hat{\beta}_n, c) - g(\beta_0, c) \}^\top \left(D_c \hat{\Sigma}_n D_c^\top  \right)^{\dagger} \{ g(\hat{\beta}_n, c) - g(\beta_0, c)\} \leq q_\alpha\right) \\
    & \geq 
    \Pr \left( \chi^2_m \leq \chi^2_{m, 1-\alpha}\right)
    =
    1-\alpha.
\end{align*}
Therefore, Proposition \ref{prop:uniform_CI_model} holds. 
\end{proof}

\begin{proof}[\bf Proof of Theorem \ref{thm:uniform_CI_model_pi}]
Note that, when 
\begin{align*}
    \sup_{c\in \mathcal{C}} n \{ g(\hat{\beta}_n, c) - \pi_{**\mid c} \}^\top \left(D_c \hat{\Sigma}_n D_c^\top  \right)^{\dagger} \{ g(\hat{\beta}_n, c) - \pi_{**\mid c}\} \leq \chi^2_{m, 1-\alpha}, 
\end{align*}
we must have 
$\pi_{**\mid c} \in \mathcal{S}_{c, \alpha}^g$ for all $c \in \mathcal{C}$. 
Theorem \ref{thm:uniform_CI_model_pi} then follows from Proposition \ref{prop:uniform_CI_model}. 
\end{proof}

\subsection{Proof of Proposition \ref{prop:multinom_cond}}

To prove Proposition \ref{prop:multinom_cond}, we need the following three lemmas. 

\begin{lemma}\label{lemma:eigen_M}
    Consider any $J\ge 2$ and any $g = (g_1, \ldots, g_J)^\top \in \mathbb{R}^J$ such that 
    $g_j \ge 0$ for all $j$ and $\sum_{j=1}^J g_j = 1$. 
    Let $M= \text{diag}(g) - g g^\top$. 
    Let $\lambda_2(M)$ and $\lambda_{\max}(M)$ be the second smallest and the large eigenvalues of $M$. 
    We then have 
    \begin{align*}
        \min_{1\le j\le J} g_j \le \lambda_2(M) 
        \le 
        \lambda_{\max}(M) \le \max_{1\le j \le J} g_j. 
    \end{align*}
\end{lemma}

\begin{proof}[\bf Proof of Lemma \ref{lemma:eigen_M}]
The upper bound on $\lambda_{\max}(M)$ follows immediately from
    \begin{align*}
        \lambda_{\max}(M) \le \lambda_{\max}\left( \text{diag}(g) \right) \le g_{\max}. 
    \end{align*}
Below we focus on the lower bound of $\lambda_2(M)$. 
Let $g_{\min} = \min_{1\le j\le J} g_j$. 
Note that $1_J = (1, \ldots, 1)^\top$ is an eigenvector of $M$ corresponding to eigenvalue zero. 
We can know that 
\begin{align*}
    \lambda_2(M)
    & = 
    \min_{x\ne 0: x\bot 1_J} 
    \frac{x^\top M x}{x^\top x}. 
\end{align*}
Let $h_j = g_j - g_{\min} \ge 0$ for $1\le j \le J$. 
For any $x\bot 1_J$, we then have $\sum_{j=1}^J x_j=0$, and consequently 
\begin{align*}
    x^\top M x
    & = 
    \sum_{j=1}^J g_j x_j^2 - \left( \sum_{j=1}^J g_j x_j \right)^2
    = 
    \sum_{j=1}^J h_j  x_j^2 + g_{\min} \sum_{j=1}^J x_j^2 - \left( \sum_{j=1}^J h_j x_j + g_{\min} \sum_{j=1}^J x_j \right)^2\\
    & = \sum_{j=1}^J h_j  x_j^2 + g_{\min} \sum_{j=1}^J x_j^2 - \left( \sum_{j=1}^J h_j x_j  \right)^2. 
\end{align*}
By the Cauchy--Schwarz inequality, 
\begin{align*}
    \left( \sum_{j=1}^J h_j x_j  \right)^2
    & \le \sum_{j=1}^J h_j \cdot \sum_{j=1}^J h_j x_j^2 
    = \left( \sum_{j=1}^J g_j - J g_{\min} \right)  \cdot \sum_{j=1}^J h_j x_j^2
    = \left( 1 - J g_{\min} \right)  \cdot \sum_{j=1}^J h_j x_j^2\\
    & \le \sum_{j=1}^J h_j x_j^2. 
\end{align*}
Thus, we must have, for any $x\bot 1_J$, 
\begin{align*}
    x^\top M x
    & = \sum_{j=1}^J h_j  x_j^2 + g_{\min} \sum_{j=1}^J x_j^2 - \left( \sum_{j=1}^J h_j x_j  \right)^2
    \ge g_{\min} \sum_{j=1}^J x_j^2 = g_{\min}  x^\top x. 
\end{align*}
This immediately implies that $\lambda_2(M) \ge g_{\min}$. 
Therefore, Lemma \ref{lemma:eigen_M} holds. 
\end{proof}

\begin{lemma}\label{lemma:sv_deri_g}
    Let $\beta = (\beta_2^\top, \ldots, \beta_J^\top)^\top \in \mathbb{R}^{(J-1)p}$. Define, for any $c\in \mathbb{R}^p$, 
    \begin{align*}
    g(\beta,  c) 
    = 
    \left\{1 + \sum_{j=2}^J \exp( c^\top \beta_j ) \right\}^{-1}
    \begin{pmatrix}
    1\\
    \exp( c^\top \beta_2 )\\
    \vdots\\
    \exp( c^\top \beta_J )
    \end{pmatrix}. 
\end{align*}
Let $M\equiv \text{diag}(g) - g g^\top$.
We then have, for any nonzero $c$,  
\begin{align*}
    \frac{\lambda_2(M) }{\sqrt{J}}\|c\| \le \sigma_{\min}^+ \left( \frac{\partial g}{\partial \beta} \right) 
    \le \sigma_{\max} \left( \frac{\partial g}{\partial \beta} \right) 
    \le \|c\| \lambda_{\max} (M).
\end{align*}
\end{lemma}
\begin{proof}
Consider any given $\beta$ and $c\ne 0$. 
By some algebra, we can show that
    \begin{align}\label{eq:deri_g}
    \frac{\partial g}{\partial \beta} 
    = 
    M 
    \cdot 
    \begin{pmatrix}
        e_2 c^\top & e_3 c^\top & \ldots & e_J c^\top 
    \end{pmatrix}. 
\end{align}

We first give an upper bound of $\sigma_{\max} ( {\partial g}/{\partial \beta})$. 
Consider any $x = (x_2^\top, x_3^\top, \ldots, x_J^\top)^\top \in \mathbb{R}^{(J-1)p}$, where $x_j \in \mathbb{R}^p$ for $2\le j \le J$. 
We have
\begin{align*}
    \frac{\partial g}{\partial \beta} x
    = 
    M 
    \cdot 
    \begin{pmatrix}
        e_2 c^\top & e_3 c^\top & \ldots & e_J c^\top 
    \end{pmatrix}
    \cdot
    \begin{pmatrix}
        x_2 \\
        \vdots\\
        x_J
    \end{pmatrix}
    = 
    M 
    \cdot
    \sum_{j=2}^J e_j c^\top x_j
    = 
    M \cdot
    \begin{pmatrix}
        0 \\
        c^\top x_2\\
        \vdots \\
         c^\top x_J
    \end{pmatrix}
    = 
    M \tilde{x},
\end{align*}
where $\tilde{x} = (0, c^\top x_2, \ldots, c^\top x_J)^\top$. 
Note that 
\begin{align*}
    \| \tilde{x} \|^2
    & = \sum_{j=2}^J (c^\top x_j)^2 \le 
    \sum_{j=2}^J \|c\|^2 \|x_j\|^2 =  \|c\|^2 \sum_{j=2}^J  \|x_j\|^2
    = \|c\|^2 \|x\|^2. 
\end{align*}
We then have 
\begin{align*}
    \left\| \frac{\partial g}{\partial \beta} x \right\|
    & \le 
    \|M\| \| \tilde{x} \| = \|c\|\|M\| \| x \| =  \|c\| \lambda_{\max}(M) \| x \|.
\end{align*}
This immediately implies that 
$
\sigma_{\max} ( \partial g/\partial \beta ) \le  \|c\| \lambda_{\max} (M). 
$

We then consider a lower bound on $\sigma_{\min}^{+} ( {\partial g}/{\partial \beta})$. 
Consider any $x = (x_2^\top, x_3^\top, \ldots, x_J^\top)^\top \in \mathbb{R}^{(J-1)p}$, where $x_j \in \mathbb{R}^p$ for $2\le j \le J$, such that $x\bot \Ker({\partial g}/{\partial \beta})$. 
From \eqref{eq:deri_g}, we have 
\begin{align*}
    \Ker({\partial g}/{\partial \beta})
    & \supset 
    \Ker\left( 
    \begin{pmatrix}
        e_2 c^\top & e_3 c^\top & \ldots & e_J c^\top 
    \end{pmatrix}
    \right)\\
    & 
    = 
    \left\{
    y = (y_2^\top, y_3^\top, \ldots, y_J^\top)^\top: 
    c^\top y_2 = \ldots = c^\top y_J = 0, 
    \text{$y_j \in \mathbb{R}^p$ for $2\le j \le J$}
    \right\}. 
\end{align*}
Thus, $x$ must have the form $x^\top = (a_2 c^\top, \ldots, a_J c^\top)$ for some $a \equiv (a_2, \ldots, a_J) \in \mathbb{R}^{J-1}$ and $a \ne 0$. 
We then have 
\begin{align*}
    \frac{\partial g}{\partial \beta}  x
    & = M 
    \cdot 
    \sum_{j=2}^J e_j a_j \|c\|^2
    = \|c\|^2 
    \cdot M 
    \cdot 
    \begin{pmatrix}
    0 \\
    a_2 \\
    \vdots\\
    a_J
    \end{pmatrix}
    = \|c\|^2 
    \cdot M \tilde{a},
\end{align*}
where $\tilde{a} \equiv (0, a_2, \ldots, a_J)^\top$. 
Note that $1_{J}$ is in $\Ker(M)$. 
Let $b = \tilde{a} - J^{-1} 1_{J} 1_J^\top \tilde{a}$ be the project of $\tilde{a}$ onto the space orthogonal to $1_{J}$. We then have 
\begin{align*}
    \| \frac{\partial g}{\partial \beta}  x \|
    & = \|c\|^2 
    \cdot \| M\tilde{a} \|
    = \|c\|^2 
    \cdot \| M b \| \ge 
    \|c\|^2 \lambda_2(M) \|b\|, 
\end{align*}
where $\lambda_2(M)$ denotes the second smallest eigenvalue of $M$. 
Note that 
\begin{align*}
    \|b\|^2 = \|\tilde{a}\|^2 - J^{-1} ( 1_J^\top \tilde{a})^2 
    = \|a\|^2 - J^{-1} (1_{J-1}^\top a)^2
    \ge 
    \|a\|^2 -  J^{-1} (J-1)\|a\|^2
    = \frac{1}{J} \|a\|^2,
\end{align*}
where the inequality follows from Cauchy--Schwarz inequality. 
Consequently, 
\begin{align*}
    \| \frac{\partial g}{\partial \beta}  x \|
    & \ge 
    \|c\|^2 \lambda_2(M) \|b\| 
    \ge  \frac{1}{\sqrt{J}}\|c\|^2 \lambda_2(M) \|a\|. 
\end{align*}
Because, by definition, 
$
    \|x\|^2 = a_2^2  \|c\|^2 + \ldots + a_J^2 \|c\|^2 = \| a\|^2 \|c\|^2,
$
we then have 
\begin{align*}
    \| \frac{\partial g}{\partial \beta}  x \|
    & \ge 
    \frac{1}{\sqrt{J}}\|c\| \lambda_2(M) \|x\|. 
\end{align*}
Because $x$ can be any vector orthogonal to $\Ker({\partial g}/{\partial \beta})$, we must have 
$
    \sigma_{\min}^+ ( {\partial g}/{\partial \beta} )  \ge  \lambda_2(M) \|c\|/\sqrt{J}. 
$

From the above, Lemma \ref{lemma:sv_deri_g} holds. 
\end{proof}

\begin{lemma}\label{lemma:g_hess}
    Define $g$ in the same way as in Lemma \ref{lemma:sv_deri_g}. 
    For any $\beta = (\beta_2^\top, \ldots, \beta_J^\top)^\top \in \mathbb{R}^{(J-1)p}$, $c\in \mathbb{R}^p$ and $1\le i \le J$, we have 
    \begin{align*}
        \left\|\frac{\partial^2 g_i}{\partial \beta \partial \beta} \right\| \le 2(J-1) \|c\|^2. 
    \end{align*}
\end{lemma}
\begin{proof}[\bf Proof of Lemma \ref{lemma:g_hess}]
By some algebra, we can show that
\begin{align*}
    \frac{\partial^2 g_i}{\partial \beta_k \partial \beta_j}
    & = 
    \left\{ g_i(\delta_{ik} - g_k)(\delta_{ij} - g_j)
    - g_i g_j (\delta_{jk} - g_k) \right\} cc^\top, 
    \quad 
    (1\le i \le J, 2\le j,k \le J). 
\end{align*}
Because $g_1, \ldots, g_J$ are all between $0$ and $1$, we have, for any $1\le i \le J$ and $2\le j,k \le J$,
\begin{align*}
    \left\|\frac{\partial^2 g_i}{\partial \beta_k \partial \beta_j} \right\|_{F}
    \le 
    \left| g_i(\delta_{ik} - g_k)(\delta_{ij} - g_j)
    - g_i g_j (\delta_{jk} - g_k) \right| \cdot \| cc^\top \|_F
    \le 2  \| cc^\top \|_F \le 2 \|c\|_F^2 = 2\|c\|^2. 
\end{align*}
This then implies that 
\begin{align*}
    \left\|\frac{\partial^2 g_i}{\partial \beta \partial \beta} \right\|^2
    \le 
    \left\|\frac{\partial^2 g_i}{\partial \beta \partial \beta} \right\|_{F}^2 
    = 
    \sum_{j=2}^J \sum_{k=2}^J \left\|\frac{\partial^2 g_i}{\partial \beta_k \partial \beta_j} \right\|_{F}^2 
    \le (J-1)^2 \cdot 4 \|c\|^4.
\end{align*}
We can then derive Lemma \ref{lemma:g_hess}.   
\end{proof}

\begin{proof}[\bf Proof of Proposition \ref{prop:multinom_cond}]
Condition (i) in Theorem \ref{thm:uniform_CI_model_pi} holds obviously. 
Recall that $\tilde{c} = (1,c^\top)^\top$. 
From Lemma \ref{lemma:g_hess}, for $1\le i \le 4$, 
\begin{align*}
    \left\|\frac{\partial^2 g_i}{\partial \beta \partial \beta} \right\| \le 2\cdot (4-1) \|\tilde{c}\|^2 =  6 (1 + \|c\|^2) 
\end{align*}
and is thus uniformly bounded over all $c\in \mathcal{C}$ and $\beta\in \mathbb{R}^m$. 
This implies condition (iii) in Theorem \ref{thm:uniform_CI_model_pi}. 
Below we prove condition (ii) in Theorem \ref{thm:uniform_CI_model_pi}.  

Let $g_i(\beta, c)$ denote the $i$th component of $g(\beta, c)$, for $1\le i \le 4$. 
From Lemmas \ref{lemma:eigen_M} and \ref{lemma:sv_deri_g}, the nonzero singular values of $D_c$ are bounded from below by 
\begin{align*}
    \sigma_{\min}^+ \left( D_c \right) 
    \ge 
    \frac{\lambda_2(M(\beta_0, c)) }{\sqrt{4}}\|\tilde{c}\|
    \ge 
    \frac{\lambda_2(M(\beta_0, c)) }{2}
    \ge 
    \frac{1}{2} \min_i g_i(\beta_0, c), 
\end{align*}
where the second last inequality holds because $\tilde{c}$ includes an intercept term, 
and they are bounded from above by 
\begin{align*}
    \sigma_{\max} \left( D_c \right) 
    \le 
    \|c\| \lambda_{\max} (M(\beta_0, c))
    \le 
    \|c\| \max_{1\le i \le 4} g_i(\beta_0, c ),
\end{align*}
where $M(\beta_0, c)\equiv \text{diag}(g(\beta_0, c)) - g(\beta_0, c) g(\beta_0, c)^\top$.
Because $c$ is uniformly bounded over $\mathcal{C}$, 
we can know that, for $1\le i \le 4$, $g_i(\beta_0, c )$ is also uniformly bounded over $\mathcal{C}$. 
These then imply condition (ii) in Theorem \ref{thm:uniform_CI_model_pi}. 

From the above, Proposition \ref{prop:multinom_cond} holds.
\end{proof}

}

\end{document}